\documentclass[manuscript, screen, nonacm]{acmart}

\usepackage{tikz}
\usetikzlibrary{decorations.pathreplacing}
\usepackage{stmaryrd}

\theoremstyle{acmplain}
\newtheorem{remark}{Remark}

\mathchardef\mhyphen="2D

\newcommand{\N}{\mathbb{N}}

\newcommand{\B}{\{0,1\}}

\newcommand{\Plays}{\mathsf{Plays}}
\newcommand{\Hist}{\mathsf{Hist}}
\newcommand{\occ}[1]{\mathsf{Occ}({#1})}
\newcommand{\infOcc}[1]{\mathsf{Inf}({#1})}

\newcommand{\payoff}[1]{\mathsf{pay}({#1})}
\newcommand{\won}[1]{\mathsf{won}({#1})}
\newcommand{\payoffObj}[2]{\mathsf{pay}_{#1}({#2})}

\newcommand{\updateReach}[1]{\mathsf{reachUpd}({#1})}
\newcommand{\updateSafe}[1]{\mathsf{safeUpd}({#1})}

\newcommand{\strategyProfile}{\sigma}
\newcommand{\outcome}[1]{\mathsf{out}({#1})}

\newcommand{\Playsigma}[1]{\mathsf{Plays}_{#1}}
\newcommand{\Histsigma}[1]{\mathsf{Hist}_{#1}}
\newcommand{\Playsigmazero}{\Playsigma{\sigma_0}}

\newcommand{\nbrObjectives}{t}
\newcommand{\Obj}{\Omega}
\newcommand{\ObjPlayer}[1]{\Omega_{#1}}
\newcommand{\target}{T} 
\newcommand{\reach}[1]{\mathsf{Reach}(#1)}
\newcommand{\reachName}{\mathsf{Reach}}
\newcommand{\safeSet}{S} 
\newcommand{\safe}[1]{\mathsf{Safe}(#1)}
\newcommand{\safeName}{\mathsf{Safe}}
\newcommand{\parity}[1]{\mathsf{Parity}(#1)}
\newcommand{\parityName}{\mathsf{Parity}}
\newcommand{\BuchiSet}{B}
\newcommand{\Buchi}[1]{\mathsf{B}\ddot{\mathsf{u}}\mathsf{chi}(#1)}
\newcommand{\BuchiName}{\mathsf{B}\ddot{\mathsf{u}}\mathsf{chi}}
\newcommand{\CoBuchi}[1]{\mathsf{co \mhyphen }\mathsf{B}\ddot{\mathsf{u}}\mathsf{chi}(#1)}
\newcommand{\CoBuchiName}{\mathsf{co \mhyphen }\mathsf{B}\ddot{\mathsf{u}}\mathsf{chi}}
\newcommand{\BooleanBuchi}[1]{\mathsf{Boolean}\mathsf{B}\ddot{\mathsf{u}}\mathsf{chi}(#1)}
\newcommand{\BooleanBuchiName}{\mathsf{Boolean}\mathsf{B}\ddot{\mathsf{u}}\mathsf{chi}}

\newcommand{\Muller}[1]{\mathsf{Muller}(#1)}
\newcommand{\MullerName}{\mathsf{Muller}}
\newcommand{\Street}[1]{\mathsf{Streett}(#1)}
\newcommand{\StreetName}{\mathsf{Streett}}
\newcommand{\Rabin}[1]{\mathsf{Rabin}(#1)}
\newcommand{\RabinName}{\mathsf{Rabin}}

\newcommand{\np}{$\mathsf{NP}$}
\newcommand{\npHard}{$\mathsf{NP}$-hard}
\newcommand{\npComplete}{$\mathsf{NP}$-complete}

\newcommand{\pspaceHard}{$\mathsf{PSPACE}$-hard}
\newcommand{\pspaceComplete}{$\mathsf{PSPACE}$-complete}
\newcommand{\nexptime}{$\mathsf{NEXPTIME}$}
\newcommand{\nexptimeHard}{$\mathsf{NEXPTIME}$-hard}
\newcommand{\nexptimeComplete}{$\mathsf{NEXPTIME}$-complete}

\newcommand{\FPT}{$\mathsf{FPT}$}

\newcommand{\problemParam}{k}
\newcommand{\problem}{Stackelberg-Pareto Synthesis problem}
\newcommand{\problemAb}{SPS problem}
\newcommand{\setCover}{Set Cover problem}
\newcommand{\setCoverAb}{SC problem}
\newcommand{\succinctSetCover}{Succinct Set Cover problem}
\newcommand{\succinctSetCoverAb}{SSC problem}
\newcommand{\dominatingSet}{Succinct Dominating Set problem}
\newcommand{\dominatingSetAb}{SDS problem}

\newcommand{\challengerProver}{Challenger-Prover}
\newcommand{\challengerProverAb}{C-P}
\newcommand{\prov}{\mathcal{P}}
\newcommand{\chal}{\mathcal{C}}
\newcommand{\provWit}{W}

\newcommand{\fixed}{$\sigma_0$-fixed}
\newcommand{\paretoOptimal}{\fixed{} Pareto-optimal}
\newcommand{\fixedStrategy}[1]{{#1}-fixed}
\newcommand{\paretoOptimalStrategy}[1]{\fixedStrategy{#1} Pareto-optimal}
\newcommand{\game}{Stackelberg-Pareto game}
\newcommand{\games}{Stackelberg-Pareto games}
\newcommand{\gameAb}{SP game}
\newcommand{\gamesAb}{SP games}
\newcommand{\paretoSet}[1]{P_{#1}}
\newcommand{\paretoSetSize}{|\Wit{\sigma_0}|}
\newcommand{\dominatedPlays}[1]{\Omega^{<}(\paretoSet{#1})}
\newcommand{\deviations}[1]{\mathsf{Dev}(\Wit{#1})}

\newcommand{\punStrat}[1]{\sigma^{\mathsf{Pun}}_{#1}}

\newcommand{\Wit}[1]{\mathsf{Wit}_{#1}}

\newcommand{\prefStrat}[1]{\mathsf{Wit}_{\sigma_0}(#1)}
\newcommand{\W}{W}

\newcommand{\region}{region}

\newcommand{\sect}{section}

\newcommand{\internal}{internal}
\newcommand{\final}{terminal}
\newcommand{\Reg}[1]{\mathsf{Reg}({#1})}

\newcommand{\Witproperty}{region-tree structure}

\AtBeginDocument{%
  \providecommand\BibTeX{{%
    \normalfont B\kern-0.5em{\scshape i\kern-0.25em b}\kern-0.8em\TeX}}}

\setcopyright{acmcopyright}
\copyrightyear{2022}
\acmYear{2022}
\acmDOI{XXXXXXX.XXXXXXX}

\acmJournal{JACM}
\acmVolume{37}
\acmNumber{4}
\acmArticle{111}
\acmMonth{8}

\begin{document}

\title{Stackelberg-Pareto Synthesis (Extended Version)}

\author{Véronique Bruyère}
\orcid{0000-0002-9680-9140}
\affiliation{%
  \institution{Université de Mons (UMONS)}
  \streetaddress{place du Parc 20}
  \city{Mons}
  \country{Belgium}
  \postcode{7000}
}

\author{Baptiste Fievet}
\orcid{0000-0002-4925-1105}
\affiliation{%
  \institution{École Normale Supérieure Paris-Saclay}
  \streetaddress{Avenue des Sciences 4}
  \city{Gif-sur-Yvette}
  \country{France}
  \postcode{91190}
}

\author{Jean-François Raskin}
\orcid{0000-0002-3673-1097}
\affiliation{%
  \institution{Université Libre de Bruxelles (ULB)}
  \streetaddress{Campus de la Plaine CP212}
  \city{Bruxelles}
  \country{Belgium}
  \postcode{1050}
}

\author{Clément Tamines}
\orcid{0000-0003-1104-911X}
\affiliation{%
  \institution{Université de Mons (UMONS)}
  \streetaddress{place du Parc 20}
  \city{Mons}
  \country{Belgium}
  \postcode{7000}
}



\begin{abstract}
We study the framework of two-player Stackelberg games played on graphs in which Player~$0$ announces a strategy and Player~$1$ responds rationally with a strategy that is an optimal response. While it is usually assumed that Player~$1$ has a single objective, we consider here the new setting where he has several. In this context, after responding with his strategy, Player~$1$ gets a payoff in the form of a vector of Booleans corresponding to his satisfied objectives. Rationality of Player~$1$ is encoded by the fact that his response must produce a Pareto-optimal payoff given the strategy of Player~$0$. We study for several kinds of $\omega$-regular objectives the \problem{} which asks whether Player~$0$ can announce a strategy which satisfies his objective, whatever the rational response of Player~$1$. We show that this problem is fixed-parameter tractable for games in which objectives are all reachability, safety, B\"uchi, co-B\"uchi, Boolean B\"uchi, parity, Muller, Streett or Rabin objectives. We also show that this problem is \nexptimeComplete{} except for the cases of B\"uchi objectives for which it is \npComplete{} and co-B\"uchi objectives for which it is in \nexptime{} and \npHard{}. The problem is already \npComplete{} in the simple case of reachability objectives and graphs that are trees. 
\end{abstract}



\maketitle

\section{Introduction}
Two-player zero-sum infinite-duration games played on graphs are a mathematical model used to formalize several important problems in computer science, such as \emph{reactive system synthesis}.  In this context, see e.g.~\cite{PnueliR89}, the graph represents the possible interactions between the system and the environment in which it operates. One player models the system to synthesize, and the other player models the (uncontrollable) environment. In this classical setting, the objectives of the two players are opposite, that is, the environment is \emph{adversarial}. Modelling the environment as fully adversarial is usually a \emph{bold abstraction} of reality as it can be composed of one or several components, each of them having their own objective. 

In this paper, we consider the framework of \emph{Stackelberg games}~\cite{Stackelberg37}, a richer non-zero-sum setting, in which Player~$0$ (the system) called \emph{leader} announces his strategy and then Player~$1$ (the environment) called \emph{follower} plays rationally by using a strategy that is an optimal response to the leader's strategy. This framework captures the fact that in practical applications, a strategy for interacting with the environment is committed before the interaction actually happens. The goal of the leader is to announce a strategy that guarantees him a payoff at least equal to some given threshold. In the specific case of Boolean objectives, the leader wants to see his objective being satisfied. The concept of leader and follower is also present in the framework of \emph{rational synthesis}~\cite{FismanKL10,KupfermanPV16} with the difference that this framework considers several followers, each of them with their own Boolean objective. In that case, rationality of the followers is modeled by assuming that the environment settles to an equilibrium (e.g. a Nash equilibrium) where each component (composing the environment) is considered to be an \emph{independent selfish individual}, excluding cooperation scenarios between components or the possibility of coordinated rational multiple deviations. Our work proposes a novel and natural \emph{alternative} in which the single follower, modeling the environment, has several objectives that he wants to satisfy. After responding to the leader with his own strategy, Player~$1$ receives a vector of Booleans which is his payoff in the corresponding outcome. Rationality of Player~$1$ is encoded by the fact that he only responds in such a way to receive \emph{Pareto-optimal payoffs}, given the strategy announced by the leader. This setting  encompasses scenarios where, for instance, several components can collaborate and agree on trade-offs. The goal of the leader is therefore to announce a strategy that guarantees him to satisfy his own objective, whatever the response of the follower which ensures him a Pareto-optimal payoff. The problem of deciding whether the leader has such a strategy is called the \emph{Stackelberg-Pareto Synthesis problem} (\problemAb{}). 

\paragraph{Contributions.} In addition to the definition of the new setting, our main contributions are the following ones. We consider reachability and safety objectives as well as several classical prefix-independent $\omega$-regular objectives (B\"uchi, co-B\"uchi, Boolean B\"uchi, parity, Muller, Streett, and Rabin).
We provide a thorough analysis of the complexity of solving the \problemAb{} for all those objectives. Our results are interesting and singular both from a theoretical and practical point of view. This paper is the extended version of~\cite{DBLP:conf/concur/BruyereRT21} where only reachability and parity objectives were studied.

First, we show in Theorem~\ref{thm:FPT_all} that the \problemAb{} is \emph{fixed-parameter tractable} (\FPT). The parameters of the \FPT{} complexity for each kind of objectives are summarized in Table~\ref{table:fpt-summary} and the number of objectives of the follower is a parameter in all cases. These are important results as it is expected that, in practice, the \emph{number} of objectives of the environment is limited to a few.

To obtain these \FPT{} results, we develop a reduction from our non-zero-sum games to a zero-sum game in which the protagonist, called \emph{Prover}, tries to show the existence of a solution to the problem, while the antagonist, called \emph{Challenger}, tries to disprove it. This zero-sum game is defined in a \emph{generic way}, independently of the actual objectives used in the initial game, and can then be easily adapted according to the case of a specific objective. The fixed-parameter complexity of the problem for prefix-independent objectives is shown by reduction to Boolean B\"uchi objectives. As a separate result, we propose an \FPT{} algorithm for Boolean B\"uchi objectives in the \emph{zero-sum} setting, which improves the complexity of the algorithm given in~\cite{BruyereHR18} (Theorem~\ref{thm:bbnewcomp}). 

Second, we prove that the \problemAb{} is \nexptimeComplete{} for all the objectives we study, except for B\"uchi objectives for which it is \npComplete{} and co-B\"uchi objectives for which it is in \nexptime{} and \npHard{} (Theorems~\ref{thm:nexptime},~\ref{thm:buchi_np_membership} and \ref{thm:nexptimehard}, see also Table~\ref{table:comp_summary}). It is already \npComplete{} in the simple setting of reachability objectives and graphs that are trees (Theorem~\ref{thm:npcomplete}). 
To the best of our knowledge, this is the first \nexptime-completeness result for a natural class of games played on graphs. To obtain the \nexptime{}-hardness, we present a natural \emph{succinct version} of the set cover problem that is complete for this class (Theorem~\ref{thm:ssc-completeness}), a result of potential independent interest. To obtain the \nexptime{}-membership of the \problemAb{}, we show that exponential-size solutions exist for positive instances of the \problemAb{} and this allows us to design a nondeterministic exponential-time algorithm. Unfortunately, it was not possible to use the \FPT{} algorithm mentioned above to show this membership due to its too high time complexity; conversely, our \nexptime{} algorithm is not \FPT{}.

\begin{table}
  \caption{Complexity class of the \problemAb{} for different objectives.}

  \label{table:comp_summary}
  \begin{tabular}{ll}
    \toprule
    Objective          & Complexity class         \\ \midrule
    Reachability (tree arena), B\"uchi & \npComplete{}            \\ 
    co-B\"uchi                    & \nexptime{}, \npHard{}      \\
    Reachability, safety, Boolean B\"uchi, parity, Muller, Streett, Rabin              & \nexptimeComplete{}      \\ 
   \bottomrule
\end{tabular}
\end{table}

\paragraph{Related Work.} Rational synthesis is introduced in~\cite{FismanKL10} for $\omega$-regular objectives in a setting where the followers are cooperative with the leader, and later in~\cite{KupfermanPV16} where they are adversarial. Precise complexity results for various $\omega$-regular objectives are established in~\cite{ConduracheFGR16} for both settings. Those complexities differ from the ones of the problem studied in this paper. Indeed, for reachability objectives, adversarial rational synthesis is \pspaceComplete, while for parity objectives, its precise complexity is not settled (the problem is \pspaceHard{} and in \nexptime{}). Extension to non-Boolean payoffs, like mean-payoff or discounted sum, is studied in~\cite{GuptaS14,GuptaS14c} in the cooperative setting and in~\cite{BalachanderGR20,FiliotGR20} in the adversarial setting. 

When several players (like the followers) play with the aim to satisfy their objectives, several solution concepts exist such as Nash equilibrium~\cite{Nas50}, subgame perfect equilibrium~\cite{selten}, secure equilibria~\cite{DBLP:conf/tacas/ChatterjeeH07,DBLP:journals/tcs/ChatterjeeHJ06}, or admissibility~\cite{Berwanger07,BrenguierRS15}. The constrained existence problem, close to the cooperative rational synthesis problem, is to decide whether there exists a solution concept such that the payoff obtained by each player is larger than some threshold. Let us mention~\cite{ConduracheFGR16,Ummels08,UmmelsW11} for results on the constrained existence for Nash equilibria and~\cite{Raskin2021,BrihayeBGRB20,Ummels06} for such results for subgame perfect equilibria. Rational verification is studied in~\cite{GutierrezNPW19,GutierrezNPW20}. This problem (which is not a synthesis problem) is to decide whether a given LTL formula is satisfied by the outcome of all Nash equilibria (resp. some Nash equilibrium). The interested reader can find more pointers to works on non-zero-sum games for reactive synthesis in~\cite{DBLP:conf/lata/BrenguierCHPRRS16,Bruyere17,DBLP:journals/siglog/Bruyere21}.

\paragraph{Structure.} The paper is structured as follows. In Section~\ref{sec:prelim}, we introduce the class of \games{} and the \problemAb{}. We establish in Section~\ref{sec:link} the relationship between the objectives we consider and useful technical results for our proofs. This section also provides an improved \FPT{} algorithm for zero-sum Boolean B\"uchi games. We show in Section~\ref{sec:newfpt} that the \problemAb{} is in \FPT{} for all the objectives we study. The \nexptime{} membership of this problem is studied in Section~\ref{sec:nexptime_member} for all the objectives we consider except for B\"uchi objectives for which the \np{} membership is established in Section~\ref{sec:buchi_np}. The hardness of the \problemAb{} is studied in Section~\ref{sec:nexptime_hard}, where it is proved to be \nexptime{}-hard for all the objectives except for B\"uchi and co-B\"uchi objectives for which it is \np{}-hard. We also show in this section that the problem is \npComplete{} in case of reachability objectives and graphs that are trees. In Section \ref{sec:conclusion}, we provide a conclusion and discuss future work.

\section{Preliminaries and Stackelberg-Pareto Synthesis Problem} \label{sec:prelim}

This section introduces the class of two-player \games{} in which the first player has a single objective and the second has several. We present a decision problem on those games called the \problem{}, which we study in this paper.

\subsection{Preliminaries}

\paragraph{Game Arena.}
A \emph{game arena} is a tuple \sloppy $G = (V, V_0, V_1, E, v_0)$ where $(V,E)$ is a finite directed graph such that: \emph{(i)} $V$ is the set of vertices and $(V_0, V_1)$ forms a partition of $V$ where $V_0$ (resp.\ $V_1$) is the set of vertices controlled by Player~$0$ (resp.\ Player~$1$), \emph{(ii)} $E \subseteq V \times V$ is the set of edges such that each vertex $v$ has at least one successor $v'$, i.e., $(v,v') \in E$, and \emph{(iii)} $v_0 \in V$ is the initial vertex. We call a game arena a \emph{tree arena} if it is a tree in which every leaf vertex has itself as its only successor. A \emph{sub-arena} $G'$ with a set $V' \subseteq V$ of vertices and initial vertex $v'_0 \in V'$ is a game arena defined from $G$ as expected.

\paragraph{Plays.}
A \emph{play} in a game arena $G$ is an infinite sequence of vertices $\rho = v_0 v_1 \ldots \in V^{\omega}$ such that it starts with the initial vertex $v_0$ and $(v_j,v_{j+1}) \in E$ for all $j \in \N$. \emph{Histories} in $G$ are finite sequences $h = v_0 \ldots v_j \in V^+$ defined similarly. A history is \emph{elementary} if it contains no cycles. We denote by $\Plays_G$ the set of plays in $G$. We write $\Hist_G$ (resp.\ $\Hist_{G,i}$) the set of histories (resp.\ histories ending with a vertex in $V_i$). We use the notations $\Plays$, $\Hist$, and $\Hist_i$ when $G$ is clear from the context. We write $\occ{\rho}$ the set of vertices occurring in $\rho$ and $\infOcc{\rho}$ the set of vertices occurring infinitely often in $\rho$. 

\paragraph{Strategies.}
A \emph{strategy} $\sigma_i$ for Player~$i$ is a function $\sigma_i\colon \Hist_i \rightarrow V$ assigning to each history $hv \in \Hist_i$ a vertex $v' = \sigma_i(hv)$ such that $(v,v') \in E$. It is \emph{memoryless} if $\sigma_i(hv) = \sigma_i(h'v)$ for all histories $hv, h'v$ ending with the same vertex $v \in V_i$. More generally, it is \emph{finite-memory} if it can be encoded by a Moore machine $\mathcal{M}$~\cite{2001automata}. The \emph{memory size} of $\sigma_i$ is the number of memory states of $\mathcal{M}$. In particular, $\sigma_i$ is memoryless when it has a memory size of one.

Given a strategy $\sigma_i$ of Player~$i$, a play $\rho = v_0 v_1 \ldots$ is \emph{consistent} with $\sigma_i$ if $v_{j+1} = \sigma_i(v_0 \ldots v_j)$ for all $j \in \N$ such that $v_j \in V_i$. Consistency is naturally extended to histories. We denote by $\Playsigma{\sigma_i}$ (resp.\ $\Histsigma{\sigma_i}$) the set of plays (resp.\ histories) consistent with $\sigma_i$. A \emph{strategy profile} is a tuple $\strategyProfile = (\sigma_0, \sigma_1)$ of strategies, one for each player. We write $\outcome{\strategyProfile}$ the unique play consistent with both strategies and we call it the \emph{outcome} of $\strategyProfile$.

\paragraph{Objectives.}
An \emph{objective} for Player~$i$ is a set of plays $\Obj \subseteq \Plays$. A play $\rho$ \emph{satisfies} the objective $\Obj$ if $\rho \in \Obj$. In this paper, we focus on two categories of classical $\omega$-regular objectives. First, we consider objectives that rely on the whole play: reachability and safety objectives.
\begin{itemize}
    \item Let $\target \subseteq V$ be a set of vertices called \emph{target set}, the \emph{reachability} objective $\reach{\target} = {\{\rho \in \Plays \mid \occ{\rho} \cap \target \neq \emptyset \}}$ asks to visit at least one vertex of $\target$. 
    
    \item Let $\safeSet \subseteq V$ be a set of vertices called the \emph{safe set}, the \emph{safety} objective $\safe{\safeSet} = {\{\rho \in \Plays \mid \occ{\rho} \cap (V\setminus\safeSet) = \emptyset \}}$ asks to avoid visiting vertices outside the safe set.
\end{itemize}
Second, we consider those classical objectives which are prefix-independent.
\begin{itemize}
    \item Given a set $\BuchiSet \subseteq V$ of vertices, the \emph{B\"uchi} objective $\Buchi{\BuchiSet} = {\{\rho \in \Plays \mid \infOcc{\rho} \cap \BuchiSet \neq \emptyset \}}$ asks to visit the set $\BuchiSet$ infinitely often. The \emph{co-B\"uchi} objective $\CoBuchi{\BuchiSet} = {\{\rho \in \Plays \mid \infOcc{\rho} \cap \BuchiSet = \emptyset \}}$ asks to visit the set $\BuchiSet$ finitely often.
    
     \item Let $c : V \rightarrow \{0, \dots, d\}$ be a function called a \emph{priority function} which assigns an integer to each vertex in the arena. We write $\infOcc{c(\rho)} = \{c(v) \mid v \in \infOcc{\rho}\}$ the set of priorities occurring infinitely often in $\rho$. The \emph{parity} objective $\parity{c} = \{\rho \in \Plays \mid \min(\infOcc{c(\rho)}) \text{ is even}\}$ asks that the minimum priority visited infinitely often be even. 
     In this paper, we assume that $d$ is even in any priority function $c$ and we also use notation $\max(c)$ to denote this maximal priority $d$.
  
    \item Let $c : V \rightarrow \{0, \dots, d\}$ be a priority function and $Q \subseteq 2^{\{0, \dots, d\}}$ a set containing sets of priorities. The \emph{Muller} objective $\Muller{c, Q} = \{\rho \in \Plays \mid \infOcc{c(\rho)} \in  Q\}$ asks that the set of priorities occurring infinitely often belongs to $Q$. 
    
    \item Given $m$ pairs of sets $(E_1, F_1), \dots, (E_m, F_m)$ such that $E_i \subseteq V, F_i \subseteq V$ with $i \in \{1, \dots, m\}$, the \emph{Streett} objective $\Street{(E_1, F_1), \dots, (E_m, F_m)} = {\{\rho \in \Plays \mid \forall i\in \{1, \dots, m\}, \infOcc{\rho} \cap E_i \neq \emptyset \lor \infOcc{\rho} \cap F_i = \emptyset\}}$ asks that for every pair $(E_i, F_i)$ if $F_i$ is visited infinitely often then $E_i$ is also visited infinitely often. The \emph{Rabin} objective $\Rabin{(E_1, F_1), \dots, (E_m, F_m)} = {\{\rho \in \Plays \mid \exists i\in \{1, \dots, m\}, \infOcc{\rho} \cap E_i = \emptyset \land \infOcc{\rho} \cap F_i \neq \emptyset\}}$ asks that there exists a pair $(E_i, F_i)$ in which $F_i$ is visited infinitely often and $E_i$ is visited finitely often.
\end{itemize}
We also consider the following prefix-independent objective, called \emph{Boolean B\"uchi objective} \cite{DBLP:journals/scp/EmersonL87,BruyereHR18}, which encompasses every prefix-independent objective mentioned above as we will discuss in the next section.
\begin{itemize}
    \item Given $m$ sets $\target_1, \dots, \target_m$ such that $\target_i \subseteq V$, $i \in \{1, \dots, m\}$ and $\phi$ a Boolean formula over the set of variables $X = \{x_1, \dots, x_m\}$, the \emph{Boolean B\"uchi} objective $\BooleanBuchi{\phi, \target_1, \dots, \target_ m} = \{\rho \in \Plays \mid \rho \text{ satisfies } (\phi, \target_1, \dots, \target_ m) \}$ is the set of plays whose valuation of the variables in $X$ satisfy formula $\phi$. Given a play $\rho$, its valuation is such that $x_i = 1$ if and only if $\infOcc{\rho} \cap \target_i \neq  \emptyset$ and $x_i = 0$ otherwise. That is, a play satisfies the objective if the Boolean formula describing sets to be visited infinitely often by a play is satisfied. We denote by $|\phi|$ the size of $\phi$ as equal to the number of symbols in $\{\land, \lor, \lnot\} \cup X $ in $\phi$.
\end{itemize}
In the sequel, when we refer to prefix-independent objectives, we speak about the ones described previously.

\paragraph{Zero-sum games.} A \emph{zero-sum game} $\mathcal{G} = (G, \Obj)$ is a game played by two players such that the first player has objective $\Obj$ and the second player has the opposite objective $\Plays \setminus \Obj$. We assume that the reader is familiar with this concept, see e.g.~\cite{2001automata}. We prefix a zero-sum game by the type of $\Obj$, e.g. \emph{Boolean B\"uchi zero-sum game\footnote{Notice that those games are also called Emerson-Lei zero-sum games.}}.

\subsection{Stackelberg-Pareto Synthesis Problem}

We now introduce a new class of two-player games, called \games{}, in which the first player has a single objective and the second has several. This model is the basis for the problem studied in this paper. We end this section with an example of such a game and of the related problem.

\begin{definition}
\label{def:SPgame}
    A \emph{\game{}} (\gameAb{}) $\mathcal{G} = (G, \ObjPlayer{0},\ObjPlayer{1}, \dots, \ObjPlayer{\nbrObjectives})$ is composed of a game arena $G$, an objective $\ObjPlayer{0}$ for Player~$0$ and $\nbrObjectives \geq 1$ objectives $\ObjPlayer{1}, \dots, \ObjPlayer{\nbrObjectives}$ for Player~$1$. Every objective in an \gameAb{} is of the same type in $\{ \reachName{}, \safeName{}, \BuchiName{}, \CoBuchiName{}, \BooleanBuchiName{}, \parityName{}, \MullerName{}, \StreetName{}, \RabinName{}\}$. We sometimes prefix \gameAb{} by a type of objectives when discussing the specific case where all objectives are of that type, e.g., \emph{parity \gameAb{}} when all objectives are parity objectives. We write $|\mathcal{G}|$ the \emph{size} of $\mathcal{G}$ which corresponds to the number of vertices $|V|$ in its arena and the number $\nbrObjectives$ of objectives for Player~$1$.
\end{definition}

\paragraph{Payoffs in SP Games.}
The \emph{payoff} of a play $\rho \in \Plays$ corresponds to the vector of Booleans $\payoff{\rho} \in \{0,1\}^{\nbrObjectives}$ such that for all $i \in \{1, \dots, \nbrObjectives\}$, $\payoffObj{i}{\rho} = 1$ if $\rho \in \ObjPlayer{i}$, and $\payoffObj{i}{\rho} = 0$ otherwise. Note that we omit to include Player~$0$ when discussing the payoff of a play. Instead we say that a play $\rho$ is \emph{won} by Player~$0$ if $\rho \in \ObjPlayer{0}$ and we write $\won{\rho} = 1$, otherwise it is \emph{lost} by Player~$0$ and we write $\won{\rho} = 0$. We write $(\won{\rho},\payoff{\rho})$ the \emph{extended payoff} of $\rho$. Given a strategy profile $\strategyProfile$, we write $\won{\strategyProfile} = \won{\outcome{\strategyProfile}}$ and $\payoff{\strategyProfile} = \payoff{\outcome{\strategyProfile}}$.  For reachability \gamesAb{} and safety \gamesAb{}, since the objectives are prefix-dependent, we also define $\won{h}$ and $\payoff{h}$ for histories $h \in \Hist$ as done for plays.

We introduce the following partial order on payoffs. Given two payoffs $p = (p_1, \dots, p_\nbrObjectives)$ and $p' = (p'_1, \dots, p'_\nbrObjectives)$ such that $p, p' \in \{0,1\}^{\nbrObjectives}$, we say that $p'$ is \emph{larger} than $p$ and write $p  \leq p'$ if $p_i \leq p'_i$ for all $i \in \{1, \dots, \nbrObjectives\}$. Moreover, when it also  holds that $p_i < p'_i$ for some $i$, we say that $p'$ is \emph{strictly larger} than $p$ and we write $p < p'$. A subset of payoffs $P \subseteq \{0,1\}^\nbrObjectives$ is an \emph{antichain} if it is composed of pairwise incomparable payoffs with respect to $\leq$.

\paragraph{Stackelberg-Pareto Synthesis Problem.} Given a strategy $\sigma_0$ of Player~$0$, we consider the set of payoffs of plays consistent with $\sigma_0$ which are \emph{Pareto-optimal}, i.e., maximal with respect to $\leq$. We write this set $\paretoSet{\sigma_0} =  \max \{\payoff{\rho} \mid \rho \in \Playsigmazero \}$. Notice that it is an antichain. We say that those payoffs are \emph{\paretoOptimal} and write $|\paretoSet{\sigma_0}|$ the number of such payoffs. Notice that $|\paretoSet{\sigma_0}|$ is at most exponential in $\nbrObjectives$.  A play $\rho \in \Playsigmazero$ is called \paretoOptimal{} if its payoff $\payoff{\rho}$ is in $\paretoSet{\sigma_0}$.

The problem studied in this paper asks whether there exists a strategy $\sigma_0$ for Player~$0$ such that every play in $\Playsigmazero$ which is \paretoOptimal{} satisfies the objective of Player~$0$. This corresponds to the assumption that given a strategy of Player~$0$, Player~$1$ will play \emph{rationally}, that is, with a strategy $\sigma_1$ such that  $\outcome{(\sigma_0, \sigma_1)}$ is \paretoOptimal{}. 
It is therefore sound to ask that Player~$0$ wins against such rational strategies.

\begin{definition}
	Given an \gameAb{}, the \emph{\problem{}} (\problemAb{}) is to decide whether there exists a strategy $\sigma_0$ for Player~$0$ (called a \emph{solution}) such that for each strategy profile ${\strategyProfile} = {(\sigma_0, \sigma_1)}$ with $\payoff{\strategyProfile} \in \paretoSet{\sigma_0}$, it holds that $\won{\strategyProfile} = 1$.
\end{definition}

\paragraph{Witnesses.} Given a strategy $\sigma_0$ that is a solution to the \problemAb{} and any payoff $p \in \paretoSet{\sigma_0}$, for each play $\rho$ consistent with $\sigma_0$ such that $\payoff{\rho} = p$ it holds that $\won{\rho}=1$. For each $p \in \paretoSet{\sigma_0}$, we arbitrarily select such a play which we call a \emph{witness} (of $p$). We denote by $\Wit{\sigma_0}$ the set of all witnesses, of which there are as many as payoffs in $\paretoSet{\sigma_0}$. The size $|\Wit{\sigma_0}|$ of $\Wit{\sigma_0}$ is at most exponential in $\nbrObjectives$ as $|\Wit{\sigma_0}| = |\paretoSet{\sigma_0}|$. In the sequel, it is useful to see this set as a tree composed of $\paretoSetSize$ branches. Additionally for a given history $h \in \Hist$, we write $\prefStrat{h}$ the set of witnesses for which $h$ is a prefix, i.e., $\prefStrat{h} = \{ \rho \in \Wit{\sigma_0} \mid h$ is prefix of $\rho \}$. Notice that $\prefStrat{h} = \Wit{\sigma_0}$ when $h = v_0$ and that the size of $\prefStrat{h}$ decreases as the size of $h$ increases, until it contains a single play or becomes empty. 

\begin{example}
\label{ex:example}

\begin{figure}
	\centering
		\resizebox{0.7\textwidth}{!}{%
		\begin{tikzpicture}
		
		\node[draw, rectangle, minimum size=0.75cm, inner sep = 0.5pt] (v0) at (-0.25,0){$v_0$};
		\node[draw, circle, minimum size=0.75cm, inner sep = 0.5pt] (v1) at (1.5,0.75){$v_1$};
		\node[draw, rectangle, minimum size=0.75cm, inner sep = 0.5pt] (x) at (1.5,-0.75){$v_2$};
		\node[draw, circle, minimum size=0.75cm, inner sep = 0.5pt] (v2) at (3.25,0){$v_3$};
		\node[draw, circle, minimum size=0.75cm, inner sep = 0.5pt] (v3) at (3.25,-1.5){$v_4$};
		\node[draw, rectangle, minimum size=0.75cm, inner sep = 0.5pt] (v4) at (4.75,0.75){$v_5$};			
		\node[draw, circle, minimum size=0.75cm, inner sep = 0.5pt] (v5) at (4.75,-0.75){$v_7$};
		\node[draw, circle, minimum size=0.75cm, inner sep = 0.5pt] (v6) at (6.25,0.75){$v_6$};
				
		\draw[-stealth, shorten >=1pt,auto] (v0) to [] (v1);
		\draw[-stealth, shorten >=1pt,auto] (v0) edge [] node {} (x);
		\draw[-stealth, shorten >=1pt,auto] (x) edge [] node {} (v2);
		\draw[-stealth, shorten >=1pt,auto] (x) edge [] node {} (v3);

		\draw[-stealth, shorten >=1pt,auto] (v1) edge [loop right] node[right] {\small $(0, (0,0,1))$} (v1);
		
		\draw[-stealth, shorten >=1pt,auto] (v3) edge [loop right] node[right] {\small $( 0, (1,0,0))$} (v3);

		\draw[-stealth, shorten >=1pt,auto] (v2) edge [] node {} (v4);
		\draw[-stealth, shorten >=1pt,auto] (v4) edge [bend left] node {} (v2);

		\draw[-stealth, shorten >=1pt,auto] (v2) edge [] node {} (v5);
		
		\draw[-stealth, shorten >=1pt,auto] (v5) edge [loop right] node {\small $(1, (1,1,0))$} (v5);
		
		\draw[-stealth, shorten >=1pt,auto] (v4) edge [] node {} (v6);
		
		\draw[-stealth, shorten >=1pt,auto] (v6) edge [loop right] node {\small$(1, (0,1,1))$} (v6);

		\end{tikzpicture}
		}%
	
	\caption{A reachability \gameAb{}.}
	\label{example_memory}
	\Description{Figure 1. Fully described in the text.}
\end{figure}

Consider the reachability \gameAb{} with arena $G$ depicted in Figure \ref{example_memory} in which Player~$1$ has $\nbrObjectives = 3$ objectives. The vertices of Player~$0$ (resp.\ Player~$1$) are depicted as ellipses (resp.\ rectangles)\footnote{This convention is used throughout this paper.}. Every objective in the game is a reachability objective defined as follows: $\ObjPlayer{0} = \reach{\{v_6, v_7\}}$, $\ObjPlayer{1} = \reach{\{v_4, v_7\}}$, $\ObjPlayer{2} = \reach{\{v_3\}}$, $\ObjPlayer{3} = \reach{\{v_1, v_6\}}$. The extended payoff of plays reaching vertices from which they can only loop is displayed in the arena next to those vertices, and the extended payoff of play $v_0 v_2 (v_3v_5)^\omega$ is $(0, (0,1,0))$.

Consider the memoryless strategy $\sigma_0$ of Player~$0$ such that he chooses to always move to $v_5$ from $v_3$. The set of payoffs of plays consistent with $\sigma_0$ is $\{(0,0,1), (0,1,0), (1, 0, 0), (0, 1, 1)\}$ and the set of those that are Pareto-optimal is $\paretoSet{\sigma_0} = \{(1, 0, 0), (0, 1, 1)\}$. Notice that play $\rho = v_0 v_2 (v_4)^\omega$ is consistent with $\sigma_0$, has payoff $(1, 0, 0)$ and is lost by Player~$0$. Strategy $\sigma_0$ is therefore not a solution to the \problemAb{}. In this game, there is only one other memoryless strategy for Player~$0$, where he chooses to always move to $v_7$ from $v_3$. One can verify that it is again not a solution to the \problemAb{}. 

We can however define a finite-memory strategy $\sigma'_0$ such that $\sigma'_0(v_0 v_2 v_3) = v_5$ and $\sigma'_0(v_0 v_2 v_3 v_5 v_3) = v_7$ and show that it is a solution to the problem. Indeed, the set of \paretoOptimalStrategy{$\sigma'_0$} payoffs is $\paretoSet{\sigma'_0} = \{(0, 1, 1),(1, 1, 0)\}$ and Player~$0$ wins every play consistent with $\sigma'_0$ whose payoff is in this set. A set $\Wit{\sigma'_0}$ of witnesses for these payoffs is $\{v_0v_2v_3v_5v_6^\omega, v_0v_2v_3v_5v_3v_7^\omega\}$ and is in this case the unique set of witnesses. This example shows that Player~$0$ sometimes needs memory in order to have a solution to the \problemAb{}.
\qed\end{example}

\section{Useful Properties about the Objectives} 
\label{sec:link}

The results presented in this section are technical results useful in our proofs throughout this paper. We first discuss the relationships between the objectives studied in this paper. We provide a translation from every prefix-independent objective to an equivalent Boolean B\"uchi objective and, for some, to an equivalent parity objective. We use these results to show that an \gameAb{} with a certain type of objectives can be translated in polynomial time into an equivalent \gameAb{} with another type  of objectives. We finish the section by presenting an alternative algorithm to solve Boolean B\"uchi zero-sum games, with an improved complexity over the algorithm from~\cite{BruyereHR18}.

\subsection{Reduction to Boolean B\"uchi and Parity Objectives}
We provide a reduction from the prefix-independent objectives presented in Section \ref{sec:prelim} to an equivalent Boolean B\"uchi objective. The following proposition describes how any objective in $\{\BuchiName{}, \CoBuchiName{}, \parityName{}, \MullerName{}, \StreetName{}, \RabinName{}\}$ can be encoded into a Boolean B\"uchi objective such that a play satisfies the Boolean B\"uchi objective if and only if it satisfies the original objective. This encoding is standard and can be found for instance in~\cite{DBLP:conf/atva/RenkinDP20}. We recall it for the sake of completeness.

\begin{table}
  \caption{Encoding of prefix-independent objectives into Boolean B\"uchi objectives.}
  \label{table:summaryBooleanBuchiTable}
  \begin{tabular}{llll}
  \toprule
                                                        & Parameter         & Size of $X$        & Size of $\phi$ \\ \midrule
  $\Buchi{B}$                             & $B$               & $1$                & $\mathcal{O}(1)$ \\
  $\CoBuchi{B}$                               & $B$               & $1$                & $\mathcal{O}(1)$ \\

  $\parity{c}$                                          & $d$ & $d + 1$ & $\mathcal{O}(d^2)$ \\
  $\Muller{c, Q}$                                       & $d, |Q|$ & $d + 1$       & $\mathcal{O}(d \cdot |Q|)$ \\
    $\Street{(E_1, F_1), \dots, (E_m, F_m)}$   & $m$                    & $2 \cdot  m$  & $\mathcal{O}(m)$  \\
  $\Rabin{(E_1, F_1), \dots, (E_m, F_m)}$   & $m$                    & $2 \cdot  m$  & $\mathcal{O}(m)$  \\ \bottomrule
\end{tabular}
\end{table}

\begin{proposition}
\label{prop_bb_encoding}
    Any objective in $\{\BuchiName{}, \CoBuchiName{}, \parityName{}, \MullerName{}, \StreetName{}, \RabinName{}\}$ can be encoded into an equivalent Boolean B\"uchi objective, a summary of those encodings is provided in Table~\ref{table:summaryBooleanBuchiTable}.

    \begin{enumerate}
        \item A B\"uchi objective $\Buchi{B}$ (resp. co-B\"uchi objective $\CoBuchi{B}$) can be encoded into a Boolean B\"uchi objective $\BooleanBuchi{\phi, T}$ such that $\BooleanBuchi{\phi, T} = \Buchi{B}$ (resp. $= \CoBuchi{B}$) with $X = \{x\}$ and $|\phi| = 1$ in $\mathcal{O}(1)$.
        
        \item A parity objective $\parity{c}$ can be encoded into a Boolean B\"uchi objective  $\BooleanBuchi{\phi, T_0, \dots, T_d}$ with $d = \max(c)$ the maximal even priority according to priority function $c$ such that $\BooleanBuchi{\phi, T_0, \dots, T_d} = \parity{c}$ with $X = \{x_0, x_1, \dots, x_{d}\}$ and $|\phi|$ in $\mathcal{O}(d^2)$.
        
        \item A Muller objective $\Muller{c, Q}$ with $Q = \{Q_1, Q_2, \dots, Q_m\}$ can be encoded into a Boolean B\"uchi objective  $\BooleanBuchi{\phi, T_0, \dots, T_d}$ with $d = \max(c)$ the maximal even priority according to priority function $c$ such that $\BooleanBuchi{\phi, T_0, \dots, T_d} = \Muller{c, Q}$ with $X = \{x_0, x_1, \dots, x_{d}\}$ and $|\phi|$ in $\mathcal{O}(d \cdot m)$ with $m =|Q|$.
        
        \item A Streett objective $\Street{(E_1, F_1), \dots, (E_m, F_m)}$ (resp. Rabin objective $\Rabin{(E_1, F_1), \dots, (E_m, F_m)}$) can be encoded into a Boolean B\"uchi objective $\BooleanBuchi{\phi, T_1, T'_1, \dots, T_m, T'_m}$ such that $\BooleanBuchi{\phi, T_1, T'_1, \dots, T_m, T'_m} = \Street{(E_1, F_1), \dots, (E_m, F_m)}$ (resp. $= \Rabin{(E_1, F_1), \dots, (E_m, F_m)}$) with $X = \{x_1, x'_1, \dots, x_m, x'_m\}$ and $|\phi|$ in $\mathcal{O}(2 \cdot m)$. 
        
    \end{enumerate}
\end{proposition}

\begin{proof} For each objective we provide the formula $\phi$ of the corresponding Boolean B\"uchi objective.
    \begin{enumerate}
        \item Let $\Buchi{B}$ be a B\"uchi objective for the set $B$. We define the Boolean B\"uchi objective $\BooleanBuchi{\phi, T}$ such that $T = B$ and $\phi = x$. 
        In case of a co-B\"uchi objective, we take $T = B$ and $\phi = \neg x$. 
        
        \item Let $\parity{c}$ be a parity objective. We create the formula $\phi$ over the set of variables $\{x_0, x_1, \dots, x_d\}$ such that $\phi = x_0 \lor (x_2 \land \neg x_1) \lor (x_4 \land \neg x_3 \land \neg x_1) \lor \dots \lor (x_d \land \neg x_{d-1} \land \neg x_{d - 3} \land \dots \land \neg x_1)$ and for $j \in \{0, \dots, d\}$, we define the set corresponding to variable $x_j$ as $T_j = \{ v \in V \mid c(v) = j \}$. This Boolean B\"uchi objective explicitly lists the sets of priorities seen infinitely often which satisfy the parity objective.
        
        \item Let $\Muller{c, Q}$ be a Muller objective such that $Q = \{Q_1,\ldots, Q_m\}$. We create the formula $\phi$ over variables $\{x_0, x_1, \dots, x_d\}$ such that $\phi = (z^1_0 \land z^1_1 \land \dots \land z^1_d) \lor \dots \lor (z^m_0 \land z^m_1 \land \dots \land z^m_d)$ with $z^i_j = x_j$ if $j \in Q_i$ and $z^i_j = \neg x_j$ otherwise. The set corresponding to variable $x_j$ is $T_j = \{ v \in V \mid c(v) = j \}$. This Boolean B\"uchi objective indicates for each $Q_i$ what are the priorities which are visited infinitely often and those which are not.
        
        \item Let $\Street{(E_1, F_1), \dots, (E_m, F_m)}$ be a Streett objective. We create the formula $\phi$ over the set of variables $\{x_1, x'_1, \dots, x_m, x'_m\}$ such that set corresponding to variable $x_j$ (resp. $x'_j$) is $T_j = E_j$ (resp. $T'_j = F_j$). We define $\phi = (x_1 \lor \neg x'_1) \land \dots \land (x_m \lor \neg x'_m)$. 
        This Boolean B\"uchi objective is easily adapted in case of a Rabin objective which is the complement of a Streett objective.
    \end{enumerate}
\end{proof}

It follows from Proposition~\ref{prop_bb_encoding} that any \gameAb{} with prefix-independent objectives can be translated into a Boolean B\"uchi \gameAb{} on the same arena with polynomial time complexities depending on the actual objectives. 

We now recall the following relationship between parity objectives and some of the prefix-independent objectives studied in this paper. We start with the following proposition on the classical translation of parity objectives into Streett or Rabin objectives (see e.g. \cite{DBLP:reference/mc/BloemCJ18, DBLP:reference/mc/Kupferman18}).

\begin{proposition}
    \label{prop:parity_into_streett_rabin}
    A parity objective $\parity{c}$ with $d = \max(c)$ can be encoded into an equivalent Rabin (resp. Streett) objective with $d/2 + 1$ pairs.
\end{proposition}

\begin{proof}
Let $d' = d/2$ and let us first provide the encoding into a Rabin objective such that $\parity{c} = \Rabin{(E_1, F_1), \dots, (E_{d' + 1}, F_{d' + 1}})$. We construct the chain $E_1 \subsetneq F_1 \subsetneq \dots \subsetneq E_{d' + 1} \subsetneq F_{d' + 1}$ of $d' + 1$ Rabin pairs as follows: let $E_1 = \emptyset$, $F_1 = \{v \mid c(v) = 0\}$ and for all $2 \leq j \leq d' + 1$, let $E_j = F_{j-1} \cup \{v \mid c(v) = 2 \cdot j -3\}$ and $F_j = E_{j} \cup \{v \mid c(v) = 2 \cdot j - 2\}$. An encoding into a Streett objective is done similarly as the complement of a Rabin objective.
\end{proof}

It follows from Proposition~\ref{prop:parity_into_streett_rabin} that any parity \gameAb{} can be translated into a Streett or Rabin \gameAb{} on the same arena in polynomial time.

A polynomial time translation of a parity \gameAb{} into a Muller \gameAb{} requires to play on a modified arena $G'$. The next proposition indicates that the size of $G'$ is polynomial in the size of $G$ and in the maximal priority $d$ among the priorities used in each parity objective of $\mathcal{G}$, and that the number of sets in any Muller objective of $\mathcal{G'}$ is polynomial in this maximal priority~$d$. 

\begin{proposition}
    \label{prop:parity_to_muller}
    Any parity \gameAb{} $\mathcal{G} = (G, \ObjPlayer{0},\ObjPlayer{1}, \dots, \ObjPlayer{\nbrObjectives})$ with each $\ObjPlayer{i} = \parity{c_i}$ can be transformed into a Muller \gameAb{} $\mathcal{G'} = (G', \ObjPlayer{0}',\ObjPlayer{1}', \dots, \ObjPlayer{\nbrObjectives}')$ such that Player~$0$ has a solution to the \problemAb{} in $\mathcal{G}$ if and only if he has one in $\mathcal{G'}$. In addition, with $d = \max_{i \in \{0, \dots, \nbrObjectives\}}(\max(c_i))$, $|V'|$ is in $\mathcal{O}(|V| + |E| \cdot d)$ and for each objective $\Omega'_i = \Muller{c'_i, Q_i}$, $|Q_i|$ is in $\mathcal{O}(\max(c_i))$.
\end{proposition}

\begin{proof}
    Let $\mathcal{G} = (G, \ObjPlayer{0},\ObjPlayer{1}, \dots, \ObjPlayer{\nbrObjectives})$ be a parity \gameAb{} such that $\ObjPlayer{i} = \parity{c_i}$ for all $i$, and let $d = \max_{i \in \{0, \dots, \nbrObjectives\}}(\max(c_i))$. Let us construct the desired Muller \gameAb{} $\mathcal{G'} = (G', \ObjPlayer{0}',\ObjPlayer{1}', \dots, \ObjPlayer{\nbrObjectives}')$. The main difficulty is to obtain Muller objectives $\ObjPlayer{i}' = (c'_i,Q_i)$ such that $Q_i$ contains a polynomial number of sets.
    
    The vertices of $G'$ consist of those of $G$ with additional vertices for each edge of $E$. Let us consider an edge $(v_i, v_j) \in E$, this edge is replaced in $G'$ by a sequence of vertices $v_i \ v_{i,j}^1 \dots v_{i,j}^d \ v_j$ such that $(v_i, v_{i,j}^1)$, $(v_{i,j}^k, v_{i,j}^{k+1})$ for $k \in \{1, \dots, d - 1\}$ and $(v_{i,j}^d, v_j)$ are in $E'$. Vertices $v_{i,j}^k$ with $k \in \{1, \dots, d\}$ belong to the same player as $v_i$. The priority function $c'$ of the Muller objective $\Omega' = (c',Q)$ corresponding to the parity objective $\Omega$ in $\mathcal{G}$ is defined as follows for this new sequence of vertices. First, the priority of $v_i$ and $v_j$ remain unchanged, that is $c'(v_i) = c(v_i)$ and  $c'(v_j) = c(v_j)$. Second, the priority of vertex $v_{i,j}^k$ is $c'(v_{i,j}^k) = \min(c(v_i) + k, d)$. Hence we observe the sequence of priorities $c(v_i), c(v_i)+1, c(v_i)+2, \ldots, d, \ldots, d, c(v_j)$ along the path replacing the edge $(v_i,v_j)$ (once priority $d$ is reached, subsequent vertices in this path keep that priority).
    The goal of this construction is to make it so that each time a vertex of priority $p$ is encountered in a play, a vertex for every larger priority $p + 1$ to $d$ is also encountered.
    
    This construction only impacts the sets of priorities seen infinitely often in a play but not the smallest of such priorities, allowing us to define each Muller objective $(c',Q)$ using a polynomial (instead of exponential) number of sets in $Q$. Indeed $Q$ is defined such that $Q = \{\{0, 1, \dots, d\}\, \{2, 3, \dots, d\}, \dots, \{\max(c),\dots, d\}\}$ (recall that $\max(c)$ is assumed to be even). Notice that $|Q| = \max(c)/2 + 1$.
    
    Let us show that Player~$0$ has a solution to the \problemAb{} in $\mathcal{G}$ if and only if he has one in $\mathcal{G'}$. It can easily be shown that, with this construction, there exists a play $\rho = v_0 v_1 v_2 \dots \in \Plays_{G}$ if and only if there exists a play $\rho' = v_0 v_{0,1}^1 \dots v_{0,1}^d v_1 v_{1,2}^1 \dots v_{1,2}^d v_2 \dots \in \Plays_{G'}$. Recall that each vertex $v_{i,j}^k \in V'$ belongs to the same player as vertex $v_i \in V$. Let us show that $\payoff{\rho} = \payoff{\rho'}$ and $\won{\rho} = \won{\rho'}$. To do so, we prove that a parity objective $\Omega$ of $\mathcal{G}$ is satisfied in $\rho$ if and only if the corresponding Muller objective $\Omega'$ of $\mathcal{G'}$ is satisfied in $\rho'$. Let $\infOcc{c(\rho)}$ be the set of priorities occurring infinitely often in $\rho$. By construction, it holds that the corresponding set for $\rho'$ is $\infOcc{c(\rho')} = \{ c'(v_i), c'(v_{i,i+1}^1), \dots, c'(v_{i,i+1}^d) \mid v_i \in \infOcc{\rho} \}$. By construction, the minimum priority in $\infOcc{c(\rho')}$ is the same as in $\infOcc{c(\rho)}$ since $c'(v_i) = c(v_i)$ and the priority of vertices $v_{i,i+1}^k$ is larger than that of $v_i$ by construction. Let us assume that $\Omega$ is satisfied in $\rho$, it follows that the set $\infOcc{c(\rho')}$ contains every priority from the minimal priority occurring in $\infOcc{c(\rho)}$, which is even, to priority $d$. Therefore $\infOcc{c(\rho')} \in Q$ and the Muller objective is satisfied in $\rho'$. Let us now assume that $\Omega$ is not satisfied in $\rho$, it follows that the minimum priority occurring in $\infOcc{c(\rho)}$, and therefore in $\infOcc{c(\rho')}$, is odd. Since no set in $Q$ is such that its minimum priority is odd, it follows that the Muller objective is not satisfied in $\rho'$. Using this result, we can show that a strategy $\sigma_0$ that is solution to the \problemAb{} in $\mathcal{G}$ can be transformed into a strategy $\sigma'_0$ which is a solution in $\mathcal{G'}$ and vice-versa.
\end{proof}

\subsection{An Alternative Approach to Solving Boolean B\"uchi Zero-Sum Games}
\label{subsec:AlternativeFPTBooleanBuchi}

We now consider Boolean B\"uchi zero-sum games and provide an alternative algorithm to solving them, improving on the fixed-parameter complexity of the algorithm from \cite{BruyereHR18}. The way we prove this result, rather than the result in itself (which is interesting in its own right), is used in the next section to show the fixed-parameter complexity of solving the \problemAb{} for \gamesAb{} with prefix-independent objectives. We refer the reader to~\cite{downey2012parameterized} for the concept of fixed-parameter complexity (\FPT). 

\begin{theorem} 
\label{thm:bbnewcomp}
Solving Boolean B\"uchi zero-sum games $(G,\BooleanBuchi{\phi, \target_1, \dots, \target_ m})$ is in \FPT{} for parameters $m$ and $|\phi|$. The algorithm complexity in the parameters is exponential in $m$ and linear in $|\phi|$.
\end{theorem}

The \FPT{} algorithm proposed in~\cite{BruyereHR18} is polynomial in $|V|$, linear in the number $\ell$ of symbols $\vee, \wedge$ in $\phi$, and double exponential in $m$. 
Notice that if we assume that $\phi$ contains no occurrence of two consecutive symbols $\neg$, then $|\phi|$ is linear in $\ell$. The complexity of our algorithm thus improves on the complexity of the algorithm from \cite{BruyereHR18} (from double to single exponential in $m$).

We now detail our algorithm, starting with an adequate structure to track the sets of a Boolean B\"uchi objective which are visited infinitely often in a play.

\paragraph{Set Appearance Record.} Latest appearance records (LAR), which are structures that keep track of the most recently occurred vertices in a play, are often used in the literature. Modifications of these structures can be used to keep track of other values, such as the latest priorities occurring in a play \cite{DBLP:conf/atva/RenkinDP20}. We provide a modification on the LAR structure, called \emph{set appearance record} (SAR), which, given a series of sets of vertices, keeps track of those most recently visited. A set is \emph{visited} if a vertex which belongs to that set is visited in the play. Given $m$ sets of vertices $T_1, \dots, T_m$, we write $\mathcal{P}(T_1, \dots, T_m)$ the set of their permutations which we denote by a word of length $m$ over the alphabet $\{T_1, \dots, T_m\}$ such that each letter appears exactly once in that word. We define a deterministic finite automaton $\textsf{SAR}= (\mathcal{P}(T_1, \dots, T_m) \times \{1, \dots, m +1\}, (s_0, {r}_0), \delta)$ over the alphabet $V$ such that $s_0 = T_1, \dots, T_m$, ${r}_0 = 1$ for the initial state $(s_0,r_0)$ and for each state $(s,r)$ and each symbol $v \in V$, we define $\delta((s, {r}), v) = (s',r')$ such that $s' = (T_{i_1}, \dots, T_{i_{j-1}}, T_{i_{j}}, \dots, T_{i_m})$ and $r' = i_j$ where sets $T_{i_1}$ to $T_{i_{j-1}}$ don't contain vertex $v$ and remain in the same order as in $s$ and sets $T_{i_{j}}$ to $T_{i_m}$ contain $v$ and are also kept in the same order as in $s$. If $v$ does not belong to any set in $T_1, \dots, T_m$, $s'$ remains equal to $s$ and we set $r' = m + 1$. Notice that $r'$, called the \emph{hit}, corresponds to the position of the leftmost set in $s$ containing $v$. Let $\rho = v_0 v_1 \dots$ be a play in $G$. We can consider the corresponding \textsf{SAR} execution on $\rho$ which is $\textsf{SAR}(\rho) = (s_0, {r}_0) (s_1, {r}_1) \dots$ such that $(s_j, {r}_j) = \delta((s_{j-1}, r_{j-1}), v_{j-1})$. Let ${r}_{min}$ be the smallest hit appearing infinitely often in $\textsf{SAR}(\rho)$. Then, the sets appearing after index ${r}_{min}$ in $s$ along $\textsf{SAR}(\rho)$ are always the same from some point on and equal to the sets in $T_1, \dots, T_m$ visited infinitely often in $\rho$, that is the sets $T_i$ such that $T_i \cap \infOcc{\rho} \neq \emptyset$. Given a value ${r}$ of the hit, we write $s_{\geq {r}}$ the sets which appear after index ${r}$ in $s$. Notice that, if from some point on the play does not visit any vertex $v$ such that $v \in T_i$ for some set $T_i$, then $r_{min} = m + 1$ and  $s_{\geq {r_{min}}} = \emptyset$. Those arguments are summarized in the next lemma.

\begin{lemma}
\label{lem:correspondance}
Let $\rho$ be a play and $\textsf{SAR}(\rho)$ be its corresponding \textsf{SAR} execution. Let ${r}_{min}$ be the smallest hit appearing infinitely often in $\textsf{SAR}(\rho)$. Then the sets in $T_1, \dots, T_m$ visited infinitely often in $\rho$ are exactly the sets appearing after index ${r}_{min}$ in $s$ along $\textsf{SAR}(\rho)$. 
\qed\end{lemma}

\paragraph{Extending the Arena with the \textsf{SAR}} Let us now consider the game arena $G'$ obtained by extending the arena $G$ of the Boolean B\"uchi zero-sum game $(G, \BooleanBuchi{\phi, T_1, \dots, T_m})$ with the \textsf{SAR} for the sets $T_1, \dots, T_m$. We define  $G' = (V', V'_0, V'_1, E', v'_0)$ with
\begin{itemize}
    \item $V' = V \times \mathcal{P}(T_1, \dots, T_m)  \times \{1, \dots, m + 1 \}$;
    \item $V'_i = V_i \times \mathcal{P}(T_1, \dots, T_m)  \times \{1, \dots, m + 1\}$;
    \item $((v, s, {r}), (v', s', {r}')) \in E'$ if and only if $(v, v') \in E$ and $(s', {r}') = \delta((s,{r}), v)$;
    \item $v'_0 = (v_0, s_0, {r}_0)$.
\end{itemize}
Given a play $\rho'$ (resp. history $h'$) in $G'$, we denote by $\rho'_{V}$ (resp. $h'_V$) the infinite (resp. finite) sequence of vertices in $V$ obtained by keeping the $v$-component of every vertex $(v,s,{r})$ appearing in $\rho'$ (resp. $h'$). By construction of $G'$, this \textit{projection} of $\rho'$ (resp. $h'$) is a play (resp. a history) in $G$. 
Conversely, given a play $\rho = v_0 v_1 \dots$ (resp. history $h = v_0 v_1 \dots v_j$) in $G$, we can construct a unique corresponding play (resp. history) in $G'$ by adding to each vertex $v_i$ in $\rho$ (resp. $h$) its corresponding state of the automaton $\textsf{SAR}$, $(s, {r}) = \hat{\delta}((s_0, {r}_0), v_0 \dots v_{i-1}))$, which is unique given the vertices which have been encountered in the play (resp. history).

\paragraph{Solving the Boolean B\"uchi Zero-sum Game Using $G'$.}
Let us now show how we can solve the Boolean B\"uchi zero-sum game $(G, \BooleanBuchi{\phi, T_1, \dots, T_m})$ by solving a parity zero-sum game with objective $\parity{c}$ in $G'$. Given a vertex $(v, s, {r})$ of $G'$, the formula $\phi$ of the Boolean B\"uchi objective and its set of variables $X$, we write $sat(\phi, s, {r}) = 1$ if and only if the valuation of $X$ such that $x_i = 1$ if and only if $T_i \in s_{\geq {r}}$ and $x_i = 0$ otherwise satisfies $\phi$, else we write $sat(\phi, s, {r}) = 0$. We define the following priority function $c : V' \rightarrow \{0, \ldots, 2 \cdot (m+1) \}$ on the vertices of $G'$: 

\[
    c((v, s, r)) = \begin{cases}
        2 \cdot r - 2 & \text{if } sat(\phi, s, {r}) = 1 \\
        2 \cdot r - 1 & \text{otherwise}
    \end{cases}
\]
Intuitively, a vertex associated with a smaller hit value has a smaller priority. In addition, its priority is even if the recent sets in $s$ (that is, the sets $s_{\geq {r}}$ appearing after the index ${r}$ in $s$) would satisfy formula $\phi$ if they were to be visited infinitely often and its priority is odd if that would not be the case. Let us now show that solving the Boolean B\"uchi zero-sum game $(G, \BooleanBuchi{\phi, T_1, \dots, T_m}))$ amounts to solving the parity zero-sum game $(G', \parity{c})$.

\begin{proposition}
\label{prop:bb_correctness}
Let $\mathcal{G} = (G, \BooleanBuchi{\phi, T_1, \dots, T_m}))$ be a Boolean B\"uchi zero-sum game and $\mathcal{G}' = (G', \parity{c})$ be the corresponding parity zero-sum game on the extended arena $G'$. Player~$0$ has a winning strategy from vertex $v_0$ in $G$ for the objective $\BooleanBuchi{\phi, T_1, \dots, T_m}$ if and only if he has a winning strategy for the objective $\parity{c}$ from $(v_0, s_0, {r}_0)$ in $G'$.
\end{proposition}
\begin{proof}
Let $\sigma_0$ be a winning strategy from $v_0$ in $G$ for the objective $\BooleanBuchi{\phi, T_1, \dots, T_m}$. Let us show that we can create a strategy $\sigma'_0$ which is winning from $(v_0, s_0, {r}_0)$ in $G'$ for the parity objective $\parity{c}$. We define $\sigma'_0$ such that for $h' \in {V'}^+$ ending in a vertex $(v, s, {r}) \in V'_0$, $\sigma'_0(h') = (v', \delta((s,{r}), v))$ such that $\sigma_0(h'_V) = v'$. Let $\rho'$ be a play consistent with $\sigma'_0$ from $(v_0, s_0, {r}_0)$. Let us show that $\rho'$ satisfies the objective $\parity{c}$. First, notice that, given the play $\rho'$ and the way we defined $\sigma'_0$, it holds that $\rho = \rho'_V$ is a play in $G$ starting in $v_0$ and consistent with $\sigma_0$. Therefore this play satisfies the objective $\BooleanBuchi{\phi, T_1, \dots, T_m}$ and the sets of vertices in $T_1, \dots, T_m$ visited infinitely often in $\rho$ satisfy the formula $\phi$. Given the construction of $G'$, the $(s, r)$-component of the vertices in $\rho'$ correspond to the run $\textsf{SAR}(\rho)$ of the $\textsf{SAR}$. Let $r_{min}$ be the smallest value of the hit occurring infinitely often in $\textsf{SAR}(\rho)$, it holds that the sets $s_{\geq {r_{min}}}$ correspond to the sets visited infinitely often in $\rho$ by Lemma~\ref{lem:correspondance}. Since $\rho$ satisfies the Boolean B\"uchi objective, it also holds that eventually $sat(\phi, s, {r_{min}}) = 1$ for every occurrence of $(s,r_{min})$ in $\textsf{SAR}(\rho)$. Hence, in $\rho'$, 
the minimum priority occurring infinitely often in $\rho'$ according to $c$ is even, satisfying $\parity{c}$.

Let $\sigma'_0$ be a winning strategy from $(v_0, s_0, {r}_0)$ in $G'$ for the objective $\parity{c}$. Let us show that we can create a strategy $\sigma_0$ which is winning from $v_0$ in $G$ for the Boolean B\"uchi objective $\BooleanBuchi{\phi, T_1, \dots, T_m}$. We define $\sigma_0$ such that for $hv \in {V}^+$ ending in a vertex $v \in V_0$, $\sigma_0(h) = v'$ such that $\sigma'_0((v, s, r)) = (v', s', r')$ with $(s, {r}) = \hat{\delta}((s_0, {r}_0), h))$. Let $\rho$ be a play consistent with $\sigma_0$, let us show that this play satisfies the objective $\BooleanBuchi{\phi, T_1, \dots, T_m}$. We can construct the unique corresponding play $\rho'$ in $G'$ such that 
the $(s, r)$-component of $\rho'$ corresponds to $\textsf{SAR}(\rho)$. Given the way we defined $\sigma_0$, this play is consistent with $\sigma'_0$. Therefore, it satisfies the objective $\parity{c}$ and the minimal priority occurring infinitely often in $\rho'$ is some even priority $2 \cdot r$ for some value $r$. It follows that in $\textsf{SAR}(\rho)$, the smallest value of the hit occurring infinitely often is $r$ and that eventually for every occurrence of $(s, r)$ in $\textsf{SAR}(\rho)$, $sat(\phi, s, {r}) = 1$. Since by Lemma~\ref{lem:correspondance}, from some point on, $s_{\geq {r}}$ correspond to the sets in $\{T_1, \dots, T_m\}$ visited infinitely often in $\rho$, it follows that $\rho$ satisfies the Boolean B\"uchi objective $\BooleanBuchi{\phi, T_1, \dots, T_m}$.
\end{proof}

We are now able to prove Theorem~\ref{thm:bbnewcomp}. To do so, we use the following result.

\begin{theorem}[\cite{CaludeJKL020}]
\label{thm:calude}
Solving a parity zero-sum game is in \FPT{} in parameter $d$, with an algorithm in $\mathcal{O}(|V|^5 + 2^{d \cdot (\log(d) + 6)})$ time where $d$ is the number of priorities.
\end{theorem}

\begin{proof}[Proof of Theorem~\ref{thm:bbnewcomp}.]
Let $\mathcal{G} = (G, \BooleanBuchi{\phi, T_1, \dots, T_m}))$ be a Boolean B\"uchi zero-sum game. We construct $\mathcal{G}' = (G', \parity{c})$ the corresponding parity zero-sum game on the extended arena $G'$. The arena $G'$ contains $|V'| = |V| \cdot m! \cdot (m+1)$ vertices. Its priority function $c$ is such that $\max(c) = 2 \cdot (m + 1)$ and requires to evaluate $sat(\phi, s, {r})$ for each vertex $(v, s, {r})$ of $G'$. Hence, constructing $\mathcal{G}'$ is in \FPT, with an algorithm polynomial in $|V|$, exponential in $m$, and linear in $|\phi|$. By Theorem~\ref{thm:calude}, solving the parity zero-sum game $\mathcal{G}'$ is in \FPT{} with a time complexity polynomial in $|V|$ and exponential in $m$. The whole algorithm (constructing $\mathcal{G}'$ and solving it) is then in \FPT, with a time complexity polynomial in $|V|$, exponential in $m$, and linear in $|\phi|$.
\end{proof}


\section{Fixed-Parameter Complexity}
\label{sec:newfpt}

This section is dedicated to showing the fixed-parameter complexity of solving the \problemAb{} in \gamesAb{}. We start by introducing a reduction from \gamesAb{} to a zero-sum game, called the \challengerProver{} game. We then show that the \problemAb{} is in \FPT{} for reachability and safety \gamesAb{} using this zero-sum game. Finally, we show that the problem is in \FPT{} for Boolean B\"uchi \gamesAb{} which then allows us to show that it is also in \FPT{} for every prefix-independent objective.

\begin{theorem}
\label{thm:FPT_all}
Solving the \problemAb{} is in \FPT{} for \gamesAb{}. The complexity and parameter for each objective is reported in Table~\ref{table:fpt-summary}.
\end{theorem}

\begin{table}
  \caption{Complexity of the \FPT{} algorithm for each kind of objective studied in this paper. We consider $(\Omega_i)_{i \in \{1,\ldots,t\}}$ where all objectives $\Omega_i$ are $\reach{T}$, $\safe{S}$, $\Buchi{B}$, $\CoBuchi{B}$, $\BooleanBuchi{\phi,T_1,\ldots,T_m}$, $\parity{c}$ with $d = \max(c)$, $\Muller{c,Q}$ with $d = \max(c)$, $\Street{(E_1,F_1),\ldots,(E_m,F_m)}$, or $\Rabin{(E_1,F_1),\ldots,(E_m,F_m)}$. Notice that all algorithms are double exponential in parameter $t$.}
  \label{table:fpt-summary}
  \begin{tabular}{lll}
    \toprule
    Objective                 & Parameters & Complexity in the parameters \\ \midrule
    Reachability              & $t$ & double exponential in $t$            \\ 
    Safety                    & $t$ & double exponential in $t$            \\  
    B\"uchi                   & $t$ & double exponential in $t$ \\  
    co-B\"uchi                & $t$ & double exponential in $t$ \\  
    Boolean B\"uchi           & $t$, $(m_i,|\phi_i|)_{i \in \{1,\ldots,t\}}$ & double exponential in $t$, exponential in $\sum_{i=0}^t m_i$, and polynomial in $\sum_{i=0}^t |\phi_i|$ \\ 
    Parity                    & $t$, $(d_i)_{i \in \{1,\ldots,t\}}$ & double exponential in $t$ and exponential in $\sum_{i=0}^t d_i$\\ 
    Muller                    & $t$, $(d_i,|Q_i|)_{i \in \{1,\ldots,t\}}$ & double exponential in $t$, exponential in $\sum_{i=0}^t d_i$, and polynomial in $\sum_{i=0}^t |Q_i|$           \\ 
    Streett                   &  $t$, $(m_i)_{i \in \{1,\ldots,t\}}$  & double exponential in $t$ and exponential in $\sum_{i=0}^t m_i$ \\ 
    Rabin                     &  $t$, $(m_i)_{i \in \{1,\ldots,t\}}$ & double exponential in $t$ and exponential in $\sum_{i=0}^t m_i$  \\ 
   \bottomrule
\end{tabular}
\end{table}

\subsection{Challenger-Prover Game} 
\label{subsec:cp}

In order to prove Theorem~\ref{thm:FPT_all}, we provide a reduction to a specific two-player zero-sum game, called the \emph{\challengerProver{}} game (\challengerProverAb{} game). This game is a zero-sum game played between \emph{Challenger} (written $\chal{}$) and \emph{Prover} (written $\prov{}$). We will show that Player~$0$ has a solution to the \problemAb{} in an \gameAb{} if and only if $\prov{}$ has a winning strategy in the corresponding \challengerProverAb{} game. In the latter game, $\prov{}$ tries to show the existence of a strategy $\sigma_0$ that is solution to the \problemAb{} in the original game and $\chal{}$ tries to disprove it. The \challengerProverAb{} game is described independently of the objectives used in the \gameAb{} and its objective is described as such in a \emph{generic way}. We later provide the proof of our \FPT{} results by adapting it specifically for reachability, safety, Boolean B\"uchi, and finally prefix-independent objectives.

\paragraph{Intuition on the \challengerProverAb{} Game.} 
Without loss of generality, the \gamesAb{} we consider in this section are such that each vertex in their arena has at most \emph{two successors}. It can be shown (see Appendix \ref{app:usefull_notions}) that any \gameAb{} $\mathcal{G}$ with $n$ vertices can be transformed into an \gameAb{} $\bar{\mathcal{G}}$ with $\mathcal{O}(n^2)$ vertices such that every vertex has at most two successors and Player~$0$ has a solution to the \problemAb{} in $\mathcal{G}$ if and only if he has a solution to the \problemAb{} in $\bar{\mathcal{G}}$.

Let $\mathcal{G}$ be an \gameAb{}. The \challengerProverAb{} game $\mathcal{G'}$ is a zero-sum game associated with $\mathcal{G}$ that intuitively works as follows. First, $\prov{}$ selects a set $P$ of payoffs which he announces as the set of Pareto-optimal payoffs $\paretoSet{\sigma_0}$ for the solution $\sigma_0$ to the \problemAb{} in $\mathcal{G}$ he is trying to construct. Then, $\prov{}$ tries to show that there exists a set of witnesses $\Wit{\sigma_0}$ in $\mathcal{G}$ for the payoffs in $P$. After the selection of $P$ in $\mathcal{G}'$, there is a one-to-one correspondence between plays in the arenas $G$ and $G'$ such that the vertices in $G'$ are augmented with a set $\provWit$ which is a subset of $P$. Initially $\provWit$ is equal to $P$ and after some history in $G'$, $\provWit$ contains payoff $p \in P$ if the corresponding history in $G$ is prefix of the witness with payoff $p$ in the set $\Wit{\sigma_0}$ that $\prov{}$ is building. In addition, the objective $\Omega_{\prov{}}$ of $\prov{}$ is such that he has a winning strategy $\sigma_\prov$ in $\mathcal{G'}$ if and only if the set $P$ that he selected coincides with the set $\paretoSet{\sigma_0}$ for the corresponding strategy $\sigma_0$ in $\mathcal{G}$ and the latter strategy is a solution to the \problemAb{} in $\mathcal{G}$. A part of the arena of the \challengerProverAb{} game for Example~\ref{ex:example} with a positional winning strategy for $\prov{}$ highlighted in bold is illustrated in Figure~\ref{fig:example_challenger_prover}. 

\begin{figure} 
	\centering
		\resizebox{\textwidth}{!}{%
		\begin{tikzpicture}[->,>=stealth, shorten >=1pt,auto]
           
		\node[draw, circle, minimum size=0.75cm, inner sep = 0.5pt] (bot) at (9,11.5){$\bot$};
		\node[draw, rectangle, minimum size=0.75cm, minimum width={width("Magnetometer")}, inner sep = 2pt, rounded corners=11pt] (v0) at (10, 10){$v_0, P, \{p_1, p_2\}$};
		\node[draw, rectangle, minimum size=0.75cm, minimum width={width("Magnetometer")}, inner sep = 2pt] (v0sq) at (13, 10){$v_0, P, (\emptyset, \{p_1, p_2\})$};
		\node[draw, rectangle, minimum size=0.75cm, minimum width={width("Magnetometer")}, inner sep = 2pt, rounded corners=11pt] (v1) at (13, 11.5){$v_1, P, \emptyset$};
		\node[draw, rectangle, minimum size=0.75cm, minimum width={width("Magnetometer")}, inner sep = 2pt, rounded corners=11pt] (x) at (13, 8.5){$v_2, P, \{p_1, p_2\}$};

		\node[draw, rectangle, minimum size=0.75cm, minimum width={width("Magnetometer")}, inner sep = 2pt] (xsq) at (16, 8.5){$v_2, P, (\{p_1, p_2\}, \emptyset)$};
		\node[draw, rectangle, minimum size=0.75cm, minimum width={width("Magnetometer")}, inner sep = 2pt, rounded corners=11pt] (v2) at (16, 10){$v_3, P, \{p_1, p_2\}$};
		\node[draw, rectangle, minimum size=0.75cm, minimum width={width("Magnetometer")}, inner sep = 2pt, rounded corners=11pt] (v3) at (16, 7){$v_4, P, \emptyset$};

		\node[draw, rectangle, minimum size=0.75cm, minimum width={width("Magnetometer")}, inner sep = 2pt, rounded corners=11pt] (v4) at (19, 10){$v_5, P, \{p_1, p_2\}$};
		\node[draw, rectangle, minimum size=0.75cm, minimum width={width("Magnetometer")}, inner sep = 2pt, rounded corners=11pt] (v5) at (16, 11.5){$v_7, P, \{p_1, p_2\}$};
		\node[draw, rectangle, minimum size=0.75cm, minimum width={width("Magnetometer")}, inner sep = 2pt] (v4sq) at (22, 10){$v_5, P, (\{p_1\}, \{p_2\})$};
		
		\node[draw, rectangle, minimum size=0.75cm, minimum width={width("Magnetometer")}, inner sep = 2pt] (v4sqtre) at (22, 8.5){$v_5, P, (\{p_2\}, \{p_1\})$};
		\node[draw, rectangle, minimum size=0.75cm, minimum width={width("Magnetometer")}, inner sep = 2pt, rounded corners=11pt] (v2tre) at (22, 7){$v_3, P, \{p_2\}$};
		\node[draw, rectangle, minimum size=0.75cm, minimum width={width("Magnetometer")}, inner sep = 2pt, rounded corners=11pt] (v6tre) at (25, 8.5){$v_6, P, \{p_1\}$};
		\node[draw, rectangle, minimum size=0.75cm, minimum width={width("Magnetometer")}, inner sep = 2pt, rounded corners=11pt] (v5tre) at (25, 7){$v_7, P, \{p_2\}$};

		\node[draw, rectangle, minimum size=0.75cm, minimum width={width("Magnetometer")}, inner sep = 2pt] (v4sqfr) at (19, 8.5){$v_5, P, (\{p_1, p_2\}, \emptyset)$};
		\node[draw, rectangle, minimum size=0.75cm, minimum width={width("Magnetometer")}, inner sep = 2pt, rounded corners=11pt] (v6fr) at (19, 7){$v_6, P, \emptyset$};
				
    	\node[draw, rectangle, minimum size=0.75cm, minimum width={width("Magnetometer")}, inner sep = 2pt, rounded corners=11pt] (v2bis) at (22, 11.5){$v_3, P, \{p_1\}$};
		\node[draw, rectangle, minimum size=0.75cm, minimum width={width("Magnetometer")}, inner sep = 2pt, rounded corners=11pt] (v6) at (25, 10){$v_6, P, \{p_2\}$};
		
	    \node[draw, rectangle, minimum size=0.75cm, minimum width={width("Magnetometer")}, inner sep = 2pt, rounded corners=11pt] (v5bis) at (25, 11.5){$v_7, P, \{p_1\}$};

	    \node[rectangle, rounded corners=11pt, minimum size=0.8cm] (dotsbot) at (8, 10){$\dots$};
	    \node[rectangle, rounded corners=11pt, minimum size=0.8cm] (dotsv0) at (9, 8.5){$\dots$};
	    \node[rectangle, rounded corners=11pt, minimum size=0.8cm] (dotsv0bis) at (10, 8.5){$\dots$};
	    \node[rectangle, rounded corners=11pt, minimum size=0.8cm] (dotsv0tre) at (11, 8.5){$\dots$};
	    \node[rectangle, rounded corners=11pt, minimum size=0.8cm] (dotsx) at (13, 7){$\dots$};
	    \node[rectangle, rounded corners=11pt, minimum size=0.8cm] (dotsxbis) at (12, 7){$\dots$};
	    \node[rectangle, rounded corners=11pt, minimum size=0.8cm] (dotsxtre) at (14, 7){$\dots$};
	    \node[rectangle, rounded corners=11pt, minimum size=0.8cm] (dotsv4) at (19,11.5){$\dots$};
	    
	    \node[rectangle, rounded corners=11pt, minimum size=0.8cm] (dotsv3) at (22,12.75){$\dots$};
	    \node[rectangle, rounded corners=11pt, minimum size=0.8cm] (dotsv3bis) at (22,5.75){$\dots$};

	    \draw[-stealth, shorten >=1pt, auto] (bot) edge [] (dotsbot.90);
		\draw[-stealth, shorten >=1pt, auto] (v0) edge [] node {} (dotsv0);
    	\draw[-stealth, shorten >=1pt, auto] (v0) edge [] node {} (dotsv0bis);
		\draw[-stealth, shorten >=1pt, auto] (v0) edge [] node {} (dotsv0tre);
		\draw[-stealth, shorten >=1pt, auto] (x) edge [] node {} (dotsx);
		\draw[-stealth, shorten >=1pt, auto] (x) edge [] node {} (dotsxbis);
		\draw[-stealth, shorten >=1pt, auto] (x) edge [] node {} (dotsxtre);
		\draw[-stealth, shorten >=1pt, auto] (v4) edge [] node {} (dotsv4);		
		\draw[-stealth, shorten >=1pt, auto] (v2bis) edge [] node {} (dotsv3);		
		\draw[-stealth, shorten >=1pt, auto] (v2tre) edge [] node {} (dotsv3bis);

		\draw[-stealth, shorten >=1pt, auto, very thick]  (bot) edge [] node {} (v0.90);
		\draw[-stealth, shorten >=1pt, auto, very thick]  (v0) edge [] node {} (v0sq);
		\draw[-stealth, shorten >=1pt, auto]  (v0sq) edge [] node {} (v1);
		\draw[-stealth, shorten >=1pt, auto]  (v0sq) edge [] node {} (x);
		
		\draw[-stealth, shorten >=1pt, auto, very thick]  (x) edge [] node {} (xsq);
		\draw[-stealth, shorten >=1pt, auto]  (xsq) edge [] node {} (v2);
		\draw[-stealth, shorten >=1pt, auto]  (xsq) edge [] node {} (v3);
		
		\draw[-stealth, shorten >=1pt, auto, very thick]  (v2) edge [] node {} (v4);
		\draw[-stealth, shorten >=1pt, auto]  (v2) edge [] node {} (v5);
		\draw[-stealth, shorten >=1pt, auto, very thick]  (v4) edge [] node {} (v4sq);
		
		\draw[-stealth, shorten >=1pt, auto]  (v4sq) edge [] node {} (v2bis);
		\draw[-stealth, shorten >=1pt, auto]  (v4sq) edge [] node {} (v6);
		
		\draw[-stealth, shorten >=1pt, auto, very thick]  (v2bis) edge [] node {} (v5bis);
		
		\draw[-stealth, shorten >=1pt, auto, very thick]  (v1) edge [loop above] node {} (v1);
		\draw[-stealth, shorten >=1pt, auto, very thick]  (v5) edge [loop above] node {} (v5);
		\draw[-stealth, shorten >=1pt, auto, very thick]  (v5bis) edge [loop above] node {} (v5bis);
		
		\draw[-stealth, shorten >=1pt, auto, very thick]  (v3) edge [loop below] node {} (v3);
		\draw[-stealth, shorten >=1pt, auto, very thick]  (v6) edge [loop above] node {} (v6);
		
		\draw[-stealth, shorten >=1pt, auto]  (19.87, 9.78) to [] (20.82, 8.88);
		\draw[-stealth, shorten >=1pt, auto]  (v4sqtre) edge [] node {} (v6tre);
		\draw[-stealth, shorten >=1pt, auto]  (v4sqtre) edge [] node {} (v2tre);
		\draw[-stealth, shorten >=1pt, auto]  (v2tre) edge [] node {} (v5tre);
		
		\draw[-stealth, shorten >=1pt, auto]  (v4) edge [] node {} (v4sqfr);
		\draw[-stealth, shorten >=1pt, auto]  (v4sqfr) edge [] node {} (v6fr);
		\draw[-stealth, shorten >=1pt, auto]  (17.82, 8.87) to [] (16.88, 9.78);
		
		\draw[-stealth, shorten >=1pt, auto, very thick]  (v6fr) edge [loop below] node {} (v6fr);

		\draw[-stealth, shorten >=1pt, auto]  (v6tre) edge [loop below] node {} (v6tre);
		\draw[-stealth, shorten >=1pt, auto]  (v5tre) edge [loop below] node {} (v5tre);

		\end{tikzpicture}
	}%
	\caption{A part of the \challengerProverAb{} game for Example \ref{ex:example} with $P = \{p_1, p_2\}, p_1 = (1,1,0)$ and $p_2 = (0,1,1)$. }
	\label{fig:example_challenger_prover}
	\Description{Figure 2. Fully described in the text.}
\end{figure} 

\paragraph{Arena of the \challengerProverAb{} Game.} The initial vertex is $\bot$ and it belongs to $\prov{}$. From this vertex, $\prov{}$ selects a successor $(v_0, P, \provWit)$ such that $W = P$ and $P$ is an antichain of payoffs which $\prov$ announces as the set $\paretoSet{\sigma_0}$ for the strategy $\sigma_0$ in $G$ he is trying to construct. All vertices in plays starting with this vertex will have this same value for their $P$-component. Those vertices are either a triplet $(v, P, \provWit)$ that belongs to $\prov$ or $(v, P, (\provWit_l, \provWit_r))$ that belongs to~$\chal{}$. Given a play $\rho$ (resp.\ history $h$) in $G'$, we denote by $\rho_V$ (resp.\ $h_V$) the play (resp.\ history) in $G$ obtained by removing $\bot$ and keeping the $v$-component of every vertex of $\prov{}$ in $\rho$ (resp.\ $h$), which we call its \emph{projection}. Notice that we only consider the $v$-component of vertices of $\prov{}$ to avoid duplication of vertices $v \in V_1$ in the projection due to the vertices of $\chal{}$.
\begin{itemize}
    \item After history $hm$ such that $m = (v, P, \provWit)$ with $v \in V_0$, $\prov{}$ selects a successor $v'$ such that $(v, v') \in E$ and vertex $(v', P, \provWit)$ is added to the play. This corresponds to Player~$0$ choosing a successor $v'$ after history $h_V v$ in~$G$. 
    \item After history $hm$ such that $m = (v, P, \provWit)$ with $v \in V_1$, $\prov{}$ selects a successor $(v, P, (\provWit_l, \provWit_r))$ with $(\provWit_l, \provWit_r)$ a partition of $\provWit$. This corresponds to $\prov{}$ splitting the set $W$ into two parts according to the two successors $v_l$ and $v_r$ of $v$. For the strategy $\sigma_0$ that $\prov{}$ tries to construct and its set of witnesses $\Wit{\sigma_0}$ he is building, he asserts that $\provWit_l$ (resp.\ $\provWit_r$) is the set of payoffs of the witnesses in $\prefStrat{h_V v_l}$ (resp.\ $\prefStrat{h_V v_r}$).
    \item From a vertex $(v, P, (\provWit_l, \provWit_r))$, $\chal{}$ can select a successor $(v_l, P, \provWit_l)$ or $(v_r, P, \provWit_r)$ which corresponds to the choice of Player~$1$ in~$G$. 
\end{itemize}   

\noindent
Formally, the game arena of the \challengerProverAb{} game is the tuple $G' = (V', V'_{\prov{}}, V'_{\chal{}}, E', \bot)$ with
\begin{itemize}
    \item $V'_{\prov{}} = \{ \bot \} \cup \{ (v, P, \provWit) \mid v \in V, P \subseteq \{0,1\}^\nbrObjectives \text{ is an antichain and } \provWit \subseteq P \}$,
    \item $V'_{\chal{}} = \{ (v, P, (\provWit_l, \provWit_r)) \mid v \in V_1, P \subseteq \{0,1\}^\nbrObjectives \text{ is an antichain and } \provWit_l, \provWit_r \subseteq P \}$,
    \item $(\bot, (v, P, \provWit)) \in E'$ if $v = v_0$ and $P = \provWit$,
    \item $((v, P, \provWit), (v', P, \provWit)) \in E'$ if $v \in V_0$ and $(v, v') \in E$,
    \item $((v, P, \provWit), (v, P, (\provWit_l, \provWit_r))) \in E'$ if $v \in V_1$ and $(\provWit_l, \provWit_r)$ is a partition of $\provWit$, 
    \item $((v, P, (\provWit_l, \provWit_r)), (v', P, \provWit) ) \in E'$ if $(v, v') \in E$ and $\{v' = v_l$ and $\provWit = \provWit_l\}$ or 
    $\{v' = v_r$ and $\provWit = \provWit_r\}$.
\end{itemize}
In the definition of $E'$, if $v$ has a single successor $v'$ in $G$, it is assumed to be $v_l$ and $W_r$ is always equal to $\emptyset$. We use as a convention that given the two successors $v_i$ and $v_j$ of vertex $v$, $v_i$ is the left successor if $i < j$. 

\paragraph{Objective of $\prov$ in the \challengerProverAb{} Game.}
Let us now discuss the objective $\Omega_{\prov{}}$ of $\prov{}$. The $\provWit$-component of the vertices controlled by $\prov{}$ has a size that decreases along a play $\rho$ in $G'$. We write $lim_\provWit(\rho)$ the value of the $\provWit$-component at the limit in $\rho$. Recall that with the $P$-component and this $\provWit$-component, $\prov{}$ tries to construct a solution $\sigma_0$ to the \problemAb{} with associated sets $P =  \paretoSet{\sigma_0}$ and $\Wit{\sigma_0}$. Therefore, for him to win in the \challengerProverAb{} game, $lim_\provWit(\rho)$ must be a singleton or empty in every consistent play such that:
\begin{itemize}
    \item $lim_\provWit(\rho)$ must be a singleton $\{p\}$ with $p$ the payoff of $\rho_V$ in $G$, showing that $\rho_V \in \Wit{\sigma_0}$ is a correct witness for $p$. In addition, it must hold that $\won{\rho_V} = 1$ as $p \in P$ and as $\prov{}$ wants $\sigma_0$ to be a solution.
    \item $lim_\provWit(\rho)$ must be the empty set such that either the payoff of $\rho_V$ belongs to $P$ and $\won{\rho_V} = 1$, or the payoff of $\rho_V$ is strictly smaller than some payoff in~$P$.
\end{itemize}
These conditions verify that the sets $P = \paretoSet{\sigma_0}$ and $\Wit{\sigma_0}$ are correct and that $\sigma_0$ is indeed a solution to the \problemAb{} in $G$. They are \emph{generic} as they do not depend on the actual objectives used in the \gameAb{}. 

Let us give the formal definition of $\Omega_{\prov{}}$. For an antichain $P$ of payoffs, we write $\Plays^P_{G'}$ the set of plays in $G'$ which start with $\bot (v_0, P, P)$ and we define the following set 
\begin{alignat}{3}
B_P = \big{\{} \rho \in \Plays^P_{G'} ~\mid~ &(lim_\provWit(\rho) = \{p\} &&\land \payoff{\rho_V} = p \in P &&\land \won{\rho_V} = 1) \ \lor \label{cp_cond_1}\\
                &(lim_\provWit(\rho) = \emptyset &&\land \payoff{\rho_V} \in P &&\land \won{\rho_V} = 1) \ \lor \label{cp_cond_2} \\
             &(lim_\provWit(\rho) = \emptyset &&\land \exists p \in P, \payoff{\rho_V} < && ~p) \big{\}}. \label{cp_cond_3}
\end{alignat}
Objective $\Omega_{\prov{}}$ of $\prov{}$ in $\mathcal{G'}$ is the union of $B_P$ over all antichains $P$. As the \challengerProverAb{} game is zero-sum, objective $\Omega_{\chal{}}$ equals $\Plays_{G'} \setminus \Omega_{\prov{}}$. The following theorem holds.

\begin{theorem}
\label{thm:CP} 
Player~$0$ has a strategy $\sigma_0$ that is solution to the \problemAb{} in $\mathcal{G}$ if and only if $\prov{}$ has a winning strategy $\sigma_{\prov{}}$ from $\bot$ in the \challengerProverAb{} game $\mathcal{G'}$.
\end{theorem}

\begin{proof}
    Let us first assume that Player~$0$ has a strategy $\sigma_0$ that is solution to the \problemAb{}{} in $\mathcal{G}$. Let $\paretoSet{\sigma_0}$ be its set of \paretoOptimal\ payoffs and let $\Wit{\sigma_0}$ be a set of witnesses. We construct the strategy $\sigma_{\prov{}}$ from $\sigma_0$ such that
    \begin{itemize}
        \item $\sigma_{\prov{}}(\bot) = (v_0, P, P)$ such that $P = \paretoSet{\sigma_0}$ (this vertex exists as $\paretoSet{\sigma_0}$ is an antichain),
        \item $\sigma_{\prov{}}(hm) = (v', P, \provWit)$ if $m = (v, P, \provWit)$ with $v \in V_0$ and $v' = \sigma_0(h_V v)$, 
        \item $\sigma_{\prov{}}(hm) = (v, P, (\provWit_l, \provWit_r))$ if $m = (v, P, \provWit)$ with $v \in V_1$ and for $i \in\{l, r\}$, $\provWit_i = \{ \payoff{\rho} \mid \rho \in \prefStrat{h_V v_i}\}$.
    \end{itemize}   
    It is clear that given a play $\rho$ in $G'$ consistent with $\sigma_{\prov}$, the play $\rho_V$ in $G$ is consistent with $\sigma_0$. Let us show that $\sigma_{\prov}$ is winning for $\prov{}$ from $\bot$ in $G'$. Consider a play $\rho$ in $G'$ consistent with $\sigma_\prov$. There are two possibilities. \emph{(i)} $\rho_V$ is a witness of $\Wit{\sigma_0}$ and by construction $lim_\provWit(\rho) = \{p\}$ with $p = \payoff{\rho_V}$; thus $\won{\rho_V} = 1$ as $\sigma_0$ is a solution and $\rho_V$ is a witness. \emph{(ii)} $\rho_V$ is not a witness and by construction $lim_\provWit(\rho) = \emptyset$; as $\sigma_0$ is a solution, then $p = \payoff{\rho_V}$ is bounded by some payoff of $\paretoSet{\sigma_0}$ and in case of equality $\won{\rho_V} = 1$. Therefore $\rho$ satisfies the objective $B_P$ of $\Omega_{\prov{}}$ since it satisfies condition (\ref{cp_cond_1}) in case \emph{(i)} and condition (\ref{cp_cond_2}) or (\ref{cp_cond_3}) in case \emph{(ii)}. 

    Let us now assume that $\prov$ has a winning strategy $\sigma_{\prov}$ from $\bot$ in $G'$. Let $P$ be the antichain of payoffs chosen from $\bot$ by this strategy. We construct the strategy $\sigma_0$ from $\sigma_{\prov}$ such that $\sigma_0(h_Vv) = v'$ given $\sigma_{\prov}(hm) = (v', P, \provWit)$ with $m = (v, P, \provWit)$ and $v \in V_0$. Notice that this definition makes sense since there is a unique history $hm$ consistent with $\sigma_{\prov}$ ending with a vertex of $\prov{}$ associated with $h_Vv$ showing a one-to-one correspondence between those histories. 
    
    Let us show $\sigma_0$ is a solution to the \problemAb{} with $\paretoSet{\sigma_0}$ being the set $P$. 
    First notice that $P$ is not empty. Indeed let $\rho$ be a play consistent with $\sigma_{\prov}$. As $\rho$ belongs to $\Omega_{\prov{}}$ and in particular to $B_P$, one can check that $P \neq \emptyset$ by inspecting conditions~(\ref{cp_cond_1}) to (\ref{cp_cond_3}). Second notice that by definition of $E'$, if $((v, P, \provWit), (v, P, (\provWit_l, \provWit_r))) \in E'$ with $\provWit \neq \emptyset$, then either $\provWit_l$ or $\provWit_r$ is not empty. Therefore given any payoff $p \in P$, there is a unique play $\rho$ consistent with $\sigma_{\prov}$ such that $lim_\provWit(\rho) = \{p\}$. By construction of $\sigma_0$ and as $\sigma_{\prov}$ is winning, the play $\rho_V$ is consistent with $\sigma_0$, has payoff $p$, and is won by Player~$0$ (see (\ref{cp_cond_1})). 
    
    Let $\rho_V$ be a play consistent with $\sigma_0$ and $\rho$ be the corresponding play consistent with $\sigma_{\prov}$. It remains to consider (\ref{cp_cond_2}) and (\ref{cp_cond_3}). These conditions indicate that $\rho_V$ has a payoff equal to or strictly smaller than a payoff in $P$ and that in case of equality $\won{\rho_V} = 1$. This shows that $\paretoSet{\sigma_0} = P$ and that $\sigma_0$ is a solution to the \problemAb{}.
\end{proof}

\subsection{Fixed-Parameter Complexity of Reachability and Safety SP Games} 
\label{subsec:fptreach}

We now develop the proof of Theorem~\ref{thm:FPT_all} for reachability and safety \gamesAb{} which works by specializing the generic objective $\Omega_{\prov{}}$ of the \challengerProverAb{} game to handle reachability or safety objectives. Let $\mathcal{G} = (G, \ObjPlayer{0},\ObjPlayer{1}, \dots, \ObjPlayer{\nbrObjectives})$ be an \gameAb{} with either reachability objectives $\ObjPlayer{i} = \reach{T_i}$, $i \in \{0, \ldots, \nbrObjectives\}$, or safety objectives $\safe{S_i}$, $i \in \{0, \ldots, \nbrObjectives\}$. We start by extending the arena of this \challengerProverAb{} game with additional information.\\

For reachability, we extend the arena $G'$ of the \challengerProverAb{} game such that its vertices keep track of the objectives of $\mathcal{G}$ which are satisfied along a play. Given an extended payoff $(w,p) \in \{0,1\} \times \{0,1\}^{\nbrObjectives}$ and a vertex $v \in V$, we define the \emph{reachability payoff update} $\updateReach{w,p,v} = (w',p')$ such that
$$\begin{array}{lll}
w' = 1   &\iff & w = 1 \text{ or } v \in \target_0,\\
p'_i = 1 &\iff& p_i = 1 \text{ or } v \in \target_i, \quad \forall i \in \{1, \ldots, \nbrObjectives\}.
\end{array}$$
We obtain the extended arena $G^*$ from $G'$ as follows: \emph{(i)} its set of vertices is $V' \times \{0,1\} \times \{0,1\}^\nbrObjectives$, \emph{(ii)} its initial vertex is $ \bot^* = (\bot, 0, (0,\ldots,0))$, and \emph{(iii)} $((q, w, p),(q', w', p'))$ with $q' = (v', P, \provWit)$ or $q' = (v', P, (\provWit_l, \provWit_r))$ is an edge in $G^*$ if $(q, q') \in E'$ and $(w',p') = \updateReach{w,p,v'}$. 
Recall that in some edges in $E'$, the same $v$-component is repeated twice, and that such repetitions are eliminated when projecting a play or a history from $G'$ to $G$. The way we keep track of the satisfied objectives is not affected by those repetitions.

Similarly, for safety we extend the arena $G'$ such that its vertices keep track of the safety objectives of $\mathcal{G}$ which are \emph{not satisfied} along a play. Given an extended payoff $(w,p) \in \{0,1\} \times \{0,1\}^{\nbrObjectives}$ and a vertex $v \in V$, we define the \emph{safety payoff update} $\updateSafe{w,p,v} = (w',p')$ such that
$$\begin{array}{lll}
w' = 1   &\iff & w = 1 \text{ and } v \in \safeSet_0,\\
p'_i = 1 &\iff& p_i = 1 \text{ and } v \in \safeSet_i, \quad \forall i \in \{1, \ldots, \nbrObjectives\}.
\end{array}$$
The extended arena is obtained in the same way as for reachability, with the difference that its initial vertex is $ \bot^* = (\bot, 1, (1,\ldots,1))$. We define the zero-sum game $\mathcal{G^*} = (G^*, \Omega^*_\prov)$ in which the three abstract conditions (\ref{cp_cond_1}-\ref{cp_cond_3}) detailed previously are encoded into the following B\"uchi objective by using the $(w, p)$-component added to vertices. We define $\Omega^*_\prov = \Buchi{B^*}$ with
\begin{alignat}{1}
B^* = \big{\{} (v, P, \provWit, w, p) \in V^*_{\prov{}} ~\mid~
            &(W = \{p\} \land  w = 1) \ \lor \tag{\ref{cp_cond_1}'} \\
            &(W = \emptyset  \land  p \in P  \land w = 1) \ \lor \tag{\ref{cp_cond_2}'} \\
            &(W = \emptyset  \land  \exists p' \in  P,~ p < p') \big{\}}. \tag{\ref{cp_cond_3}'}
\end{alignat}

We introduce the following proposition, the proof of which is a consequence of Theorem~\ref{thm:CP}.

\begin{proposition}
\label{prop:FPTreach}
Player~$0$ has a strategy $\sigma_0$ that is solution to the \problemAb{} in a reachability or safety \gameAb{} $\mathcal{G}$ if and only if $\prov{}$ has a winning strategy $\sigma^*_{\prov}$ in $\mathcal{G^*}$.
\end{proposition}

\begin{proof}
Given the fact that $\mathcal{G}$ is a reachability or safety \gameAb{}, the $(w,p)$-component in vertices of $G^*$ allows us to easily retrieve the extended payoff of a play in $G$. Indeed, in a play $\rho \in \Plays_{G^*}$, given the construction of $G^*$ and the payoff update function, it holds that from some point on the $W$- and $(w,p)$-components are constant. Therefore it holds that $w = \won{\rho_V}$, $p = \payoff{\rho_V}$ and $\provWit = lim_\provWit(\rho)$ for that play $\rho$. Moreover the {$P$-component} is constant along a play in $G^*$. 
It is direct to see that the plays $\rho$ in $G^*$ which visit infinitely often the set $B^*$, and therefore satisfy the B\"uchi objective $\Omega^*_\prov = \Buchi{B^*}$, satisfy one of the three conditions (\ref{cp_cond_1}-\ref{cp_cond_3}) stated in Subsection~\ref{subsec:cp}. The converse is also true.
\end{proof}

We now describe an \FPT{} algorithm for deciding the existence of a solution to the \problemAb{} in a reachability or safety \gameAb{}, thus proving Theorem~\ref{thm:FPT_all} for reachability and safety \gamesAb{}. The parameter of this \FPT{} algorithm is $\nbrObjectives$ and its complexity is double exponential in $\nbrObjectives$ as stated in Table~\ref{table:fpt-summary}.

\begin{proof}[Proof of Theorem~\ref{thm:FPT_all} for reachability and safety \gamesAb{}.]
We describe the following \FPT{} algorithm (for parameter~$\nbrObjectives$) for deciding the existence of a solution to the \problemAb{} in a reachability or safety \gameAb{} $\mathcal{G}$ by using Proposition~\ref{prop:FPTreach}. First, we construct the zero-sum game $\mathcal{G^*}$. Its number $n$ of vertices is upper-bounded by 
\begin{eqnarray}
\label{eq:size}
1 + |V| \cdot 2^{2 \cdot 2^{\nbrObjectives}} \cdot 2^{\nbrObjectives+1} + |V| \cdot 2^{3 \cdot 2^t} \cdot 2^{\nbrObjectives+1}. 
\end{eqnarray}
Indeed, except the initial vertex, vertices are of the form either $(v, P, \provWit,w,p)$ or $(v, P, (\provWit_l, \provWit_r),w,p)$ such that $P$, $\provWit$, $\provWit_l$ and $\provWit_r$ are antichains of payoffs in $\{0,1\}^\nbrObjectives$, and $(w,p)$ is an extended payoff. The construction of $\mathcal{G^*}$ is thus in \FPT{} for parameter~$\nbrObjectives$, with a time complexity double exponential in $\nbrObjectives$. Second, By Proposition~\ref{prop:FPTreach}, deciding whether there exists a solution to the \problemAb{} in $\mathcal{G}$ amounts to deciding if $\prov$ has a winning strategy from $\bot^*$ in $\mathcal{G^*}$. Since the objective $\Omega^*_\prov$ of $\prov$ in $\mathcal{G^*}$ is a B\"uchi objective, this game can be solved in $\mathcal{O}(n^2)$~\cite{ChatterjeeH14}. It follows that $\mathcal{G^*}$ can be solved in \FPT{} for parameter~$\nbrObjectives$, with a time complexity double exponential in $\nbrObjectives$.
\end{proof}

\subsection{Fixed-Parameter Complexity of SP Games with Prefix-Independent Objectives} 
\label{subsec:fptboolean}

In order to prove Theorem~\ref{thm:FPT_all} for prefix-independent objectives, we first show that solving the \problemAb{} for Boolean B\"uchi \gamesAb{} is in \FPT{}. Then, using the relationships between Boolean B\"uchi and the other prefix-independent objectives stated in Table~\ref{table:summaryBooleanBuchiTable}, we extend this result to the later objectives by reduction to Boolean B\"uchi \gamesAb{}. \\

The intuition behind the \FPT{} result for Boolean B\"uchi \gamesAb{} is the following: we extend the \challengerProverAb{} game to keep track for each Boolean B\"uchi objective of the sets that are visited infinitely often along a play. This allows us to retrieve for a play in the \challengerProverAb{} game the extended payoff of the corresponding play in the original game arena. We then encode the generic objective of the \challengerProverAb{} game into a parity objective using this information.

\paragraph{Extending the \challengerProverAb{} Game with the \textsf{SAR}}
We start by considering a Boolean B\"uchi \gameAb{} $\mathcal{G} = (G, \Omega_0, \Omega_1, \dots, \Omega_t)$ and the arena $G'$ of the corresponding \challengerProverAb{} game. 
 Our goal is to extend the \challengerProverAb{} game with a single \textsf{SAR} structure, as introduced in Subsection~\ref{subsec:AlternativeFPTBooleanBuchi}, for all Boolean B\"uchi objectives of $\mathcal{G}$. This allows us to keep track in a play in $G'$ of the recently occurring sets for each Boolean B\"uchi objective in the corresponding play in $G$. To do so, we devise the following \textsf{SAR} structure. Given the objectives $\Omega_i = \BooleanBuchi{\phi_i, T^i_1, \dots, T^i_{m_i}}$ for $i \in \{0, 1, \dots, \nbrObjectives\}$, we write $T = \bigcup_{i=0}^{\nbrObjectives} \{T^i_1, \dots, T^i_{m_i}\}$ the set of all sets used to define the Boolean B\"uchi objectives. 
The extended arena $G^*= (V^*, V^*_{\prov{}}, V^*_{\chal{}}, E^*, \bot)$ of the \challengerProverAb{} game is defined as follows (we use the previous notations of the \textsf{SAR} automaton):
\begin{itemize}
    \item $V^*_{\prov{}} = V'_{\prov{}} \times \mathcal{P}(T^0_1, \dots, T^0_{m_0}, \dots, T^\nbrObjectives_1, \dots, T^\nbrObjectives_{m_\nbrObjectives})  \times \{1, \dots, |T| +1\}$,
    \item $V^*_{\chal{}} = V'_{\chal{}} \times \mathcal{P}(T^0_1, \dots, T^0_{m_0}, \dots, T^\nbrObjectives_1, \dots, T^\nbrObjectives_{m_\nbrObjectives})  \times \{1, \dots, |T| +1\}\}$,
    \item $(\bot, (v, P, \provWit, s, r)) \in E^*$ if $v = v_0$, $P = \provWit$, $s = s_0$ and $r = r_0$,
    \item $((q, s, r), (q', s', r')) \in E^*$ with $q, q' \in V'$ and $v$ the $v$-component of $q$ if $(q, q') \in E'$ and $(s', r') = \delta((s, r),v)$.
\end{itemize}

Notice that the \textsf{SAR} structure is updated using the $v$-component of the vertices in $V^*$, which is a vertex from $V$. While in a play $\rho^*$ of $G^*$ the $\textsf{SAR}$ may be updated with the same vertex twice (this vertex being eliminated in the projection $\rho^*_V$ of $\rho^*$), this does not affect the $\textsf{SAR}$. Hence the vertices of $G^*$ embed the run of the $\textsf{SAR}$ for the sets in each Boolean B\"uchi objective in the corresponding projection in $G$. Given a play in $G^*$ we can thus determine the payoff of the corresponding play in $G$ for each of the Boolean B\"uchi objective $\Omega_i$, with $i \in \{0, \dots, \nbrObjectives\}$.

\paragraph{Transforming the Generic Objective of the \challengerProverAb{} Game}
Now that we can obtain the payoff of the projection in $G$ of a play in $G^*$, we can use this information to transform the generic objective of the \challengerProverAb{} game into a parity objective in $G^*$, exactly as we did in Subsection~\ref{subsec:AlternativeFPTBooleanBuchi}. The idea behind this objective is to ensure that a play in the extended \challengerProverAb{} game $G^*$ satisfies the parity objective if and only if one of the three abstract conditions (\ref{cp_cond_1}-\ref{cp_cond_3}) is satisfied in that play. Notice that these conditions require obtaining the extended payoff of the projection of a play in $G^*$, which we obtain using the \textsf{SAR}. Given a state $(q, s, r)$ of $G^*$, the formula $\phi_i$ of each Boolean B\"uchi objective $\Omega_i$ and its set of variables $X_i$, we write $pay(s, r) \in \{0, 1\}^t$ the vector of Booleans such that $pay(s, r)_i = 1$ if and only if the valuation of $X_i$ such that $x^i_j = 1$ if and only if $T^i_j \in s_{\geq {r}}$ and $x^i_j = 0$ otherwise satisfies $\phi_i$. We define $won(s, {r}) \in \{0, 1\}$ similarly for $\Omega_0$. We define the following priority function $c$ on the vertices of $G^*$. 
\[
    c((v, P, W, s, r)) = \begin{cases}
        2 \cdot {{r}} - 2 & \text{if } (W = \{pay(s, {r})\} \land won(s, {r}) = 1) \ \lor\\
                        & (W = \emptyset \land pay(s, {r}) \in P \land won(s, {r}) = 1) \ \lor\\
                        & (W = \emptyset \land  \exists p \in P, pay(s, {r}) < p),\\
        2 \cdot {{r}} - 1 & \text{otherwise.}
    \end{cases}
\]


It can be shown that a play satisfies the parity objective $\parity{c}$ if and only if it is winning in the \challengerProverAb{} game and therefore satisfies one of the three conditions.

We can now show the fixed-parameter complexity of solving the \problemAb{} on \gamesAb{} for Boolean B\"uchi objectives and then for all prefix-independent objectives, therefore completing the proof of Theorem~\ref{thm:FPT_all} while establishing the complexities stated in Table~\ref{table:fpt-summary}. 

\begin{proof}[Proof of Theorem~\ref{thm:FPT_all} for prefix-independent objectives.] 
We begin by proving the result for Boolean B\"uchi objectives. Let $m = |T|$ and $k = \sum_{i=0}^t |\phi_i|$.
Deciding whether there exists a solution to the \problemAb{} reduces to finding a winning strategy for $\prov{}$ in the extended \challengerProverAb{} game $\mathcal{G}^* = (G^*,\parity{c})$. The number of vertices in $G^*$ is computed as in (\ref{eq:size}) (which is the size of the extended \challengerProverAb{} game for reachability objectives), where each factor $2^{t+1}$ (equal to the number extended payoffs $(w,p)$) is replaced by $m! \cdot (m+1)$ (equal to the number of possible configurations of the \textsf{SAR} structure given $T$). Its priority function $c$ is such that $\max(c) = 2 \cdot (m + 1)$ and requires to evaluate $pay(s, {r})$ and $won(s, {r})$ for each vertex $(q, s, r)$ of $G^*$. We then repeat the argument as given in the proof of Theorem~\ref{thm:bbnewcomp}. Constructing $\mathcal{G}^*$ is in \FPT, with a time complexity polynomial in $|V|$, double exponential in $\nbrObjectives$, exponential in $m$, and linear in $k$. By Theorem~\ref{thm:calude}, solving $\mathcal{G}^*$ is in \FPT{} with a time complexity polynomial in $|V|$, double exponential in $\nbrObjectives$, and exponential in $m$. The whole algorithm (constructing $\mathcal{G}^*$ and solving it) is then in \FPT, with a time complexity polynomial in $|V|$, double exponential in $\nbrObjectives$, exponential in $m$, and linear in $k$.

Let us now consider an \gameAb{} $\mathcal{G}$ for a prefix-independent objective. Using Proposition~\ref{prop_bb_encoding}, we can encode each objective of $\mathcal{G}$ into an equivalent Boolean B\"uchi objective in polynomial time in the size of the original objective and independently of the size of the arena. We consider the Boolean B\"uchi \gameAb{} $\mathcal{G'}$ with the same arena $G$ and where every objective $\Omega'_i$ is the Boolean B\"uchi encoding of the corresponding objective $\Omega_i$ in $\mathcal{G}$. It follows that solving the \problemAb{} in $\mathcal{G}$ is equivalent to solving the \problemAb{} in $\mathcal{G'}$ which is in \FPT{} by the first part of the proof. In view of Table~\ref{table:summaryBooleanBuchiTable}, one can check the time complexity of this \FPT{} algorithm as given in Table~\ref{table:fpt-summary}.
\end{proof}

We conclude this section with the following remark. Instead of imposing that an \gameAb{} have the same type of objective for all $\Omega_i$, $i \in \{0,\ldots,\nbrObjectives\}$, we could allow the mixing of different kinds of prefix-independent objectives. One can check that the \FPT{} result of Theorem~\ref{thm:FPT_all} is still valid in this general context.

\begin{remark}
Solving the \problemAb{} is in \FPT{} for \gamesAb{} with a collection of different prefix-independent objectives.
\end{remark}

\section{\textsf{NEXPTIME} Membership}
\label{sec:nexptime_member}

In this section we show the  membership to \nexptime{} of the \problemAb{} by providing a nondeterministic algorithm with time exponential in the size $|\mathcal{G}|$ of the \gameAb{} $\mathcal{G}$ (recall that $|\mathcal{G}|$ depends on the number $|V|$ of vertices of $G$ and on $\nbrObjectives$, see Definition~\ref{def:SPgame}). Notice that the time complexity of the \FPT{} algorithms obtained in the previous section is too high, preventing us from directly using the \challengerProverAb{} game to show a tight membership result. Conversely, the nondeterministic algorithm provided in this section is not \FPT{} as it is exponential in $|V|$. In the next section, we provide a better upper bound for B\"uchi \gamesAb{} by providing an \np{} algorithm. We study the hardness of the \problemAb{} in Section~\ref{sec:nexptime_hard}.

\begin{theorem} \label{thm:nexptime}
The \problemAb{} is in \nexptime{} for \gamesAb{}. 
\end{theorem}

We show that the \problemAb{} is in \nexptime{} by proving that if Player~$0$ has a strategy which is a solution to the problem, then he has one which is finite-memory with at most an exponential number of memory states\footnote{Recall that to have a solution to the \problemAb{}, memory is sometimes necessary as shown in Example~\ref{ex:example}.}.
This yields a \nexptime{} algorithm in which we nondeterministically guess such a strategy and check in exponential time that it is indeed a solution to the problem. 

\begin{proposition} \label{prop:tildesigma}
Let $\mathcal{G}$ be an \gameAb{}. Let $\sigma_0$ be a solution to the \problemAb{}. Then there exists another solution $\tilde{\sigma}_0$ that is finite-memory and has a memory size exponential in the size of $\mathcal{G}$.
\end{proposition}

While the proof of Proposition~\ref{prop:tildesigma} requires some specific arguments to treat reachability, safety or prefix-independent objectives $\Omega_i$, $i \in \{0,\ldots,\nbrObjectives\}$, it is based on the following common principles. 

\begin{itemize}
    \item We start from a winning strategy $\sigma_0$ for the \problemAb{} and the objectives $\ObjPlayer{0},\ObjPlayer{1},\dots,\ObjPlayer{\nbrObjectives}$ and we consider a set of witnesses $\Wit{\sigma_0}$, that contains one play for each element of the set  $\paretoSet{\sigma_0}$ of \paretoOptimal{} payoffs.
    \item We start by showing the existence of a strategy $\hat{\sigma}_0$ constructed from $\sigma_0$, in which Player~$0$ follows $\sigma_0$ as long as the current consistent history is prefix of at least one witness in $\Wit{\sigma_0}$. Then when a deviation from $\Wit{\sigma_0}$ occurs, Player~$0$ switches to a so-called \emph{punishing strategy}. A deviation is a history that leaves the set of witnesses $\Wit{\sigma_0}$ after a move of Player~$1$ (this is not possible by a move of Player~$0$). After such a deviation, $\hat{\sigma}_0$ systematically imposes that the consistent play either satisfies $\ObjPlayer{0}$ or is not \paretoOptimal{}, i.e., it gives to Player~1 a payoff that is strictly smaller than the payoff of a witness in $\Wit{\sigma_0}$. This makes the deviation \emph{irrational} for Player~1.  We show that this can be done for each kind of objective with at most exponentially many different punishing strategies, each having a memory size bounded exponentially in the size of the game. The strategy $\hat{\sigma}_0$ that we obtain is therefore composed of the part of $\sigma_0$ that produces $\Wit{\sigma_0}$ and a punishment part whose size is at most exponential.
    \item Then, we show how to decompose each witness in $\Wit{\sigma_0}$ into at most exponentially many \emph{sections} that can, in turn, be compacted into finite elementary paths or lasso shaped paths of polynomial length. As $\Wit{\sigma_0}$ contains exactly $|\paretoSet{\sigma_0}|$ witnesses $\rho$, those compact witnesses $c\rho$ can be produced by a finite-memory strategy with an exponential size, regardless of the objective considered. This allows us to construct a strategy $\tilde{\sigma}_0$ that produces the compact witnesses and acts as $\hat{\sigma}_0$ after any deviation. This strategy is a solution of the \problemAb{} and has an exponential size as announced.
\end{itemize}

\begin{figure}
	\centering
		\resizebox{\textwidth}{!}{%
		\begin{tikzpicture}
          

        \draw[-, thick] (10,10) edge [] node {} (7, 5);
        \draw[-, thick] (10,10) edge [] node {} (13, 5);

        \node[] (0) at (9.55,9.9){$\sigma_0$};
        \node[] (0) at (10,10){};
        
        \node[rectangle, fill=darkgray, inner sep = 0pt, minimum size=4pt] (1) at (10,9){};
        \node[rectangle, fill=darkgray, inner sep = 0pt, minimum size=4pt] (2) at (9.5,8){};
        \node[rectangle, fill=darkgray, inner sep = 0pt, minimum size=4pt] (3) at (11.25,6.5){};
        
        \node[] (4) at (7.5,5){};
        \node[] () at (7.85,5.2){$\rho_1$};
        \node[] (5) at (9.6,5){};
        \node[] () at (9.95,5.2){$\rho_2$};
        \node[] (6) at (11.5,5){};
        \node[] () at (11.15,5.2){$\rho_3$};
        \node[] (7) at (12.5,5){};
        \node[] () at (12.22,5.2){$\rho_4$};

		\draw[semithick] plot [smooth, tension=0.7] coordinates { (0) (1) (11,7.75) (3) (12.25, 5.75) (7)};
		\draw[semithick] plot [smooth, tension=0.7] coordinates { (3) (11.25,5.75) (6)};
		\draw[semithick] plot [smooth, tension=0.7] coordinates { (1) (2) (9.5,7) (9.75,6) (5)};
		\draw[semithick] plot [smooth, tension=0.7] coordinates { (2) (8.25, 6.5) (4)};
		
        \node[draw, rectangle, fill=darkgray, inner sep = 0pt, minimum size=4pt] () at (10,9){};
        \node[draw, rectangle, fill=darkgray, inner sep = 0pt, minimum size=4pt] () at (9.5,8){};
        \node[draw, rectangle, fill=darkgray, inner sep = 0pt, minimum size=4pt] () at (11.25,6.5){};

        
        \draw[-, thick] (16.5,10) edge [] node {} (13.5, 5);
        \draw[-, thick] (16.5,10) edge [] node {} (19.5, 5);

        \node[] (0bis) at (16.05,9.9){$\hat{\sigma}_0$};
        \node[] (0bis) at (16.5,10){};
        
        \node[rectangle, fill=darkgray, inner sep = 0pt, minimum size=4pt] (1bis) at (16.5,9){};
        \node[rectangle, fill=darkgray, inner sep = 0pt, minimum size=4pt] (2bis) at (16,8){};
        \node[rectangle, fill=darkgray, inner sep = 0pt, minimum size=4pt] (3bis) at (17.75,6.5){};
        
        \node[] (4bis) at (14,5){};
        \node[] (5bis) at (16.1,5){};
        \node[] (6bis) at (18,5){};
        \node[] (7bis) at (19,5){};

		\draw[semithick] plot [smooth, tension=0.7] coordinates { (0bis) (1bis) (17.5,7.75) (3bis) (18.75, 5.75) (7bis)};
		\draw[semithick] plot [smooth, tension=0.7] coordinates { (3bis) (17.75,5.75) (6bis)};
		\draw[semithick] plot [smooth, tension=0.7] coordinates { (1bis) (2bis) (16,7) (16.25,6) (5bis)};
		\draw[semithick] plot [smooth, tension=0.7] coordinates { (2bis) (14.75, 6.5) (4bis)};

        \node[draw, rectangle, fill=darkgray, inner sep = 0pt, minimum size=4pt] () at (16.5,9){};
        \node[draw, rectangle, fill=darkgray, inner sep = 0pt, minimum size=4pt] () at (16,8){};
        \node[draw, rectangle, fill=darkgray, inner sep = 0pt, minimum size=4pt] () at (17.75,6.5){};
        
		\node[draw, rectangle, fill=lightgray, inner sep = 0pt, minimum size=4pt] (pun1bis) at (16,7){};
		\node[inner sep = 0pt, minimum size=0pt] (pun1bisroot) at (15.5,6.5){};
        \draw[-] (15.5,6.5) edge [] node {} (14.75, 5);
        \draw[-] (15.5,6.5) edge [] node {} (15.9, 5);
		\node[] () at (15.35,5.4){$\punStrat{1}$};
        \draw[-stealth] (pun1bis) edge [] node {} (pun1bisroot);
        \draw [draw=gray, fill=lightgray, opacity=0.2]
       (15.5,6.5) -- (14.75, 5) -- (15.9, 5) --  cycle;

        \node[draw, rectangle, fill=lightgray, inner sep = 0pt, minimum size=4pt] (pun2bis) at (17.25,8.1){}; 
		\node[inner sep = 0pt, minimum size=0pt] (pun2bisroot) at (17,7.8){};
        \draw[-] (17,7.8) edge [] node {} (16.35, 5);
        \draw[-] (17,7.8) edge [] node {} (17.75, 5);
		\node[] () at (17,5.95){$\punStrat{2}$};
        \draw[-stealth] (pun2bis) edge [] node {} (pun2bisroot);
        \draw [draw=gray, fill=lightgray, opacity=0.2]
       (17,7.8) -- (16.35, 5) -- (17.75, 5) --  cycle;
       
        \draw[-, thick] (23,10) edge [] node {} (20, 5);
        \draw[-, thick] (23,10) edge [] node {} (26, 5);

        \node[] (0tre) at (22.55,9.9){$\tilde{\sigma}_0$};
        \node[] (0tre) at (23,10){};
        
        \node[rectangle, fill=darkgray, inner sep = 0pt, minimum size=4pt] (1tre) at (23,9){};
        \node[rectangle, fill=darkgray, inner sep = 0pt, minimum size=4pt] (2tre) at (22.5,8){};
        \node[rectangle, fill=darkgray, inner sep = 0pt, minimum size=4pt] (3tre) at (24.25,6.5){};
        
        \node[] (4tre) at (20.5,5){};
        \node[] (5tre) at (22.6,5){};
        \node[] (6tre) at (24.5,5){};
        \node[] (7tre) at (25.5,5){};

		\draw[semithick] plot [smooth, tension=0.7] coordinates { (0tre) (1tre) (24,7.75) (3tre) (25.25, 5.75) (7tre)};
		\draw[semithick] plot [smooth, tension=0.7] coordinates { (3tre) (24.25,5.75) (6tre)};
		\draw[semithick] plot [smooth, tension=0.7] coordinates { (1tre) (2tre) (22.5,7) (22.75,6) (5tre)};
		\draw[semithick] plot [smooth, tension=0.7] coordinates { (2tre) (21.25, 6.5) (4tre)};

        \node[draw, rectangle, fill=darkgray, inner sep = 0pt, minimum size=4pt] () at (23,9){};
        \node[draw, rectangle, fill=darkgray, inner sep = 0pt, minimum size=4pt] () at (22.5,8){};
        \node[draw, rectangle, fill=darkgray, inner sep = 0pt, minimum size=4pt] () at (24.25,6.5){};
        
		\node[draw, rectangle, fill=lightgray, inner sep = 0pt, minimum size=4pt] (pun1tre) at (22.5,7){};
		\node[inner sep = 0pt, minimum size=0pt] (pun1treroot) at (22,6.5){};
        \draw[-] (22,6.5) edge [] node {} (21.25, 5);
        \draw[-] (22,6.5) edge [] node {} (22.4, 5);
		\node[] () at (21.85,5.4){$\punStrat{1}$};
        \draw[-stealth] (pun1tre) edge [] node {} (pun1treroot);
        \draw[draw=gray, fill=lightgray, opacity=0.2]
       (22,6.5) -- (21.25, 5) -- (22.4, 5) --  cycle;

        \node[draw, rectangle, fill=lightgray, inner sep = 0pt, minimum size=4pt] (pun2tre) at (23.75,8.1){}; 
		\node[inner sep = 0pt, minimum size=0pt] (pun2treroot) at (23.5,7.8){};
        \draw[-] (23.5,7.8) edge [] node {} (22.85, 5);
        \draw[-] (23.5,7.8) edge [] node {} (24.25, 5);
		\node[] () at (23.5,5.95){$\punStrat{2}$};
        \draw[-stealth] (pun2tre) edge [] node {} (pun2treroot);
        \draw[draw=gray, fill=lightgray, opacity=0.2]
       (23.5,7.8) -- (22.85, 5) -- (24.25, 5) --  cycle;
        
        \node[circle, color=lightgray, fill=lightgray, inner sep = 0pt, minimum size=3pt] (cycle1) at (24,7.75){};
        \node[ circle, color=lightgray, fill=lightgray, inner sep = 0pt, minimum size=3pt] (cycle2) at (24.125,7){};
        \draw[ultra thick, dotted, draw=white] plot [smooth, tension=0.7] coordinates { (cycle1) (24.09, 7.5)  (cycle2)};
                
        \node[draw, circle, color=lightgray, fill=lightgray, inner sep = 0pt, minimum size=3pt] () at (24,7.75){};
        \node[draw, circle, color=lightgray, fill=lightgray, inner sep = 0pt, minimum size=3pt] () at (24.125,7){};
        
        \node[circle, color=lightgray, fill=lightgray, inner sep = 0pt, minimum size=3pt] (cycle3) at (24.5,6.24){};
        \node[ circle, color=lightgray, fill=lightgray, inner sep = 0pt, minimum size=3pt] (cycle4) at (25,5.95){};
        \draw[ultra thick, dotted, draw=white] plot [smooth, tension=0.7] coordinates { (cycle3) (cycle4)};
        \node[circle, color=lightgray, fill=lightgray, inner sep = 0pt, minimum size=3pt] () at (24.5,6.24){};
        \node[ circle, color=lightgray, fill=lightgray, inner sep = 0pt, minimum size=3pt] () at (25,5.95){};    
        
        \node[circle, color=lightgray, fill=lightgray, inner sep = 0pt, minimum size=3pt] (cycle5) at (22.64,6.5){};
        \node[ circle, color=lightgray, fill=lightgray, inner sep = 0pt, minimum size=3pt] (cycle6) at (22.75,6){};
        \draw[ultra thick, dotted, draw=white] plot [smooth, tension=0.7] coordinates { (cycle5) (22.65,6.45) (22.7,6.35) (cycle6)};
        \node[circle, color=lightgray, fill=lightgray, inner sep = 0pt, minimum size=3pt] (cycle5) at (22.64,6.5){};
        \node[ circle, color=lightgray, fill=lightgray, inner sep = 0pt, minimum size=3pt] (cycle6) at (22.75,6){};
        
        \node[circle, color=lightgray, fill=lightgray, inner sep = 0pt, minimum size=3pt] (cycle7) at (21.25,6.5){};
        \node[ circle, color=lightgray, fill=lightgray, inner sep = 0pt, minimum size=3pt] (cycle8) at (20.955,6){};
        \draw[ultra thick, dotted, draw=white] plot [smooth, tension=0.7] coordinates { (cycle7) (cycle8)};
        \node[circle, color=lightgray, fill=lightgray, inner sep = 0pt, minimum size=3pt] (cycle7) at (21.25,6.5){};
        \node[ circle, color=lightgray, fill=lightgray, inner sep = 0pt, minimum size=3pt] (cycle8) at (20.955,6){};
        
        \node[circle, color=lightgray, fill=lightgray, inner sep = 0pt, minimum size=3pt] (cycle7) at (22.22,7.7){};
        \node[ circle, color=lightgray, fill=lightgray, inner sep = 0pt, minimum size=3pt] (cycle8) at (21.71,7.1){};
        \draw[ultra thick, dotted, draw=white] plot [smooth, tension=0.7] coordinates { (cycle7) (cycle8)};
        \node[circle, color=lightgray, fill=lightgray, inner sep = 0pt, minimum size=3pt] (cycle7) at (22.22,7.7){};
        \node[ circle, color=lightgray, fill=lightgray, inner sep = 0pt, minimum size=3pt] (cycle8) at (21.71,7.1){};
        
        \node[circle, color=lightgray, fill=lightgray, inner sep = 0pt, minimum size=3pt] (cycle9) at (22.872,8.8){};
        \node[ circle, color=lightgray, fill=lightgray, inner sep = 0pt, minimum size=3pt] (cycle10) at (22.662,8.4){};
        \draw[ultra thick, dotted, draw=white] plot [smooth, tension=0.7] coordinates { (cycle9) (cycle10)};
        \node[circle, color=lightgray, fill=lightgray, inner sep = 0pt, minimum size=3pt] (cycle9) at (22.872,8.8){};
        \node[ circle, color=lightgray, fill=lightgray, inner sep = 0pt, minimum size=3pt] (cycle10) at (22.662,8.4){};
        
        
        \node[rectangle, inner sep = 0pt, minimum size=4pt] (loop4) at (25.243, 5.76){};
        \draw[draw=white, fill=white]
       (25.29, 5.645) -- (25.4, 5.645) -- (25.58,4.99)  -- (25.43,4.99) -- cycle;
        \draw[-stealth, semithick, out=-53,in=-120, looseness=10] (loop4) edge [] node {} (loop4);
        
        \node[rotate=20, rectangle, inner sep = 0pt, minimum size=4pt] (loop3) at (24.265, 5.7){};
        \draw[draw=white, fill=white]
       (24.35, 5.6) -- (24.23, 5.6) --  (24.43,4.99) -- (24.6,4.99) -- cycle;
        \draw[-stealth, semithick, out=-85,in=-30, looseness=10] (loop3) edge [] node {} (loop3);
        
        \node[ rotate=-30, rectangle, inner sep = 0pt, minimum size=4pt] (loop2) at (22.757, 5.9){};
        \draw[draw=white, fill=white]
       (22.68, 5.8) -- (22.82, 5.8) -- (22.7,4.99) -- (22.55,4.99) -- cycle;
        \draw[-stealth, semithick, out=-90,in=-160, looseness=10] (loop2) edge [] node {} (loop2);
        
        \node[rectangle, inner sep = 0pt, minimum size=4pt] (loop1) at (20.823, 5.7){};
        \draw[draw=white, fill=white]
       (20.6, 5.6) -- (20.8, 5.6) -- (20.6,4.99) -- (20.4,4.99) -- cycle;
        \draw[-stealth, semithick, rotate=-27, out=-100,in=-30, looseness=10] (loop1) edge [] node {} (loop1);
        
        \node[] () at (20.8,5.2){$c \rho_1$};
        \node[] () at (22.615,5.45){$c \rho_2$};
        \node[] () at (24.5,5.25){$c \rho_3$};
        \node[] () at (25.35,5.2){$c \rho_4$};
    
		\end{tikzpicture}
	}%
	\caption{The creation of strategies $\hat{\sigma}_0$ and $\tilde{\sigma_0}$ from a solution $\sigma_0$ with $\Wit{\sigma_0} = \{\rho_1, \rho_2, \rho_3, \rho_4\}$ viewed as trees.}
	\label{fig:nexptime_strategy}
	\Description{Figure 3. Fully described in the text.}
\end{figure}
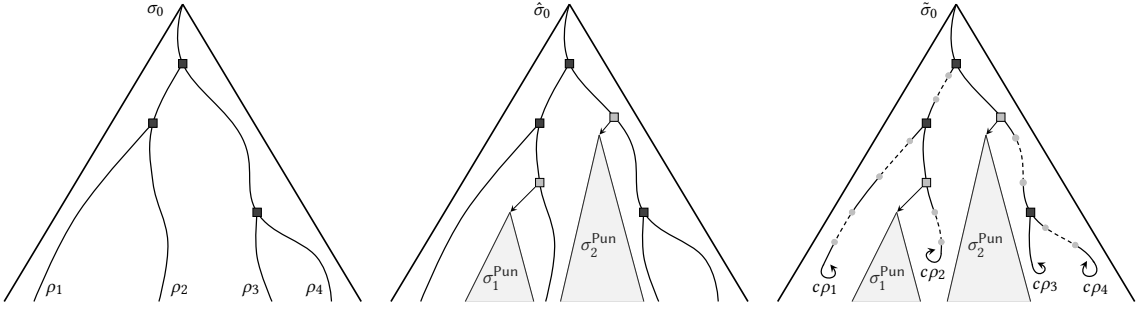 

We now develop the details of the construction of the strategies $\hat{\sigma}_0$ and $\tilde{\sigma}_0$. Figure~\ref{fig:nexptime_strategy} illustrates this construction. It is done in several steps to finally get the proof of Proposition~\ref{prop:tildesigma}. For the rest of this section, we fix an \gameAb{} $\mathcal{G}$ with objectives $\ObjPlayer{0},\ObjPlayer{1},\dots,\ObjPlayer{t}$, a strategy $\sigma_0$ that is solution to the \problemAb{}, a set of witnesses $\Wit{\sigma_0}$ for the \paretoOptimal{} payoffs in $P_{\sigma_0}$, and we write $\dominatedPlays{\sigma_0}$ the set of all plays whose payoff is strictly smaller than some payoff in $\paretoSet{\sigma_0}$.

\subsection{Punishing Strategies}
\label{subsec:punishing}

Before formally defining the punishing strategies, we first define the set of \emph{deviations} $\deviations{\sigma_0}$ as follows:
$$\deviations{\sigma_0} = \{hv \in \Histsigma{\sigma_0} \mid \prefStrat{h} \neq \emptyset \land \prefStrat{hv} = \emptyset \}.$$
As explained above, a deviation is a history that leaves the set of witnesses $\Wit{\sigma_0}$ (by a move of Player~$1$).

Second, we establish the existence of canonical forms for punishing strategies. We potentially need an exponential number of them for reachability and safety objectives and a polynomial number of them for prefix-independent objectives. In all cases, each punishing strategy has a size which can be bounded exponentially. The existence of those strategies are direct consequences of the following lemmas.

\begin{lemma}[Prefix-Independent Objectives] 
\label{lem:parity-case}
Let $v \in V$ be such that there exists $hv \in \deviations{\sigma_0}$. 

\noindent 
Then there exists a finite-memory strategy $\punStrat{v}$ such that for all deviations $hv \in \deviations{\sigma_0}$, when Player~$0$ plays $\punStrat{v}$ from $hv$, all consistent plays $\rho$ starting in $v$ are such that either $h\rho \in \ObjPlayer{0}$ or $h\rho \in \dominatedPlays{\sigma_0}$. The size of $\punStrat{v}$ is at most exponential in the size of $\mathcal{G}$.
\end{lemma}

\begin{proof}
First, we note that, after a deviation $hv \in \deviations{\sigma_0}$, if Player~$0$ continues to play the strategy $\sigma_0$ from $hv$, then all consistent plays $\rho$ are such that either $\rho \in \ObjPlayer{0}$ or $\rho \in \dominatedPlays{\sigma_0}$ as $\sigma_0$ is a solution to the \problemAb{}. Therefore, we know that Player~$0$ has a punishing strategy for all such deviations $hv$. Second, since we consider prefix-independent objectives, Player~$0$ can use one uniform strategy that only depends on $v$ (and not on $hv$). There exists such a strategy with finite memory that can be constructed as follows. We express the objective $\ObjPlayer{0} \cup \dominatedPlays{\sigma_0}$ as an \emph{explicit Muller} objective~\cite{Horn08} for a zero-sum game played on the arena $G$ from initial vertex $v$. This objective is defined by the set $\{ B \subseteq V \mid \exists \rho \mbox{ such that }\infOcc{\rho}= B \land \rho \in \ObjPlayer{0} \cup \dominatedPlays{\sigma_0} \}$. This exactly encodes the objective of Player~$0$ when he plays the punishing strategy after a deviation $hv$. It is well-known that in explicit Muller zero-sum games, there always exist finite-memory winning strategies with a size exponential in the number $|V|$ of vertices of the arena~\cite{DziembowskiJW97}.
\end{proof}

\begin{lemma}[Reachability and safety] 
\label{lem:reach-case}
Let $v \in V$ and $(w,p) \in \{0,1\} \times \{0,1\}^{\nbrObjectives}$ be such that there exists $hv \in \deviations{\sigma_0}$ with $(\won{hv},\payoff{hv}) = (w,p)$. 

\noindent
Then there exists a finite-memory strategy $\punStrat{(v,w,p)}$ such that for all deviations $hv \in \deviations{\sigma_0}$ with $(\won{hv},\payoff{hv}) = (w,p)$, when Player~$0$ plays $\punStrat{(v,w,p)}$ from $hv$, all consistent plays $\rho$ starting in $v$ are such that either $h\rho \in \ObjPlayer{0}$ or $h\rho \in \dominatedPlays{\sigma_0}$. The size of $\punStrat{(v,w,p)}$ is at most exponential in the size of $\mathcal{G}$.
\end{lemma}
\begin{proof} We begin with the case of reachability objectives. We follow the same reasoning as in the proof of Lemma~\ref{lem:parity-case}, except that reachability objectives are not prefix-independent. We thus need to take into account the set of objectives $\ObjPlayer{i}$ already satisfied along the history $hv$, which is recorded in $(w,p)$. The uniform finite-memory strategy $\punStrat{(v,w,p)}$ that Player~$0$ can use from all deviations $hv$ such that $\won{hv} = w$ and $\payoff{hv} = p$ is constructed as follows. First, notice that if $w = 1$, meaning that objective $\Omega_0$ is already satisfied, then Player~$0$ can play using any memoryless strategy as punishing strategy. Second, if $w = 0$, as done in Subsection~\ref{subsec:fptreach}, we consider the extension of $G$ such that its vertices are of the form $(v', w', p')$ where the $(w', p')$-component keeps track of the objectives that have been satisfied so far and such that its initial vertex is equal to $(v,w,p)$. On this extended arena, we consider the zero-sum game with the objective $\ObjPlayer{0} \cup \dominatedPlays{\sigma_0}$ encoded as the disjunction of a reachability objective ($\ObjPlayer{0}$) and a safety objective ($\dominatedPlays{\sigma_0}$). More precisely, in the extended game, Player~0 has the objective either to reach a vertex in the set $\{ (v', w', p') \mid w' = 1 \}$ or to stay forever within the set of vertices $\{ (v', w', p') \mid \exists p'' \in \paretoSet{\sigma_0} : p' < p'' \}$. It is known, see e.g.~\cite{BruyereHR18}, that there always exist memoryless winning strategies for zero-sum games with an objective which is the disjunction of a reachability objective and a safety objective. Therefore, this is the case here for the extended game, and thus also in the original game however with a winning finite-memory strategy with exponential size. 

We now shift to the case of safety objectives and show how to adapt the first part of the proof. The uniform finite-memory strategy $\punStrat{(v,w,p)}$ that Player~$0$ can use from all deviations $hv$ such that $\won{hv} = w$ and $\payoff{hv} = p$ is constructed as follows, by using the same extension of $G$ with vertices of the form $(v', w', p')$. First, if $w = 0$, meaning that safety objective $\Omega_0$ is not satisfied, Player~$0$ needs to ensure that the consistent plays with $hv$ as a prefix have a payoff strictly smaller than a payoff of $\paretoSet{\sigma_0}$. To do so, on the extended arena, we consider the zero-sum game with the objective $\dominatedPlays{\sigma_0}$ encoded as a reachability objective: reaching the set of vertices $\{ (v', w', p') \mid \exists p'' \in \paretoSet{\sigma_0} : p' < p'' \}$. This is sound as in the case of safety objectives, once the payoff of a history $h$ is strictly smaller than a payoff $p$, any play with $h$ as a prefix also has a payoff strictly smaller than $p$. Second, if $w = 1$, meaning that the safety objective of Player~$0$ is satisfied \emph{so far}, Player~$0$ needs to ensure that the consistent plays with $hv$ as a prefix have a payoff strictly smaller than a payoff of $\paretoSet{\sigma_0}$ or satisfy the objective of Player~$0$. This is encoded as the disjunction of a safety objective ($\ObjPlayer{0}$) and a reachability objective ($\dominatedPlays{\sigma_0}$) as previously explained. The results mentioned in the first part proof yield the desired memory bound.
\end{proof}

If we systematically change within $\sigma_0$ the behavior of Player~$0$ after a deviation from $\Wit{\sigma_0}$, and use the punishing strategies as defined in the proofs of Lemmas~\ref{lem:parity-case} and \ref{lem:reach-case}, we obtain a new strategy $\hat{\sigma}_0$ that is solution to the \problemAb{}. The total size of the punishing finite-memory strategies in $\hat{\sigma}_0$ is at most exponential in the size of $\mathcal{G}$. To obtain our results, it remains to show how to compact the plays in $\Wit{\sigma_0}$. To that end, we study the histories and plays within $\Wit{\sigma_0}$. 

\subsection{Compacting Witnesses}
\label{subsec:compacting}

We now show how to compact the set of witnesses in a way to produce them with a finite-memory strategy. Together with the punishing strategies this will lead to a solution $\tilde{\sigma}_0$ to \problemAb{} with a memory of exponential size as announced in Proposition~\ref{prop:tildesigma}. We first consider reachability objectives, then explain the small modifications required for dealing with safety objectives. We finish by explaining how to modify the construction for prefix-independent objectives.

Consider reachability \gamesAb{}. Given a history $h$ that is prefix of at least one witness in $\Wit{\sigma_0}$, we call \emph{\region} and we denote by $\Reg{h}$ the tuple $\Reg{h} = (\won{h},\payoff{h},\prefStrat{h})$.
We also use notation $R = (w,p,\W)$ for a \region. Given a witness $\rho = v_0v_1 \ldots \in \Wit{\sigma_0}$, we consider $\rho^* = (v_0,R_0) (v_1,R_1) \ldots$ such that each $v_j$ is extended with the \region\ $R_j = (w_j,p_j,\W_j) = \Reg{v_0v_1 \dots v_j}$. Similarly we define $h^*$ associated with any history $h$ prefix of a witness. The following properties hold for a witness $\rho$ and its corresponding play $\rho^*$:
\begin{itemize}
    \item for all $j \geq 0$, we have $w_j \leq w_{j+1}$, $p_j \leq p_{j+1}$, and $\W_j \supseteq \W_{j+1}$,
    \item the sequence $(w_j,p_j)_{j\geq 0}$ eventually stabilizes on $(w,p)$ equal to the extended payoff $( \won{\rho}, \payoff{\rho})$ of $\rho$, 
    \item the sequence $(\W_j)_{j\geq 0}$ eventually stabilizes on a set $\W$ which is a singleton such that $\W = \{\rho\}$.
\end{itemize}

Thanks to the previous properties, each $\rho \in \Wit{\sigma_0}$ can be \emph{\region\ decomposed} into a sequence of paths $\pi[1]\pi[2]\cdots\pi[k]$ where the corresponding decomposition $\pi^*[1]\pi^*[2]\cdots\pi^*[k]$ of $\rho^*$ is such that for each $\ell$: \emph{(i)} the \region\ is constant along the path $\pi^*[\ell]$ and \emph{(ii)} it is distinct from the \region\ of the next path $\pi^*[\ell + 1]$ (if $\ell < k$). Each $\pi[\ell]$ is called a \emph{\sect} of $\rho$, such that it is \emph{\internal} (resp.\ \emph{\final}) if $\ell < k$ (resp.\ $\ell = k$).

Notice that the number of regions that are traversed by $\rho$ is bounded by 
\begin{eqnarray} \label{eq:traversed}
(\nbrObjectives + 2) \cdot |\Wit{\sigma_0}|.
\end{eqnarray}
Indeed, along $\rho$, the first two components $(w,p)$ of a region correspond to a monotonically increasing vector of $\nbrObjectives + 1$ Boolean values (from $(0,(0,\ldots,0))$ to $(1,(1,\ldots,1))$ in the worst case), and the last component  $W$ is a monotonically decreasing set of witnesses (from $\Wit{\sigma_0}$ to $\{\rho\}$ in the worst case).  So the number of regions traversed by a witness is bounded exponentially in the size of the game $\mathcal{G}$ (recall that $|\Wit{\sigma_0}|$ is exponential in $\nbrObjectives$ in the worst case).

We have the following important properties for the \sect s of the witnesses of $\Wit{\sigma_0}$.
\begin{itemize}
    \item Let $\rho,\rho' \in \Wit{\sigma_0}$, with region decompositions $\rho = \pi[1]\cdots\pi[k]$ and $\rho' = \pi'[1]\cdots\pi'[k']$ and let $h$ be the longest common prefix of $\rho$ and $\rho'$. Then there exists $k_1 < k,k'$ such that $h = \pi[1]\cdots\pi[k_1]$, $\pi[\ell] = \pi'[\ell]$ for all $\ell \in \{1,\ldots,k_1\}$ and $\pi[k_1+1] \neq \pi'[k_1+1]$. Therefore, when $\Wit{\sigma_0}$ is seen as a tree, the branching structure of this tree is respected by the \sect s. 
    \item Let $R = (w,p,\W)$ be a region and consider the set of all histories $h$ such that $\Reg{h} = R$. Then all these histories are prefixes of each other and are prefixes of exactly $|W|$ witnesses (as $\prefStrat{h} = \W$ for each such $h$). Therefore, the branching structure of $\Wit{\sigma_0}$ is respected by the \sect s such that the associated \region s are all pairwise distinct. The latter property is called the \emph{\Witproperty} of $\Wit{\sigma_0}$.
\end{itemize}
\noindent
We consider a \emph{compact} version $c \Wit{\sigma_0}$ of $\Wit{\sigma_0}$ defined as follows:
\begin{itemize}
    \item each \internal\ \sect\ $\pi$ of $\Wit{\sigma_0}$ is replaced by the elementary path $c \pi$ obtained by eliminating all the cycles of $\pi$. Each \final\ \sect\ $\pi$ of $\Wit{\sigma_0}$ is replaced by a lasso $c \pi = \pi'_1(u\pi'_2)^\omega$ such that $u$ is a vertex, $\pi'_1u\pi'_2$ is an elementary path, and $\pi'_1u\pi'_2u$ is prefix of $\pi$.
    \item each witness $\rho$ of $\Wit{\sigma_0}$ with \region\ decomposition $\rho = \pi[1]\cdots\pi[k]$ is replaced by $c \rho = c \pi[1] \cdots c \pi[k]$ such that each $\pi[\ell]$ is replaced by $c \pi[\ell]$. Notice that as the \region\ is constant inside the \sect s, the \region\ decomposition of $c \rho$ coincide with the sequence of its $c \pi[\ell]$, $\ell \in \{1, \ldots, k\}$.
\end{itemize}
Therefore, by construction of the compact witnesses, the \Witproperty\ of $\Wit{\sigma_0}$ is kept by the set $\{ c\rho \mid \rho \in \Wit{\sigma_0}\}$ and we have for each $c\rho \in c\Wit{\sigma_0}$,
\begin{eqnarray} \label{eq:payoffkept}
(\won{c \rho},\payoff{c \rho}) =  (\won{\rho},\payoff{\rho}).
\end{eqnarray}
 
We then construct the announced strategy $\tilde{\sigma}_0$ that produces the set $c \Wit{\sigma_0}$ of compact witnesses and after any deviation acts with the adequate punishing strategy (as mentioned in Lemma~\ref{lem:reach-case}). More precisely, let $gv$ be such that $g$ is prefix of a compact witness and $gv$ is not (Player~$1$ deviates from $c \Wit{\sigma_0}$). Then by definition of the compact witnesses, there exists a deviation $hv$ such that $(\won{gv},\payoff{gv}) = (\won{hv},\payoff{hv}) = (w,p)$. Then from $gv$ Player~0 switches to the punishing strategy $\punStrat{(v,w,p)}$.

\begin{lemma}
\label{lem:reach-correct} 
The strategy $\tilde{\sigma}_0$ is a solution to the \problemAb{} for reachability \gamesAb{} and its size is bounded exponentially in the size of the game $\mathcal{G}$.
\end{lemma}

\begin{proof}
Let us first prove that $\tilde{\sigma}_0$ is a solution to the \problemAb{}. \emph{(i)} By (\ref{eq:payoffkept}), the set of extended payoffs of plays in $c \Wit{\sigma_0}$ is equal to the set of extended payoffs of witnesses in $\Wit{\sigma_0}$. This means that with $c \Wit{\sigma_0}$, we keep the same set $\paretoSet{\sigma_0}$ and the objective $\Omega_0$ is satisfied along each compact witness. \emph{(ii)} The punishing strategies used by $\tilde{\sigma}_0$ guarantee the satisfaction of the objective $\ObjPlayer{0} \cup \dominatedPlays{\sigma_0}$ by Lemma~\ref{lem:reach-case}. Therefore $\tilde{\sigma}_0$ is a solution to the \problemAb{}.

Let us now show that the memory size of $\tilde{\sigma}_0$ is bounded exponentially in the size of $\mathcal{G}$. \emph{(i)} By Lemma~\ref{lem:reach-case}, each punishing strategy used by $\tilde{\sigma}_0$ is of exponential size and the number of punishing strategies is exponential. \emph{(ii)} To produce the compact witnesses, $\tilde{\sigma}_0$ keeps in memory the current region and produces in a memoryless way the corresponding compact section (which is an elementary path or lasso). Thus the required memory size for producing $c \Wit{\sigma_0}$ is the number of regions. By (\ref{eq:traversed}), every play in $c \Wit{\sigma_0}$ traverses at most an exponential number of regions and there is an exponential number of such plays (equal to $|\Wit{\sigma_0}|$).
\end{proof}

Let us now consider safety \gamesAb{}. The construction of $\tilde{\sigma}_0$ is nearly the same as for reachability \gamesAb{}. The minor differences are the following ones. Given a witness $\rho = v_0v_1 \ldots$ and its corresponding play $\rho^* = (v_0,R_0) (v_1,R_1) \ldots$, we have $w_j = 1$ and $p_j \geq p_{j+1}$ for all $j$ due to the safety objectives. The number of regions that are traversed by $\rho$ is bounded by $(\nbrObjectives + 1) \cdot |\Wit{\sigma_0}|$ (instead of (\ref{eq:traversed}), by a monotonically decreasing vector of $\nbrObjectives$ Boolean values (from $(1,(1,\ldots,1))$ to $(1,(0,\ldots,0))$ in the worst case). The strategy $\tilde{\sigma}_0$ is constructed exactly as for reachability \gamesAb{}. We get the next corollary to Lemma~\ref{lem:reach-correct}.

\begin{corollary}
\label{cor:safety-correct}
The strategy $\tilde{\sigma}_0$ is a solution to the \problemAb{} for \gamesAb{} with safety objectives and its size is bounded exponentially in the size of the game $\mathcal{G}$.
\end{corollary}

We finally switch to \gamesAb{} with prefix-independent objectives and state the following lemma whose proof follows the same type of arguments as those given for reachability objectives.

\begin{lemma}
\label{lem:parity-correct}
There exists a strategy $\tilde{\sigma}_0$ that is solution to the \problemAb{} for \gamesAb{} with prefix-independent objectives and whose size is bounded exponentially in the size of the game $\mathcal{G}$.
\end{lemma}
\begin{proof}
We highlight here the main differences from reachability \gamesAb{}.
 \begin{itemize}
      \item We associate to each history $h$ of a play $\rho \in \Wit{\sigma_0}$ a singleton $\Reg{h} = \prefStrat{h}$ instead of the triplet $(\won{h},\payoff{h},\prefStrat{h})$ used in the case of reachability. This is because $(\won{h},\payoff{h})$ does not make sense for prefix-independent objectives. The number of regions traversed by a witness $\rho$ is thus bounded by $|\Wit{\sigma_0}|$.
      \item For the definition of the compact witnesses, we proceed identically as for reachability by simply removing cycles inside each section with the exception of terminal sections. Given the terminal section $\pi[k]$ of a witness $\rho \in \Wit{\sigma_0}$, we replace it by a lasso $c \pi[k] = \pi'_1(\pi'_2)^\omega$ such that $c\pi[k]$ and $\pi[k]$ start at the same vertex, $\occ{c \pi[k]} = \occ{\pi[k]}$, $\infOcc{c \pi[k]} = \infOcc{\pi[k]}$, and $|\pi'_1\pi'_2|$ is quadratic in $|V|$~\cite[Proposition 3.1]{BouyerBMU15}. Therefore, by construction, the objectives $\Omega_i$ satisfied by a witness $\rho$ are exactly the same as for its corresponding compact play $c\rho$.
  \end{itemize} 

The required strategy $\tilde{\sigma}_0$ produces the set $c \Wit{\sigma_0}$ of compact witnesses. After any deviation it acts with the adequate punishing strategy as follows. Let $gv$ be a deviation from $c \Wit{\sigma_0}$, then there exists a deviation $hv$ from $\Wit{\sigma_0}$, and from $gv$ Player~0 switches to the punishing strategy $\punStrat{v}$ of Lemma~\ref{lem:parity-case}.
\end{proof}

With Lemmas~\ref{lem:reach-correct},~\ref{lem:parity-correct} and Corollary~\ref{cor:safety-correct}, we have thus proved Proposition~\ref{prop:tildesigma}.

\subsection{Verifying the Correctness of the Guessed Solution}

We know by Proposition~\ref{prop:tildesigma} that when there exists a solution to the \problemAb{}, there is one that is finite-memory and has a memory size exponential in the size of the game. We are thus able to show that the \problemAb{} is in \nexptime{} for \gamesAb{}. The proof is divided into two parts, first for reachability and safety objectives and then for prefix-independent objectives.

\begin{proof}[Proof of Theorem~\ref{thm:nexptime} for reachability and safety objectives.] 
We begin with the case of reachability \gamesAb{}, let $\mathcal{G}$ be such a game. Proposition~\ref{prop:tildesigma} states the existence of solutions to the \problemAb{} for $\mathcal{G}$ that use a finite memory bounded exponentially. Let $\sigma_0$ be such a solution. We can guess it as a Moore machine $\mathcal{M}$ with a set of memory states at most exponential in the size of $\mathcal{G}$. Let us explain how to verify that the guessed solution $\sigma_0$ is a solution to the \problemAb{}, i.e., that every play in $\Playsigmazero$ which is \paretoOptimal{} satisfies the objective $\Omega_0$ of Player~$0$. 
\begin{itemize}
    \item First, we construct the cartesian product $G \times \mathcal{M}$ of the arena $G$ with the Moore machine $\mathcal{M}$ which is a graph of exponential size whose infinite paths (starting from the initial vertex $v_0$ and the initial memory state $m_0$) are exactly the plays consistent with $\sigma_0$. 
    \item Second, to compute $\paretoSet{\sigma_0}$, we test for the existence of a play $\rho$ in $G \times \mathcal{M}$ with a given payoff $p = \payoff{\rho}$, beginning with the largest possible payoff $p = (1,\ldots,1)$ and finishing with the smallest possible one $p = (0,\ldots,0)$. As we begin with the largest payoffs, deciding the existence of a play $\rho$ for payoff $p$ corresponds to deciding the existence of a play $\rho$ that satisfies an intersection of reachability objectives (those $\Omega_i$ such that $p_i = 1$, without considering those $\Omega_i$ such that $p_i = 0$). The latter property can be checked in exponential time as follows. We extend the graph $G \times \mathcal{M}$ with a Boolean vector in $\B^\nbrObjectives$ keeping track of the objectives of Player~$1$ that are already satisfied. The resulting graph is still of exponential size and the intersection of reachability objectives becomes a single reachability objective that can be checked in polynomial time in the size of this graph. Therefore, as there is at most an exponential number of payoffs $p$ to consider, the set $\paretoSet{\sigma_0}$ can be computed in exponential time.
    \item Third, to check that each \paretoOptimal{} play in $\Playsigmazero$ satisfies $\Omega_0$, we test for each $p \in \paretoSet{\sigma_0}$ whether there exists a play that satisfies the objectives $\Omega_i$ such that $p_i = 1$ as well as the objective $\Plays_G \setminus \Omega_0$. As the complement $\Plays_G \setminus \Omega_0$ of $\Omega_0 = \reach{T_0}$ is a safety objective, given $p \in \paretoSet{\sigma_0}$, we remove vertices of $G \times \mathcal{M}$ whose first component belongs to $T_0$ before checking whether there exists a play that satisfies the objectives $\Omega_i$ such that $p_i = 1$. It follows that this third step can also be done in exponential time.
\end{itemize}
As a consequence we have a \nexptime{} algorithm for reachability \gameAb{}s.

We now consider the case of safety \gamesAb{}. The approach is similar to the reachability case, and we therefore only explain the differences. We also use $G \times \mathcal{M}$ as described in the first step. In the second step, for a given payoff $p$, we need to decide whether there exists a play $\rho$ that satisfies an intersection of safety (instead of reachability) objectives which is again a safety objective. This can be checked in polynomial time in the size of $G \times \mathcal{M}$. In the third step, we test for each $p \in \paretoSet{\sigma_0}$ whether there exists a play that satisfies the objectives $\Omega_i$ such that $p_i = 1$ (a safety objective) as well as the objective $\Plays_G \setminus \Omega_0$ (a reachability objective). For each $p$, this can be done in polynomial time in the size of $G \times \mathcal{M}$ as explained for the case of reachability. This yields the desired \nexptime{} algorithm for safety \gameAb{}s.
\end{proof}

The proof of Theorem~\ref{thm:nexptime} for prefix-independent objectives is based on the following lemma that is useful for computing the payoff of plays. Let $\mathcal{G}$ be an \gameAb{} with prefix-independent objectives and $\sigma_0$ be a strategy for Player~$0$. We denote by $\infOcc{\Playsigmazero}$ the set of all $I \subseteq V$ such that there exists $\rho \in \Playsigmazero$ with $\infOcc{\rho} = I$.

\begin{lemma}
\label{lem:ComputeInfOcc}
If $\sigma_0$ is given by a Moore machine $\mathcal{M}$ with a number of memory states at most exponential in $|\mathcal{G}|$, then the set $\infOcc{\Playsigmazero}$ can be computed in time exponential in $|\mathcal{G}|$.
\end{lemma}

\begin{proof}
The required algorithm works as follows. We construct the graph $G \times \mathcal{M}$ and remove all the vertices that are not reachable from the initial vertex $(v_0, m_0)$. The resulting graph is denoted by $G'$. Then, given $I \subseteq V$, \emph{(i)} we construct the subgraph $G'_I$ obtained by only keeping the vertices $(v,m)$ of $G'$ in which $v \in I$, \emph{(ii)} we compute the set of strongly connected components (SCCs) $C$ of $G'_I$. The set $I$ is added to $\infOcc{\Playsigmazero}$ if and only if there exists an SCC $C$ such that $I = \{v \mid (v,m) \in C \}$.

Let us show that this procedure is correct. Let us consider a set $I \subseteq V$ and an SCC $C$ in $G'_I$ as described above. Let us show that there exists a play $\rho \in \Playsigmazero$ such that $\infOcc{\rho} = I$. The vertices in $G'_I$ are reachable from the initial vertex $(v_0, m_0)$ of $G'$ and the SCC $C$ in $G'_I$ contains a vertex $(v, m)$ for each $v \in I$. It follows that there exists an infinite path $\rho'$ in $G'$ starting in $(v_0, m_0)$ with $\infOcc{\rho'} = C$. The play $\rho \in \Playsigmazero$ corresponding to this path $\rho'$ is such that $\infOcc{\rho} = I$ by definition of $C$. Conversely, let us consider a play $\rho\in \Playsigmazero$ with $\infOcc{\rho} = I$. We can consider the corresponding infinite path $\rho'$ in $G'$ starting in $(v_0, m_0)$. Let us consider the set $I' = \infOcc{\rho'}$. It follows that $I'$ contains at least a vertex $(v, m)$ for each $v \in I$. Moreover, $I'$ is reachable from $(v_0, m_0)$ and contained in $G'_{I}$. Consider the SCC $C$ of $G'_{I}$ containing $I'$. Then as $I'$, $C$ contains at least a vertex $(v, m)$ for each $v \in I$. By definition of $G'_I$, $C$ contains no vertex $(v,m)$ with $v \not\in I$. Therefore $I = \{v \mid (v,m) \in C \}$.

It remains to show that the proposed algorithm executes in exponential time in $|\mathcal{G}|$. Computing $G$ and then $G'$ is done in exponential time. There are $2^{|V|}$ possible sets $I \subseteq V$ and for each such set $I$, the steps (\emph{i}) and (\emph{ii}) can be done in polynomial time in the size of $G'$, thus in exponential in $|\mathcal{G}|$.
\end{proof}

\begin{proof}[Proof of Theorem~\ref{thm:nexptime} for prefix-independent objectives.]
The proof for prefix-independent objectives follows the same steps as for reachability objectives. We again guess a solution $\sigma_0$ as a Moore machine $\mathcal{M}$ of exponential size and check whether it is indeed a solution to the \problemAb{} as follows.
\begin{itemize}
    \item First, we construct the graph $G \times \mathcal{M}$. 
    \item Second, we compute the set $\infOcc{\Playsigmazero}$. From this set, we compute the set  $\mathsf{Ext}(\sigma_0) = \{(w, p) \mid \text{ there exists } \rho \in \Playsigmazero \text{ with extended payoff } (w,p) \}$. Indeed recall that the objectives are prefix-independent, meaning that the extended payoff of $\rho$ can be retrieved from $\infOcc{\rho}$. We then compute $\paretoSet{\sigma_0}$ from $\mathsf{Ext}(\sigma_0)$. 
    \item Third, we check for each \paretoOptimal{} payoff $p \in \paretoSet{\sigma_0}$ whether there exists an extended payoff $(0, p)$ in $\mathsf{Ext}(\sigma_0)$.
\end{itemize}
The overall complexity of the previous steps is the following. The first step is performed in exponential time. For the second step, computing $\infOcc{\Playsigmazero}$ is done in exponential time by Lemma~\ref{lem:ComputeInfOcc}, computing $\mathsf{Ext}(\sigma_0)$ amounts to retrieving from each $I \in \infOcc{\Playsigmazero}$ the corresponding extended payoff which can be done in exponential time, and computing $\paretoSet{\sigma_0}$ can also be done in exponential time by comparing payoffs. The third step can be done in polynomial time in the size of $\paretoSet{\sigma_0}$ and $\mathsf{Ext}(\sigma_0)$, which are exponential in $|\mathcal{G}|$. We obtain a \nexptime{} algorithm for \gameAb{}s for prefix-independent objectives.
\end{proof}

\section{\textsf{NP} Membership for B\"uchi Objectives} \label{sec:buchi_np}
We have shown in the previous section that the \problemAb{} is in \nexptime{} for every objective considered in this paper. In this section, we improve this result for B\"uchi objectives by proving the following theorem.
\begin{theorem}
\label{thm:buchi_np_membership}
The \problemAb{} is in \np{} for B\"uchi \gameAb{}.
\end{theorem}
We do so by first showing that the number of witnesses for a strategy which is a solution to the problem in a B\"uchi \gameAb{} can be bounded by the number of vertices in its arena. We then show that given a set of witnesses, we can check in polynomial time whether deviations from these witnesses are correctly punished. We also highlight that in the case of B\"uchi objectives, witnesses can be compacted into finite paths and lassos of polynomial size. Finally, we devise an \np{} algorithm where we guess a set of witnesses, verify that they indeed are a possible set of witnesses for some strategy which is a solution to the problem and check that deviations from those witnesses are properly handled.

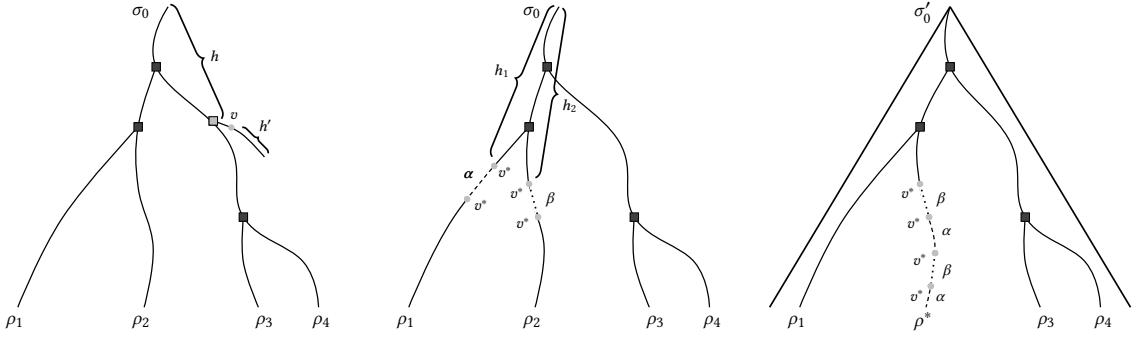
\begin{figure}
	\centering
		\resizebox{\textwidth}{!}{%
		\begin{tikzpicture}
          


        \node[] (0bis) at (9.55,9.9){$\sigma_0$};
        \node[] (0bis) at (10,10){};
        
        \node[rectangle, fill=darkgray, inner sep = 0pt, minimum size=4pt] (1bis) at (9.8,9){};
        \node[rectangle, fill=darkgray, inner sep = 0pt, minimum size=4pt] (2bis) at (9.5,8){};
        \node[rectangle, fill=darkgray, inner sep = 0pt, minimum size=4pt] (3bis) at (11.25,6.5){};
        
        \node[] (4bis) at (7.5,5){};
        \node[] (5bis) at (9.6,5){};
        \node[] (6bis) at (11.5,5){};
        \node[] (7bis) at (12.5,5){};

		\draw[semithick] plot [smooth, tension=0.7] coordinates { (0bis) (1bis) (11,7.75) (3bis) (12.25, 5.75) (7bis)};
		\draw[semithick] plot [smooth, tension=0.7] coordinates { (3bis) (11.25,5.75) (6bis)};
		\draw[semithick] plot [smooth, tension=0.7] coordinates { (1bis) (2bis) (9.5,7) (9.75,6) (5bis)};
		\draw[semithick] plot [smooth, tension=0.7] coordinates { (2bis) (8.25, 6.5) (4bis)};

        \node[draw, rectangle, fill=darkgray, inner sep = 0pt, minimum size=4pt] () at (9.8,9){};
        \node[draw, rectangle, fill=darkgray, inner sep = 0pt, minimum size=4pt] () at (9.5,8){};
        \node[draw, rectangle, fill=darkgray, inner sep = 0pt, minimum size=4pt] () at (11.25,6.5){};
        
        \node[draw, rectangle, fill=lightgray, inner sep = 0pt, minimum size=4pt] (pun2bis) at (10.75,8.1){}; 
        \draw[semithick] plot [smooth, tension=0.7] coordinates { (pun2bis) (11.2,7.9) (11.6,7.5)};
        \node[draw, rectangle, fill=lightgray, inner sep = 0pt, minimum size=4pt] (pun2bis) at (10.75,8.1){}; 
        
        \node[circle, color=lightgray, fill=lightgray, inner sep = 0pt, minimum size=3pt] (pun2bis2) at (11.05,7.99){}; 
        \node[] (trash) at (11.09,8.19){\footnotesize $v$}; 
        \draw [thick,decorate,decoration={brace,amplitude=4pt},xshift=1pt,yshift=1pt] (10,10) -- (10.85,8.1) node [midway,xshift=9pt,yshift=3pt] {\footnotesize $h$};
        
        \draw [thick,decorate,decoration={brace,amplitude=2pt},xshift=2.3pt,yshift=2.3pt] (11.15,7.91) -- (11.58,7.48) node [midway,xshift=5pt,yshift=7pt] {\footnotesize $h'$};
       
        \node[] () at (7.45,4.75){$\rho_1$};
        \node[] () at (9.55,4.75){$\rho_2$};
        \node[] () at (11.6,4.75){$\rho_3$};
        \node[] () at (12.55,4.75){$\rho_4$};



        \node[] (0bis) at (16.05,9.9){$\sigma_0$};
        \node[] (0bis) at (16.5,10){};
        
        \node[rectangle, fill=darkgray, inner sep = 0pt, minimum size=4pt] (1bis) at (16.3,9){};
        \node[rectangle, fill=darkgray, inner sep = 0pt, minimum size=4pt] (2bis) at (16,8){};
        \node[rectangle, fill=darkgray, inner sep = 0pt, minimum size=4pt] (3bis) at (17.75,6.5){};
        
        \node[] (4bis) at (14,5){};
        \node[] (5bis) at (16.1,5){};
        \node[] (6bis) at (18,5){};
        \node[] (7bis) at (19,5){};

		\draw[semithick] plot [smooth, tension=0.7] coordinates { (0bis) (1bis) (17.5,7.75) (3bis) (18.75, 5.75) (7bis)};
		\draw[semithick] plot [smooth, tension=0.7] coordinates { (3bis) (17.75,5.75) (6bis)};
		\draw[semithick] plot [smooth, tension=0.7] coordinates { (1bis) (2bis) (16,7) (16.25,6) (5bis)};
		\draw[semithick] plot [smooth, tension=0.7] coordinates { (2bis) (14.75, 6.5) (4bis)};

        \node[draw, rectangle, fill=darkgray, inner sep = 0pt, minimum size=4pt] () at (16.3,9){};
        \node[draw, rectangle, fill=darkgray, inner sep = 0pt, minimum size=4pt] () at (16,8){};
        \node[draw, rectangle, fill=darkgray, inner sep = 0pt, minimum size=4pt] () at (17.75,6.5){};

        \node[] () at (13.95,4.75){$\rho_1$};
        \node[] () at (16.05,4.75){$\rho_2$};
        \node[] () at (18.1,4.75){$\rho_3$};
        \node[] () at (19.05,4.75){$\rho_4$};
        
        \node[ circle, color=lightgray, fill=lightgray, inner sep = 0pt, minimum size=3pt] (cycle8) at (15.42,7.35){};
        \node[ circle, color=lightgray, fill=lightgray, inner sep = 0pt, minimum size=3pt] (cycle9) at (14.97,6.8){};
        \draw[ultra thick, dotted, draw=white] plot [smooth, tension=0.7] coordinates { (cycle8) (cycle9)};
        \node[] () at (15,7.2){\footnotesize $\alpha$};
        \node[ circle, color=lightgray, fill=lightgray, inner sep = 0pt, minimum size=3pt] (cycle8) at (15.42,7.35){};
        \node[] (trash) at (15.6,7.25){\footnotesize $v^*$}; 
        \node[ circle, color=lightgray, fill=lightgray, inner sep = 0pt, minimum size=3pt] (cycle9) at (14.97,6.8){};
        \node[] (trash) at (15.2,6.7){\footnotesize $v^*$}; 

        \draw [thick,decorate,decoration={brace,mirror,amplitude=4pt},xshift=-2.3pt,yshift=0.4pt] (16.5,10) -- (15.47,7.49) node [midway,xshift=-10pt,yshift=5pt] {\footnotesize $h_1$};
        
        \draw [thick, decorate,decoration={brace,amplitude=4pt},xshift=3pt,yshift=3pt] (16.45,9.87) -- (16.0,7.05) node [midway,xshift=10pt,yshift=-6pt] {\footnotesize $h_2$};
                
        \node[ circle, color=lightgray, fill=lightgray, inner sep = 0pt, minimum size=3pt] (cycle10) at (16.0,7.05){};
        \node[ circle, color=lightgray, fill=lightgray, inner sep = 0pt, minimum size=3pt] (cycle11) at (16.15,6.5){};
        \draw[ultra thick, draw=white] plot [smooth, tension=0.7] coordinates { (15.95,7.05) (cycle11)};
        \draw[thick, dotted] plot [smooth, tension=0.7] coordinates { (cycle10) (cycle11)};
        \node[] () at (15,7.2){\footnotesize $\alpha$};
        \node[ circle, color=lightgray, fill=lightgray, inner sep = 0pt, minimum size=3pt] (cycle10) at (16.0,7.05){};
        \node[] () at (16.35,6.8){\footnotesize $\beta$};
        \node[] (trash) at (15.8,6.9){\footnotesize $v^*$}; 
        \node[ circle, color=lightgray, fill=lightgray, inner sep = 0pt, minimum size=3pt] (cycle11) at (16.15,6.5){};
        \node[] (trash) at (15.9,6.4){\footnotesize $v^*$}; 

        \draw[-, thick] (23,10) edge [] node {} (20, 5);
        \draw[-, thick] (23,10) edge [] node {} (26, 5);

        \node[] (0tre) at (22.55,9.9){$\sigma'_0$};
        \node[] (0tre) at (23,10){};
        
        \node[rectangle, fill=darkgray, inner sep = 0pt, minimum size=4pt] (1tre) at (23,9){};
        \node[rectangle, fill=darkgray, inner sep = 0pt, minimum size=4pt] (2tre) at (22.5,8){};
        \node[rectangle, fill=darkgray, inner sep = 0pt, minimum size=4pt] (3tre) at (24.25,6.5){};
        
        \node[] (4tre) at (20.5,5){};
        \node[] (5tre) at (22.6,5){};
        \node[] (6tre) at (24.5,5){};
        \node[] (7tre) at (25.5,5){};

		\draw[semithick] plot [smooth, tension=0.7] coordinates { (0tre) (1tre) (24,7.75) (3tre) (25.25, 5.75) (7tre)};
		\draw[semithick] plot [smooth, tension=0.7] coordinates { (3tre) (24.25,5.75) (6tre)};
		\draw[semithick] plot [smooth, tension=0.7] coordinates { (1tre) (2tre) (22.5,7) (22.75,6) (5tre)};
		\draw[semithick] plot [smooth, tension=0.7] coordinates { (2tre) (21.25, 6.5) (4tre)};

        \node[draw, rectangle, fill=darkgray, inner sep = 0pt, minimum size=4pt] () at (23,9){};
        \node[draw, rectangle, fill=darkgray, inner sep = 0pt, minimum size=4pt] () at (22.5,8){};
        \node[draw, rectangle, fill=darkgray, inner sep = 0pt, minimum size=4pt] () at (24.25,6.5){};
        
        \node[ circle, color=lightgray, fill=lightgray, inner sep = 0pt, minimum size=3pt] (cycle10) at (22.5,7.05){};
        \node[ circle, color=lightgray, fill=lightgray, inner sep = 0pt, minimum size=3pt] (cycle11) at (22.65,6.5){};
        \draw[ultra thick, draw=white] plot [smooth, tension=0.7] coordinates { (22.5,7.05) (22.61,6.5)};
        \draw[thick, dotted] plot [smooth, tension=0.7] coordinates { (cycle10) (cycle11)};
        \node[ circle, color=lightgray, fill=lightgray, inner sep = 0pt, minimum size=3pt] (cycle10) at (22.5,7.05){};
        \node[ circle, color=lightgray, fill=lightgray, inner sep = 0pt, minimum size=3pt] (cycle11) at (22.65,6.5){};
        \node[] () at (22.85,6.8){\footnotesize $\beta$};
        \node[] (trash) at (22.3,6.9){\footnotesize $v^*$}; 
        \node[] (trash) at (22.4,6.4){\footnotesize $v^*$}; 
        
        \node[ circle, color=lightgray, fill=lightgray, inner sep = 0pt, minimum size=3pt] (cycle10) at (22.65,6.5){};
        \node[ circle, color=lightgray, fill=lightgray, inner sep = 0pt, minimum size=3pt] (cycle11) at (22.75,5.9){};
        \draw[ultra thick, dotted, draw=white] plot [smooth, tension=0.7] coordinates { (22.65,6.5) (22.775,5.9)};
        \node[ circle, color=lightgray, fill=lightgray, inner sep = 0pt, minimum size=3pt] (cycle10) at (22.65,6.5){};
        \node[ circle, color=lightgray, fill=lightgray, inner sep = 0pt, minimum size=3pt] (cycle11) at (22.75,5.9){};
        \node[] () at (22.95,6.25){\footnotesize $\alpha$};
        \node[] (trash) at (22.5,5.8){\footnotesize $v^*$}; 
        
        \node[ circle, color=lightgray, fill=lightgray, inner sep = 0pt, minimum size=3pt] (cycle10) at (22.75,5.9){};
        \node[ circle, color=lightgray, fill=lightgray, inner sep = 0pt, minimum size=3pt] (cycle11) at (22.68,5.35){};
        \draw[ultra thick, draw=white] plot [smooth, tension=0.7] coordinates { (22.75,5.9) (22.7,5.35)};
        \draw[thick, dotted] plot [smooth, tension=0.7] coordinates { (cycle10) (cycle11)};
        \node[ circle, color=lightgray, fill=lightgray, inner sep = 0pt, minimum size=3pt] (cycle10) at (22.75,5.9){};
        \node[ circle, color=lightgray, fill=lightgray, inner sep = 0pt, minimum size=3pt] (cycle11) at (22.68,5.35){};
        \node[] () at (22.95,5.6){\footnotesize $\beta$};
        \node[] (trash) at (22.45,5.2){\footnotesize $v^*$}; 
        
        \draw[ultra thick, dotted, draw=white] plot [smooth, tension=0.7] coordinates { (22.68,5.35) (22.6,4.98)};
        
        \node[ circle, color=lightgray, fill=lightgray, inner sep = 0pt, minimum size=3pt] (cycle11) at (22.68,5.35){};
        
        \node[] () at (22.85,5.15){\footnotesize $\alpha$};

        \node[] () at (20.45,4.75){$\rho_1$};
        \node[] () at (22.55,4.77){$\rho^*$};
        \node[] () at (24.6,4.75){$\rho_3$};
        \node[] () at (25.55,4.75){$\rho_4$};
        
		\end{tikzpicture}
	}%
	\caption{The creation of strategy $\sigma'_0$ from a solution $\sigma_0$ with $\Wit{\sigma_0} = \{\rho_1, \rho_2, \rho_3, \rho_4\}$ viewed as a tree.}
	\label{fig:np_strategy}
	\Description{Figure 3. Fully described in the text.}
\end{figure} 

\subsection{Bounding the Number of Witnesses}
We start by showing the following lemma on the number of witnesses for a strategy which is a solution to the \problemAb{} in a B\"uchi \gameAb{}.
\begin{lemma}
\label{lem:bounding_buchi}
Let $\mathcal{G}$ be a B\"uchi \gameAb{} and let $\sigma_0$ be a solution to the \problemAb{} in $\mathcal{G}$. Then there exists a solution $\sigma'_0$ to the problem in $\mathcal{G}$ such that $|\Wit{\sigma'_0}| \leq |V|$.
\end{lemma}
\begin{proof}
Let $\mathcal{G}$ be a B\"uchi \gameAb{} and let $\sigma_0$ be a solution to the \problemAb{} in $\mathcal{G}$. Let us assume that the number of witnesses of $\sigma_0$ is the smallest possible for a solution in $\mathcal{G}$. Towards contradiction, let us also assume that this number of witnesses is larger than the number of vertices in $\mathcal{G}$, that is $|\Wit{\sigma_0}| > |V|$. It therefore holds that there exist two witnesses $\rho_1$ and $\rho_2$ in $\Wit{\sigma_0}$ such that $\infOcc{\rho_1} \cap \infOcc{\rho_2} \neq \emptyset$. Let $v^* \in V$ be a vertex such that $v^* \in \infOcc{\rho_1}$ and $v^* \in \infOcc{\rho_2}$. Our goal is to create a strategy $\sigma'_0$ which is a solution to the problem, which yields every witness in $\Wit{\sigma_0}$ except for $\rho_2$ that is replaced by a new witness $\rho^*$, and which properly punishes deviations from the witnesses. We will show that the substitution of $\rho_2$ by $\rho^*$ will result in a set of witnesses of $\sigma'_0$ such that $|\Wit{\sigma'_0}| < |\Wit{\sigma_0}|$, contradicting our assumptions on $\sigma_0$.

Strategy $\sigma'_0$ behaves as $\sigma_0$ for every history $h$ which is a prefix of at least one witness that is not $\rho_2$, i.e., $\sigma'_0(h) = \sigma_0(h)$ if $\prefStrat{h} \not \in \{\emptyset, \{\rho_2\} \}$. This allows $\sigma'_0$ to yield every witness except for $\rho_2$. Upon deviation from such witnesses, $\sigma'_0$ behaves as $\sigma_0$. That is, given some history $hvh'$ 
such that $\prefStrat{h} \not \in \{\emptyset, \{\rho_2\} \}$ and $\prefStrat{hv} = \emptyset$, we define $\sigma'_0(hvh') = \sigma_0(hvh')$ (as illustrated in Figure~\ref{fig:np_strategy} left). This allows $\sigma'_0$ to make sure that upon deviation of any prefix of at least one witness that is not $\rho_2$, either the objective of Player~$0$ is satisfied or the payoff of the resulting play is not \paretoOptimal{} (as $\sigma_0$ is a solution to the problem).

We now define $\sigma'_0$ such that it replaces witness $\rho_2$ by a new witness $\rho^*$ which alternates between parts of $\rho_1$ and $\rho_2$. Let $h_1 v^* \alpha v^*$ be the smallest history consistent with $\sigma_0$ such that $\prefStrat{h_1 v^*} = \{\rho_1\}$ and $\infOcc{\rho_1} = \{v^*\} \cup \{v \mid v \in \alpha\}$. That is, $h_1 v^* \alpha v^*$ is the smallest history such that $v^*$ is visited twice and the vertices occurring in $v^* \alpha$ are exactly those visited infinitely often in $\rho_1$. The existence of such a history is guaranteed as $v^* \in \infOcc{\rho_1} \cap \infOcc{\rho_2}$ by definition. Similarly, let $h_2 v^* \beta v^*$ have the same property for $\rho_2$. An example of these histories is illustrated in Figure~\ref{fig:np_strategy} (center). We easily define $\sigma'_0$ to replace $\rho_2$ with a witness $\rho^*$ that alternates between $\beta$ and $\alpha$ every other occurrence of $v^*$. We first make sure that $\sigma'_0$ properly yields history $h_2 v^*$ and punishes deviations from $h_2 v^*$ by behaving as $\sigma_0$ along this history and its deviations. 
From $h_2 v^*$ on, we define $\sigma'_0$ to yield $\beta$ and punish deviations from $\beta$ as done in $\rho_2$. On the next occurrence of $v^*$ after $\beta$ has been completely visited, $\sigma'_0$ switches to yielding $\alpha$ and punishes deviations from $\alpha$ as done in $\rho_1$. This is repeated infinitely often to yield the new witness $\rho^* = h_2 v^* (\beta v^* \alpha v^*)^\omega$ replacing $\rho_2$ (see Figure~\ref{fig:np_strategy} right). Notice that, since B\"uchi objectives are prefix-independent, deviations from each occurrence of $\alpha$ (resp. $\beta$) can be punished as if they were from the first. It is direct to see that $\infOcc{\rho^*} = \infOcc{\rho_1} \cup \infOcc{\rho_2}$ and that therefore the payoff of $\rho^*$ is such that $\payoff{\rho^*}_i = 1$ if and only if $\payoff{\rho_1}_i = 1$ or $\payoff{\rho_2}_i = 1$ for $i \in \{1, \dots, \nbrObjectives\}$ and it holds that $\won{\rho^*} = 1$ (as $\rho_1$ and $\rho_2$ were won by Player~$0$). In particular, we have that $\payoff{\rho_1} < \payoff{\rho^*}$ and  $\payoff{\rho_2} < \payoff{\rho^*}$.

Let us now first show that $|\Wit{\sigma'_0}| \leq |\Wit{\sigma_0}| - 1$ and then that $\sigma'_0$ is a solution to the \problemAb{} in $\mathcal{G}$. Let us establish that the set of witnesses of $\sigma'_0$ is equal to $\Wit{\sigma'_0} = \{\rho \in \Wit{\sigma_0} \mid \payoff{\rho} \not < \payoff{\rho^*} \} \cup \{\rho^*\}$. Notice that $\rho_1, \rho_2$ do not belong to this set as $\payoff{\rho_1} < \payoff{\rho^*}$ and  $\payoff{\rho_2} < \payoff{\rho^*}$, showing that the size of $\Wit{\sigma'_0}$ is strictly smaller than that of $\Wit{\sigma_0}$. The set $\Wit{\sigma'_0}$ is a proper set of witnesses as \emph{(i)} the payoffs of its elements are pairwise incomparable due to the removal of the witnesses of $\Wit{\sigma_0}$ whose payoff is strictly smaller than that of $\rho^*$, and \emph{(ii)} the witnesses in $\Wit{\sigma'_0}$ are won by Player~$0$, as those which also belong to $\Wit{\sigma_0}$ do and given our previous observations on $\rho^*$. It also holds that $\sigma'_0$ is a solution, that is, upon deviation from any witness in $\Wit{\sigma'_0}$, either the objective of Player~$0$ is satisfied or the payoff of the resulting play is strictly smaller than the payoff of some witness in $\Wit{\sigma'_0}$. Notice that by construction, strategy $\sigma'_0$ punishes deviations the same ways as $\sigma_0$ (in particular for $\rho^*$ because B\"uchi objectives are prefix-independent objectives). Assume that a deviating play does not satisfy the objective of Player~$0$. Since $\sigma'_0$ punishes deviations as done in $\sigma_0$ it follows that the payoff of this play is smaller than that of some witness $\rho$ in $\Wit{\sigma_0}$, and thus smaller than that of some witness $\rho'$ in $\Wit{\sigma'_0}$.
\end{proof}

\subsection{Compacting Witnesses}

We have established that if there is a solution to a B\"uchi \gameAb{}, we can construct one such that its number of witnesses is bounded by $|V|$. We now show that these witnesses can be compacted to have a size polynomial in $|V|$. 

\begin{lemma}
\label{lem:buchi_compact}
Given the set $\Wit{\sigma_0}$ of witnesses for a solution $\sigma_0$ to the problem such that $|\Wit{\sigma_0}| \leq |V|$, each witness $\rho$ can be expressed in the form $h\ell^\omega$ such that $h$ and $\ell$ have a size polynomial in $|V|$.
\end{lemma}
\begin{proof}
We proceed as follows.  
\begin{itemize}
      \item As explained in Subsection \ref{subsec:compacting} for prefix-independent objectives (see in particular the proof of Lemma~\ref{lem:parity-correct}), the number of regions traversed by a witness $\rho$ is bounded by $|\Wit{\sigma_0}|$.
      \item The compacted witnesses are created as done in Subsection \ref{subsec:compacting} by removing cycles inside each internal section with the exception of terminal sections (see again the proof of Lemma~\ref{lem:parity-correct}). These compacted internal sections thus have a size that is bounded by $|V|$. The terminal section of a witness $\rho \in \Wit{\sigma_0}$ is replaced by a lasso which retains the same set of vertices occurring (in)finitely often, with a size being quadratic in $|V|$~\cite[Proposition 3.1]{BouyerBMU15}. It follows that compacted witnesses retain the same payoff.
  \end{itemize} 
Therefore, since $|\Wit{\sigma_0}| \leq |V|$, the size of any compacted witness is at most quadratic in $|V|$.
\end{proof}

\subsection{Punishing Strategies}

Before describing an \np{} algorithm to solve the \problemAb{} for B\"uchi \gamesAb{} in the next section, we first need to discuss about punishing strategies. Recall that we explained in Subsection \ref{subsec:punishing} how punishing strategies can be \emph{synthesized} to punish deviations from the witnesses yielded by some strategy $\sigma_0$. We here proceed differently by studying the complexity of \emph{verifying} whether such strategies exist in the following way. Assume that we are given a finite set $W$ of plays and its related set of payoffs $P = \{\payoff{\rho} \mid \rho \in W\}$ (our \np{} algorithm will guess such a set $W$ and check that it is a set of witnesses for some solution $\sigma_0$). Let $hv$ be a history such that $h$ is prefix of some play in $W$ but $hv$ is not ($hv$ can be seen as a deviation from $W$). We want to check whether Player~$0$ has a (punishing) strategy $\sigma_0$ such that all consistent plays $\rho$ starting in $v$ are such that either $h\rho \in \ObjPlayer{0}$ or $h\rho \in \Omega^{<}(P)$. This can be done in polynomial time when $|W| \leq |V|$ as stated in the following lemma.

\begin{lemma}
\label{lem:buchi_punish}
Let $W$ be a finite set containing at most $|V|$ plays and $P$ be its corresponding set of payoffs. Let $hv$ be a history such that $h$ is prefix of some play in $W$ but $hv$ is not. We can check in polynomial time whether there exists a strategy $\sigma_0$ for Player~$0$ such that all consistent plays $\rho$ starting in $v$ are such that either $h\rho \in \ObjPlayer{0}$ or $h\rho \in \Omega^{<}(P)$. 
\end{lemma}

\begin{proof}
For the given history $hv$ we want to verify whether there exists a strategy $\sigma_0$ such that every consistent play $\rho$ starting in $v$ satisfies either $\rho \in \ObjPlayer{0}$ or $h\rho \in \Omega^{<}(P)$. This amounts to verifying the existence of a winning strategy $\sigma_0$ for Player~$0$ in a zero-sum game played on the same arena and where he has objective $\ObjPlayer{0} \cup \Omega^{<}(P)$. We express this objective using a Boolean combination of B\"uchi and co-B\"uchi objectives as follows: 
\begin{eqnarray} \label{eq:formula}
\Buchi{B_0} \lor \bigvee_{p \in P} \Big{(} \bigwedge_{p_i = 0} \CoBuchi{B_i} \land \bigvee_{p_j = 1} \CoBuchi{B_j} \Big{)}
\end{eqnarray}
using the fact that the complement of a B\"uchi objective $\Buchi{B}$ is the co-B\"uchi objective $\CoBuchi{B}$. Indeed, given a play $\rho$ satisfying objective (\ref{eq:formula}), it holds that either \emph{(i)} the objective of Player~$0$ is satisfied or \emph{(ii)} there exists a payoff $p \in P$ such that the objectives not satisfied in $p$ are not satisfied in $\rho$ and $\rho$ also does not satisfy at least one of the objectives satisfied in $p$. It follows that the payoff of $\rho$ is strictly smaller than $p$. Using the fact that the conjunction of co-B\"uchi objectives is again a co-B\"uchi objective, we obtain the equivalent objective 
$$\Buchi{B_0} \lor \bigvee_{p \in P} \Big{(} \CoBuchi{\bigcup_{p_i = 0} B_i} \land \bigvee_{p_j = 1} \CoBuchi{B_j} \Big{)}$$
which we rewrite to obtain the following: 
$$\Buchi{B_0} \lor \bigvee_{p \in P} \bigvee_{p_j = 1} \Big{(} \CoBuchi{B_j \cup \bigcup_{p_i = 0} B_i}\Big{)}.$$
This corresponds to the negation of an objective $\Obj'$ which is the conjunction between one co-B\"uchi and one generalized B\"uchi objective. Verifying whether there exists a strategy for Player~$0$ ensuring objective $\ObjPlayer{0} \cup \Omega^{<}(P)$ from $v$ thus amounts to deciding whether Player~$1$ has a winning strategy for objective $\Obj'$ (because of the determinacy of zero-sum games with such an objective) from $v$. The complexity of checking the existence of the latter strategy is in $\mathcal{O}(|V|^3 \cdot k)$ with $k$ the number of objectives in the generalized B\"uchi objective (see e.g. \cite{BruyereHR16}). In this case, this number $k$ is bounded by $|P| \cdot \nbrObjectives$ and thus by $|V| \cdot \nbrObjectives$ as $|W| \leq |V|$ by hypothesis. This check exhibits the announced polynomial complexity.
\end{proof}

\subsection{\textsf{NP} Algorithm for B\"uchi SP Games}
We finally prove Theorem~\ref{thm:buchi_np_membership} by describing an \np{} algorithm to solve the \problemAb{} for  B\"uchi \gamesAb{}. First, we non-deterministically guess a set of witnesses. As we have established in Lemma \ref{lem:bounding_buchi}, the number of such witnesses required in a solution to the problem in a B\"uchi \gameAb{} is bounded by $|V|$. We assume that the witnesses we guess are compacted, as we have shown in Lemma~\ref{lem:buchi_compact}, and they are thus guessed as a finite word $h\ell$ for each witness $h\ell^\omega$. 
Therefore, the size of each witness is at most quadratic in $|V|$ and as we guess at most $|V|$ witnesses, this step is in polynomial time. 

Second, we make sure that the set of witnesses we guessed is proper. It must be the case that these witnesses can actually be yielded by some strategy of Player~$0$, that is, avoiding cases where two witnesses starting with a same history $hv$ branch to different successors of $v$ when $v\in V_0$. This amounts to verifying that given two witnesses, their longest common prefix ends with a vertex of Player~$1$. 
This step is in polynomial time as the number of pairs of witnesses is bounded by $|V|^2$ and given the polynomial size of the witnesses. We then compute the payoff of each witness (which are of the form $h \ell^\omega$) in polynomial time by checking for each objective $\Buchi{B}$ whether some vertex of $B$ occurs in $\ell$. We also check that Player~$0$ wins in every witness with the same technique. Finally, we check that the payoffs of the witnesses are pairwise incomparable, which can also be done in polynomial time. 

Now that we have guessed a set of witnesses and checked that it is proper, it suffices to verify that deviations from these witnesses can be properly punished. If this is the case, we have shown the existence of a strategy which is a solution to the problem (consisting of a part yielding the witnesses and of a part consisting of punishing strategies). For every $v \in V$ such that there exists a deviation $hv$ from the witnesses, we use Lemma~\ref{lem:buchi_punish} to check whether some punishing strategy exists from $v$. This is done in polynomial time for each such vertex $v$.

\section{\textsf{NEXPTIME}-Hardness}
\label{sec:nexptime_hard}

In this section, we show the \nexptime{}-hardness of solving the \problemAb{} in \gamesAb{} for each kind of objectives, except for B\"uchi and co-B\"uchi objectives for which we show the \np{}-hardness. Given Theorem~\ref{thm:nexptime}, the \problemAb{} is therefore \nexptimeComplete{} for most \gamesAb{}, and given Theorem~\ref{thm:buchi_np_membership} it is \npComplete{} for B\"uchi \gamesAb{} (see Table~\ref{table:comp_summary}). We start this section with a proof of the \np-completeness of the \problemAb{} for reachability objectives and tree arenas, which is a simpler setting, to introduce the intuition behind our reasoning for general arenas. In the subsequent subsections, we prove the \nexptime-hardness first for reachability objectives, then for safety objectives, and finally for prefix-independent objectives other than (co-)B\"uchi.  

\begin{theorem}
\label{thm:nexptimehard}
 The \problemAb{} is \nexptimeHard{} for \gamesAb{}, except for B\"uchi and co-B\"uchi objectives for which it is \np{}-hard. The complexity results for each family of objectives is summarized in Table~\ref{table:comp_summary}.
\end{theorem}

\subsection{\textsf{NP}-Completeness for Reachability SP Games Played on Tree Arenas}
Before turning to the \nexptime-hardness of the \problemAb{} in the next section, we first want to show that the \problemAb{} is already \npComplete{} in the simple setting of reachability objectives and arenas that are trees. To do so, we use a reduction from the \emph{\setCover{}} (\setCoverAb{}) which is \npComplete{}~\cite{Karp72}.
 
\begin{theorem} \label{thm:npcomplete}
    The \problemAb{} is \npComplete{} for reachability \gamesAb{} on tree arenas. 
\end{theorem}

The proof of this theorem is obtained as the consequence of the following arguments and of Proposition~\ref{prop:tree_nphard} below. Notice that when the game arena is a tree, it is easy to design an algorithm for solving the \problemAb{} that is in \np{}. First, we nondeterministically guess a strategy $\sigma_0$ that can be assumed to be memoryless as the arena is a tree. Second, we apply a depth-first search algorithm from the root vertex which accumulates to leaf vertices the extended payoff of plays which are consistent with $\sigma_0$. Finally, we check that $\sigma_0$ is a solution.

Let us explain why the \problemAb{} is \npHard{} on tree arenas by reduction from the \setCoverAb{}. We recall that an instance of the \setCoverAb{} is defined by a set $C = \{e_1,e_2, \dots, e_n\}$ of $n$ elements, $m$ subsets $S_1, S_2, \dots, S_m$ such that $S_i \subseteq C$ for each $i\in \{1, \dots, m\}$, and an integer $\problemParam \leq m$. The problem consists in finding $\problemParam$ indexes $i_1, i_2, \dots, i_\problemParam$ such that the union of the corresponding subsets equals $C$, i.e., $C = \bigcup\limits_{j = 1}^\problemParam S_{i_j}$.

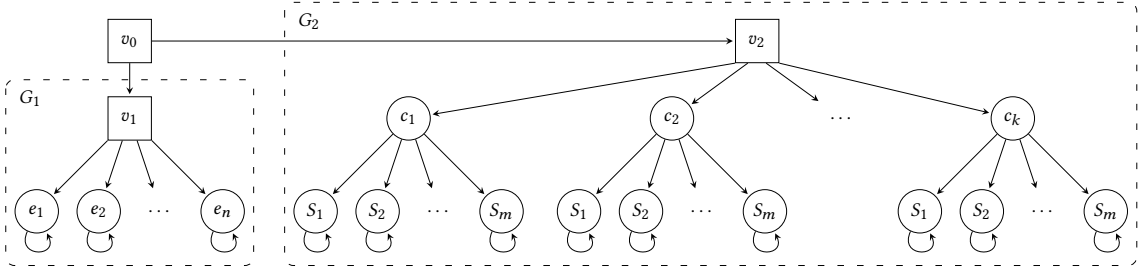
\begin{figure}
	\centering
	\resizebox{\textwidth}{!}{%
		
		\begin{tikzpicture}
		
		\draw[loosely dashed,  rounded corners] (3,10.625) rectangle (7, 7.625) {};
		\node[] at (3 + 0.4,10.625 - 0.3) {$G_1$};
		
		\node[draw, rectangle, minimum size=0.7cm, inner sep = 0.5pt] (v1) at (5,10){$v_1$};
		\node[draw, circle, minimum size=0.7cm, inner sep = 0.5pt] (1) at (3.5,8.5){$e_1$};
		\node[draw, circle, minimum size=0.7cm, inner sep = 0.5pt] (2) at (4.5,8.5){$e_2$};
		\node[minimum size=0.7cm, circle, inner sep = 0.5pt] (dots) at (5.5,8.5){$\dots$};
		\node[draw, circle, minimum size=0.7cm, inner sep = 0.5pt] (n) at (6.5,8.5){$e_n$};
		
		\draw[-stealth, shorten >=1pt, auto] (v1) edge [] node {} (1);
		\draw[-stealth, shorten >=1pt, auto] (v1) edge [] node {} (2);
		\draw[-stealth, shorten >=1pt, auto] (v1) edge [] node {} (dots);
		\draw[-stealth, shorten >=1pt, auto] (v1) edge [] node {} (n);

		\draw[-stealth, shorten >=1pt, auto, in=295, out=245, looseness=4] (1) edge [] node {} (1);
		\draw[-stealth, shorten >=1pt, auto, in=295, out=245, looseness=4] (2) edge [ ] node {} (2);
		\draw[-stealth, shorten >=1pt, auto, in=295, out=245, looseness=4] (n) edge [ ] node {} (n);
		
        
		\draw[loosely dashed,  rounded corners] (7.5,11.875) rectangle (21.25, 7.625) {};
		\node[] at (7.5 + 0.4,11.875 - 0.3) {$G_2$};
        
		\node[draw, circle, minimum size=0.7cm, inner sep = 0.5pt] (c1) at (9.5,10){$c_1$};
		\node[draw, circle, minimum size=0.7cm, inner sep = 0.5pt] (S1c1) at (8,8.5){$S_1$};
		\node[draw, circle, minimum size=0.7cm, inner sep = 0.5pt] (S2c1) at (9,8.5){$S_2$};
		\node[minimum size=0.7cm, circle, inner sep = 0.5pt] (dotsc1) at (10,8.5){$\dots$};
		\node[draw, circle, minimum size=0.7cm, inner sep = 0.5pt] (Smc1) at (11,8.5){$S_m$};
		
		\draw[-stealth, shorten >=1pt, auto] (c1) edge [] node {} (S1c1);
		\draw[-stealth, shorten >=1pt, auto] (c1) edge [] node {} (S2c1);
		\draw[-stealth, shorten >=1pt, auto] (c1) edge [] node {} (dotsc1);
		\draw[-stealth, shorten >=1pt, auto] (c1) edge [] node {} (Smc1);

		\draw[-stealth, shorten >=1pt, auto, in=295, out=245, looseness=4] (S1c1) edge [] node {} (S1c1);
		\draw[-stealth, shorten >=1pt, auto, in=295, out=245, looseness=4] (S2c1) edge [ ] node {} (S2c1);
		\draw[-stealth, shorten >=1pt, auto, in=295, out=245, looseness=4] (Smc1) edge [ ] node {} (Smc1);


		\node[draw, circle, minimum size=0.7cm, inner sep = 0.5pt] (c2) at (13.75,10){$c_2$};
		\node[draw, circle, minimum size=0.7cm, inner sep = 0.5pt] (S1c2) at (12.25,8.5){$S_1$};
		\node[draw, circle, minimum size=0.7cm, inner sep = 0.5pt] (S2c2) at (13.25,8.5){$S_2$};
		\node[minimum size=0.7cm, circle, inner sep = 0.5pt] (dotsc2) at (14.25,8.5){$\dots$};
		\node[draw, circle, minimum size=0.7cm, inner sep = 0.5pt] (Smc2) at (15.25,8.5){$S_m$};
		
		\draw[-stealth, shorten >=1pt, auto] (c2) edge [] node {} (S1c2);
		\draw[-stealth, shorten >=1pt, auto] (c2) edge [] node {} (S2c2);
		\draw[-stealth, shorten >=1pt, auto] (c2) edge [] node {} (dotsc2);
		\draw[-stealth, shorten >=1pt, auto] (c2) edge [] node {} (Smc2);

		\draw[-stealth, shorten >=1pt, auto, in=295, out=245, looseness=4] (S1c2) edge [] node {} (S1c2);
		\draw[-stealth, shorten >=1pt, auto, in=295, out=245, looseness=4] (S2c2) edge [ ] node {} (S2c2);
		\draw[-stealth, shorten >=1pt, auto, in=295, out=245, looseness=4] (Smc2) edge [ ] node {} (Smc2);

	    \node[minimum size=0.7cm, inner sep = 0.5pt] (dotsmid) at (16.5,10){$\dots$};


		\node[draw, circle, minimum size=0.7cm, inner sep = 0.5pt] (ck) at (19.25,10){$c_k$};
		\node[draw, circle, minimum size=0.7cm, inner sep = 0.5pt] (S1ck) at (17.75,8.5){$S_1$};
		\node[draw, circle, minimum size=0.7cm, inner sep = 0.5pt] (S2ck) at (18.75,8.5){$S_2$};
		\node[minimum size=0.7cm, circle, inner sep = 0.5pt] (dotsck) at (19.75,8.5){$\dots$};
		\node[draw, circle, minimum size=0.7cm, inner sep = 0.5pt] (Smck) at (20.75,8.5){$S_m$};
		
		\draw[-stealth, shorten >=1pt, auto] (ck) to [] (S1ck);
		\draw[-stealth, shorten >=1pt, auto] (ck) to [] (S2ck);
		\draw[-stealth, shorten >=1pt, auto] (ck) to [] (dotsck);
		\draw[-stealth, shorten >=1pt, auto] (ck) to [] (Smck);

		\draw[-stealth, shorten >=1pt, auto, in=295, out=245, looseness=4] (S1ck) edge [] node {} (S1ck);
		\draw[-stealth, shorten >=1pt, auto, in=295, out=245, looseness=4] (S2ck) edge [ ] node {} (S2ck);
		\draw[-stealth, shorten >=1pt, auto, in=295, out=245, looseness=4] (Smck) edge [ ] node {} (Smck);
		
		\node[draw, rectangle, minimum size=0.7cm, inner sep = 0.5pt] (v0) at (5,11.25){$v_0$};
		\node[draw, rectangle, minimum size=0.7cm, inner sep = 0.5pt] (v2) at (15.125,11.25){$v_2$};
		
		\draw[-stealth, shorten >=1pt, auto] (v0) to [] (v1);
		\draw[-stealth, shorten >=1pt, auto] (v0) to [] (v2);
		
		\draw[-stealth, shorten >=1pt, auto] (v2.225) to [] (c1);
		\draw[-stealth, shorten >=1pt, auto] (v2.247) to [] (c2);
		\draw[-stealth, shorten >=1pt, auto] (v2.292) to [] (dotsmid);
		\draw[-stealth, shorten >=1pt, auto] (v2.315) to [] (ck);

		\end{tikzpicture}
		}%
	
	\caption{The tree arena used in the reduction from the \setCoverAb{}.}
	\label{scp_game}
	\Description{Figure 4. Fully described in the text.}
\end{figure}

Given an instance of the \setCoverAb{}, we construct a reachability \gameAb{} played on a tree arena consisting of a polynomial number ($n + \problemParam \cdot (m + 1) + 3$) of vertices. The arena $G$ of the game is provided in Figure \ref{scp_game} and can be seen as two sub-arenas reachable from the initial vertex $v_0$. 
The game is such that there is a solution to the \setCoverAb{} if and only if Player~$0$ has a strategy from $v_0$ in $G$ which is a solution to the \problemAb{}. 
The game is played between Player~$0$ with reachability objective $\Omega_0$ and Player~$1$ with $n + 1$ reachability objectives. The objectives are defined as follows: $\ObjPlayer{0} = \reach{\{v_2\}}$, $\ObjPlayer{i} = \reach{\{e_i\} \cup \{S_j \mid e_i \in S_j\}}$ for $i \in \{1,2, \dots, n\}$ and $\ObjPlayer{n+1} = \reach{\{v_2\}}$. First, notice that every play in $G_1$ is consistent with any strategy of Player~$0$ and is lost by that player. It holds that for each $\ell \in \{1, 2, \dots, n\}$, there is such a play with payoff $(p_1, \ldots, p_{n+1})$ such that $p_\ell = 1$ and $p_j = 0$ for $j \neq \ell$. These payoffs correspond to the elements $e_\ell$ we aim to cover in the \setCoverAb{}. A play in $G_2$ visits $v_2$ and then a vertex $c$ from which Player~$0$ selects a vertex $S$. Such a play is always won by Player~$0$ and its payoff is $(p_1, \ldots, p_{n+1})$ such that $p_{n+1} = 1$ and $p_r = 1$ if and only if the element $e_r$ belongs to the set $S$. It follows that the payoff of such a play corresponds to a set of elements in the \setCoverAb{}. The following proposition holds and it follows that, as a consequence, Theorem~\ref{thm:npcomplete} holds.

\begin{proposition}
\label{prop:tree_nphard}
    There is a solution to an instance of the \setCoverAb{} if and only if Player~$0$ has a strategy from $v_0$ in the corresponding \gameAb{} that is a solution to the \problemAb{}.
\end{proposition}

\begin{proof}
    First, let us assume that there is a solution to the \setCoverAb{}. It holds that there exists a set of $\problemParam$ indexes $i_1, i_2, \dots, i_\problemParam$ such that the union of the corresponding sets equals the set $C$ of elements we aim to cover. We define the strategy $\sigma_0$ as follows: $\sigma_0(v_0 v_2 c_j) = S_{i_j}$ for $j \in \{1, \ldots, k\}$. Let us show that this strategy is solution to the \problemAb{} by showing that any play with a \paretoOptimal{} payoff is won by Player~$0$. This amounts to showing that for every play in $G_1$ there is a play in $G_2$ with a strictly larger payoff. This is sufficient as it makes sure that the payoff of plays in $G_1$ are not \paretoOptimal{} and as every play in $G_2$ is won by Player~$0$. Let $p =(p_1, \ldots, p_{n+1})$ be the payoff of a play in $G_1$. It holds that $p_\ell = 1$ for some $\ell \in \{1, 2, \dots, n\}$ and $p_j = 0$ for $\ell \neq j$. This corresponds to the element $e_\ell$ in $C$. Since the $\problemParam$ indexes $i_1, i_2, \dots, i_\problemParam$ are a solution to the \setCoverAb{}, it holds that there exists some index $i_j$ such that $e_\ell \in S_{i_j}$. It also holds that the play $v_0 v_2 c_j (S_{i_j})^\omega$ is consistent with $\sigma_0$. Its payoff is $p' = (p'_1, \ldots, p'_{n+1})$ with $p'_\ell = 1$ since $e_\ell \in S_{i_j}$ and $p'_{n+1} = 1$. It follows that payoff $p'$ is strictly larger than $p$. 
    
    Now, let us assume that Player~$0$ has a strategy $\sigma_0$ from $v_0$ that is a solution to the \problemAb{}. Let us show that the set of indexes $\{i_j \mid \sigma_0(v_0 v_2 c_j) = S_{i_j}, j \in \{1, \dots, \problemParam \}\}$ is a solution to the \setCoverAb{}. It is easy to see that since strategy $\sigma_0$ is a solution to the \problemAb{}, every payoff $p$ in $G_1$ is strictly smaller than some payoff $p'$ in $G_2$. It follows that in the \setCoverAb{}, each element $e \in C$ corresponding to $p$ is contained in some set $S$ corresponding to $p'$. Since it also holds that $S \subseteq C$ for each set $S$, it follows that the sets mentioned above are an exact cover of $C$.
\end{proof}

Notice that the proof of the \np-hardness of solving the \problemAb{} for reachability \gamesAb{} played on tree arenas can easily be adapted to other objectives. In particular, the same tree arena can be used and the reachability objectives used in the reduction can be translated into B\"uchi or co-B\"uchi objectives. This yields the \np-hardness of solving the \problemAb{} in (co-)B\"uchi \gamesAb{} mentioned in Theorem~\ref{thm:nexptimehard}. The \nexptime{}-hardness for the other objectives is obtained thanks to the succinct variant of the \setCoverAb{} presented in the next subsection.

\subsection{Succinct Set Cover Problem}

The \emph{\succinctSetCover{}} (\succinctSetCoverAb{}) is defined as follows. We are given a Conjunctive Normal Form (CNF) formula $\phi = C_1 \land C_2 \land \dots \land C_p$ over the variables $X = \{x_1, x_2, \ldots, x_m\}$ made up of $p$ clauses, each containing some disjunction of literals of the variables in $X$. The set of valuations of the variables $X$ which satisfy $\phi$ is written $\llbracket \phi \rrbracket$. We are also given an integer $\problemParam \in \mathbb{N}$ (encoded in binary) and an other CNF formula $\psi = D_1 \land D_2 \land \dots \land D_q$ over the variables $X \cup Y$ with $ Y = \{y_1, y_2, \ldots, y_n \}$, made up of $q$ clauses. Given a valuation $val_Y: Y \rightarrow \{0, 1\}$ of the variables in $Y$, called a \emph{partial valuation}, we write $\psi[val_Y]$ the CNF formula obtained by replacing in $\psi$ each variable $y \in Y$ by its valuation $val_Y(y)$. We write $\llbracket \psi[val_Y] \rrbracket$ the valuations of the remaining variables $X$ which satisfy $\psi[val_Y]$. The \succinctSetCoverAb{} is to decide whether there exists a set $K = \big{\{}val_Y \mid val_Y: Y \rightarrow \{0, 1\} \big{\}}$ of $\problemParam$ valuations of the variables in $Y$ such that the valuations of the remaining variables $X$ which satisfy the formulas $\psi[val_Y]$ include the valuations of $X$ which satisfy $\phi$. Formally, we write this $\llbracket \phi \rrbracket \subseteq \bigcup\limits_{val_Y \in K} \llbracket \psi[val_Y] \rrbracket$. 

We can show that this corresponds to a set cover problem succinctly defined using CNF formulas. The set $\llbracket \phi \rrbracket$ of valuations of $X$ which satisfy $\phi$ corresponds to the set of elements we aim to cover. Parameter $\problemParam$ is the number of sets that can be used to cover these elements. Such a set is described by a formula $\psi[val_Y]$, given a partial valuation $val_Y$, and its elements are the valuations of $X$ in $\llbracket \psi[val_Y] \rrbracket$. This is illustrated in the following example.

\begin{example}
\label{example_phi}
Consider the CNF formula $\phi = (x_1 \lor \neg x_2) \land (x_2 \lor x_3)$ over the variables $X = \{x_1, x_2, x_3\}$. The set of valuations of the variables which satisfy $\phi$ is $\llbracket \phi \rrbracket = \{(1,1,1), (1,1,0), (1,0,1), (0,0,1)\}$. Each such valuation corresponds to one element we aim to cover. Consider the CNF formula $\psi = (y_1 \lor y_2) \land (x_1 \lor y_2) \land (x_2 \lor x_3 \lor y_1)$ over the variables $X \cup Y$ with $Y = \{y_1, y_2\}$. Given the partial valuation $val_Y$ of the variables in $Y$ such that $val_Y(y_1) = 0$ and $val_Y(y_2) = 1$, we get the CNF formula $\psi[val_Y] = (0 \lor 1) \land (x_1 \lor 1) \land (x_2 \lor x_3 \lor 0)$. This formula describes the contents of the set identified by the partial valuation (as a partial valuation yields a unique formula). The valuations of the variables $X$ which satisfy $\psi[val_Y]$ are the elements contained in the set. In this case, these elements are $\llbracket \psi[val_Y] \rrbracket = \{(0,1,0), (0,0,1), (0,1,1), (1,1,0), (1,0,1), (1,1,1)\}$. We can see that $\llbracket \psi[val_Y] \rrbracket$ contains all the elements $(1,1,1), (1,1,0), (1,0,1), (0,0,1)$ of $\llbracket \phi \rrbracket$.
\qed\end{example}

The following result is used in the proof of our \nexptime{}-hardness results and is of potential independent interest.

\begin{theorem}
\label{thm:ssc-completeness}
The \succinctSetCoverAb{} is \nexptimeComplete{}.
\end{theorem}

\begin{proof}
It is easy to see that the \succinctSetCoverAb{} is in \nexptime{}. We can show that the \succinctSetCoverAb{} is \nexptimeHard{} by reduction from the \emph{\dominatingSet{}} (\dominatingSetAb{}) which is known to be \nexptimeComplete{} for graphs \emph{succinctly} defined using CNF formulas \cite{DasST17}. An instance of the \dominatingSetAb{} is defined by a CNF formula $\theta$ over two sets of $n$ variables $X = \{x_1, x_2, \dots, x_n\}$ and $Y=\{y_1, y_2,\dots, y_n\}$ and an integer $\problemParam$ (encoded in binary). The formula $\theta$ succinctly defines an undirected graph in the following way. The set of vertices is the set of all valuations of the $n$ variables in $X$ (or over the $n$ variables in $Y$) of which there are $2^n$. Let $val_X$ and $val_Y$ be two such valuations, representing two vertices. Then, there is an edge between $val_X$ and $val_Y$ if and only if $\theta[val_X, val_Y]$ or $\theta[val_Y, val_X]$ is true. An instance of the \dominatingSetAb{} is positive if there exists a set $K = \{val^1_X, val^2_X, \dots, val^\problemParam_X\}$ of $\problemParam$ valuations of the variables in $X$, corresponding to $\problemParam$ vertices, such that all vertices in the graph are adjacent to a vertex in $K$. Formally, we write this $|\bigcup\limits_{val_X \in K} \{ val_Y \mid \theta[val_X, val_Y] \lor  \theta[val_Y, val_X] \mbox{ is true} \}| = 2^n$.

The \dominatingSetAb{} can be reduced in polynomial time to the \succinctSetCoverAb{} as follows. We define the CNF formula $\phi$ over the set of variables $X$ such that the formula is empty. Therefore, the set $\llbracket \phi\rrbracket$ is equal to the $2^n$ valuations of the variables in $X$. We then define the CNF formula $\psi$ over the set of variables $X$ and $Y$ such that it is the CNF equivalent to $\theta(X, Y) \lor \theta(Y, X)$.  
The latter formula has a size which is polynomial in the size of the CNF formula $\theta$ which defines the graph. We keep the same integer $\problemParam$. Then, it is direct to see that the instance of the \dominatingSetAb{} is positive if and only if the instance of \succinctSetCoverAb{} is positive. Indeed, there is a positive instance to the \dominatingSetAb{} if and only if there exists a set $K$ of $\problemParam$ valuations of the variables in $Y$ such that $\llbracket \phi \rrbracket \subseteq \bigcup\limits_{val_Y \in K} \llbracket \psi[val_Y] \rrbracket$.
\end{proof}

\subsection{\textsf{NEXPTIME}-Hardness for Reachability SP Games}
\label{subsec:nexptime_reach}

In this subsection, we describe our reduction from the \succinctSetCoverAb{} which allows us to show the \nexptime-hardness of solving the \problemAb{} in reachability \gamesAb{} and therefore prove Theorem~\ref{thm:nexptimehard} for reachability objectives.

Let $(\phi,\psi,k)$ be an instance of the \succinctSetCoverAb{} where $\phi = C_1 \land C_2 \land \dots \land C_p$ is a conjunction of $p$ clauses over $X = \{x_1, x_2, \ldots, x_m\}$, $\psi = D_1 \land D_2 \land \dots \land D_q$ is a conjunction of $q$ clauses over $X \cup Y$ with $Y = \{y_1, y_2, \ldots, y_n \}$, and $k$ is an integer given in binary.
We construct a reachability \gameAb{} with arena $G$ consisting of a polynomial number of vertices in the number of clauses and variables in the formulas $\phi$ and $\psi$ and in the length of the binary encoding of the integer $k$. This reduction is such that there is a solution to the \succinctSetCoverAb{} if and only if Player~$0$ has a strategy from $v_0$ in $G$ which is a solution to the \problemAb{}. The arena $G$, provided in Figure \ref{sscp_game_long}, can be viewed as three sub-arenas reachable from $v_0$. We call these sub-arenas $G_1$, $G_2$ and $G_3$. Sub-arena $G_3$ starts with a gadget $Q_\problemParam$ whose vertices belong to Player~$1$ and which provides exactly $k$ different paths from $v_0$ to $v_3$.

\paragraph{Gadget $Q_\problemParam$.} 
Parameter $\problemParam$ can be represented in binary using $r = \lfloor log_2(\problemParam) \rfloor + 1$ bits. It also holds that the binary encoding of $\problemParam$ corresponds to the sum of at most $r$ powers of 2. Given the binary encoding $b_{r-1} \dots b_1 b_0$ of $\problemParam$ such that $b_i \in \B$, let $ones = \{ i \in \{0, \dots, r-1\} \mid b_i = 1\}$. It holds that $\problemParam = \sum_{i \in ones}^{} 2^i$. Our gadget $Q_\problemParam$ is a graph with a polynomial number of vertices (in the length $r$ of the binary encoding of $k$) such that all these vertices belong to Player~$1$. For each $i \in ones$ there is $2^i$ different paths from the initial vertex $g_1$ to vertex $g_2$. Therefore, it holds that in $Q_\problemParam$ there are $\problemParam$ different paths from vertex $g_1$ to vertex $g_2$. 
\begin{example}
Let $\problemParam = 11$, it holds that it can be represented in binary using $\lfloor log_2(11) \rfloor + 1 = 4$ bits. The binary representation of $11$ is $1011$ and it can be obtained by the following sum $2^3 + 2^1 + 2^0$. The gadget $Q_{11}$ is detailed in Figure \ref{gadget}.\qed

\begin{figure}
	\centering
		\resizebox{0.5\textwidth}{!}{%
		\begin{tikzpicture}
		
	    \draw[loosely dashed, rounded corners] (-0.5,4) rectangle (9.5,-3) {};
	    \node[] at (-0.5 + 0.4,4 - 0.3) {$Q_{11}$};

		\node[draw, rectangle, minimum size=0.7cm, inner sep = 0.5pt] (al) at (0,0.5){$g_1$};
		
		\node[draw, rectangle, minimum size=0.7cm, inner sep = 0.5pt] (t1) at (2.5,0.5){};

		\node[draw, rectangle, minimum size=0.7cm, inner sep = 0.5pt] (m1) at (1.5,2.5){};
		\node[draw, rectangle, minimum size=0.7cm, inner sep = 0.5pt] (m11) at (2.5,1.5){};
		\node[draw, rectangle, minimum size=0.7cm, inner sep = 0.5pt] (m12) at (2.5,3.5){};
		\node[draw, rectangle, minimum size=0.7cm, inner sep = 0.5pt] (m2) at (3.5,2.5){};
		
		\node[draw, rectangle, minimum size=0.7cm, inner sep = 0.5pt] (b1) at (1.5,-1.5){};
		\node[draw, rectangle, minimum size=0.7cm, inner sep = 0.5pt] (b11) at (2.5,-0.5){};
		\node[draw, rectangle, minimum size=0.7cm, inner sep = 0.5pt] (b12) at (2.5,-2.5){};

		\node[draw, rectangle, minimum size=0.7cm, inner sep = 0.5pt] (b2) at (3.5,-1.5){};
		\node[draw, rectangle, minimum size=0.7cm, inner sep = 0.5pt] (b21) at (4.5,-0.5){};
		\node[draw, rectangle, minimum size=0.7cm, inner sep = 0.5pt] (b22) at (4.5,-2.5){};
		
		\node[draw, rectangle, minimum size=0.7cm, inner sep = 0.5pt] (b3) at (5.5,-1.5){};
		\node[draw, rectangle, minimum size=0.7cm, inner sep = 0.5pt] (b31) at (6.5,-0.5){};
		\node[draw, rectangle, minimum size=0.7cm, inner sep = 0.5pt] (b32) at (6.5,-2.5){};
		
		\node[draw, rectangle, minimum size=0.7cm, inner sep = 0.5pt] (b4) at (7.5,-1.5){};
		
	    \node[draw, rectangle, minimum size=0.7cm, inner sep = 0.5pt] (be) at (9,0.5){$g_2$};

		\draw[-stealth, shorten >=1pt,auto] (al) to [] (t1);

		\draw[-stealth, shorten >=1pt,auto] (al.45) to [] (m1.225);
		\draw[-stealth, shorten >=1pt,auto] (al.315) to [] (b1.135);

		\draw[-stealth, shorten >=1pt,auto] (m1) to [] (m11);
		\draw[-stealth, shorten >=1pt,auto] (m1) to [] (m12);
		\draw[-stealth, shorten >=1pt,auto] (m11) to [] (m2);
		\draw[-stealth, shorten >=1pt,auto] (m12) to [] (m2);

		\draw[-stealth, shorten >=1pt,auto] (b1) to [] (b11);
		\draw[-stealth, shorten >=1pt,auto] (b1) to [] (b12);
		\draw[-stealth, shorten >=1pt,auto] (b11) to [] (b2);
		\draw[-stealth, shorten >=1pt,auto] (b12) to [] (b2);
		
		\draw[-stealth, shorten >=1pt,auto] (b2) to [] (b21);
		\draw[-stealth, shorten >=1pt,auto] (b2) to [] (b22);
		\draw[-stealth, shorten >=1pt,auto] (b21) to [] (b3);
		\draw[-stealth, shorten >=1pt,auto] (b22) to [] (b3);
		
		\draw[-stealth, shorten >=1pt,auto] (b3) to [] (b31);
		\draw[-stealth, shorten >=1pt,auto] (b3) to [] (b32);
		\draw[-stealth, shorten >=1pt,auto] (b31) to [] (b4);
		\draw[-stealth, shorten >=1pt,auto] (b32) to [] (b4);
		
		\draw[] (m2) to [] (7.85, 2.5);
		\draw[-stealth, shorten >=1pt,auto] (7.85, 2.5) to [] (be.135);
		
		\draw[-stealth, shorten >=1pt,auto] (t1) to [] (be);

		\draw[-stealth, shorten >=1pt,auto] (b4.45) to [] (be.225);
		
		\end{tikzpicture}
		}%
	
	\caption{The gadget $Q_{11}$.}
	\label{gadget}
	\Description{Figure 5. Fully described in the text.}
\end{figure}
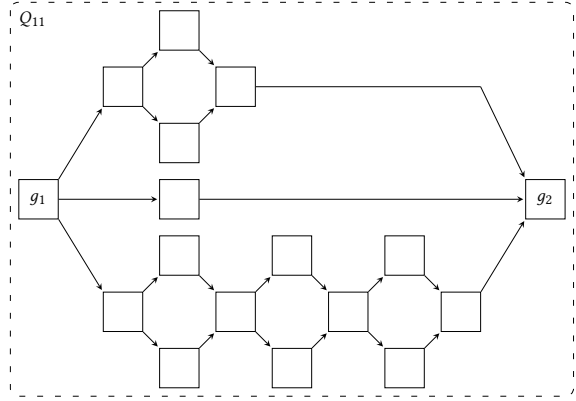
\end{example}

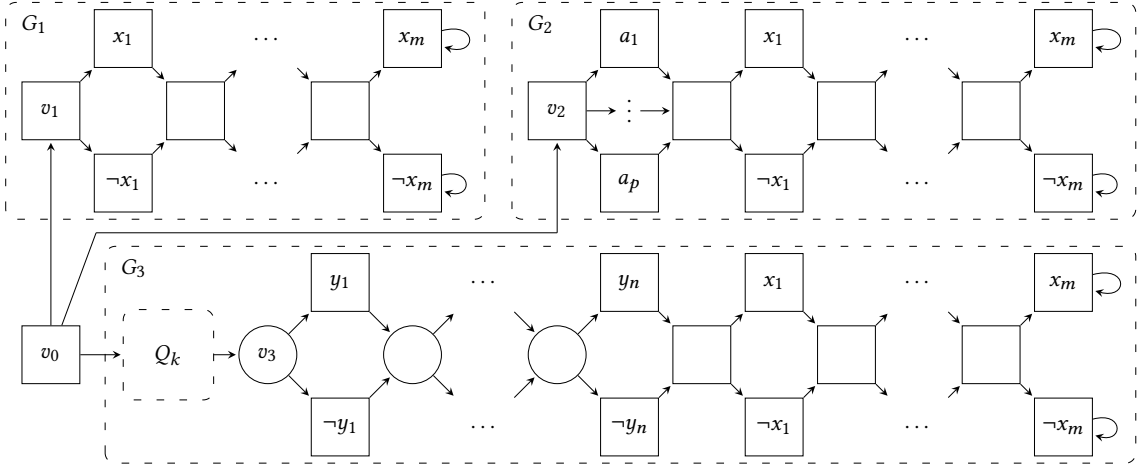
\begin{figure}
	\centering
	\resizebox{\textwidth}{!}{%
		
		\begin{tikzpicture}
		
		\node[draw, rectangle, minimum size=0.8cm] (o) at (-5,-4){$v_0$};
		
		
		\node[draw, rectangle, minimum size=0.8cm] (i0) at (2,-0.625){$v_2$};

		\draw[loosely dashed, rounded corners] (1.375,-2.125) rectangle (10, 0.875)  {};
        \node[] at (1.375 + 0.4, 0.875 - 0.3) {$G_2$};
		        
		\node[draw, rectangle, minimum size=0.8cm] (i1) at (3,0.375){$a_1$};
		\node[] (idot) at (3,-0.525){$\vdots$};
		\node[draw, rectangle, minimum size=0.8cm] (in) at (3,-1.625){$a_p$};
		
		\node[draw, rectangle, minimum size=0.8cm] (b1) at (4,-0.625){};
		\node[draw, rectangle, minimum size=0.8cm] (x11) at (5,0.375){$x_1$};
		\node[draw, rectangle, minimum size=0.8cm] (nx11) at (5,-1.625){$\neg x_1$};
		\node[draw, rectangle, minimum size=0.8cm] (b2) at (6,-0.625){};
		\node[minimum size=0.8cm] (b2i1) at (7,0.375){$\dots$};
		\node[minimum size=0.8cm] (b2i2) at (7,-1.625){$\dots$};
		\node[draw, rectangle, minimum size=0.8cm] (bm) at (8,-0.625){};
		\node[draw, rectangle, minimum size=0.8cm] (x1m) at (9,0.375){$x_m$};
		\node[draw, rectangle, minimum size=0.8cm] (nx1m) at (9,-1.625){$\neg x_m$};
		
		\draw[-] (o) to [] (-4.375,-2.3125);
		\draw[-] (-4.375,-2.3125) to [] node []{} (2,-2.3125);
		\draw[-stealth, shorten >=1pt,auto] (2,-2.3125) to [] node []{} (i0);
		\draw[-stealth, shorten >=1pt,auto] (i0) to [] node []{} (i1);
		\draw[-stealth, shorten >=1pt,auto] (i0) to [] node []{} (2.85,-0.625);
		\draw[-stealth, shorten >=1pt,auto] (i0) to [] node []{} (in);
		\draw[-stealth, shorten >=1pt,auto] (i1) to [] node []{} (b1);
		\draw[-stealth, shorten >=1pt,auto] (3.15,-0.625) to [] node []{} (b1);
		\draw[-stealth, shorten >=1pt,auto] (in) to [] node []{} (b1);
		\draw[-stealth, shorten >=1pt,auto] (b1) to [] node []{} (x11);
		\draw[-stealth, shorten >=1pt,auto] (b1) to [] node []{} (nx11);
		\draw[-stealth, shorten >=1pt,auto] (x11) to [] node []{} (b2);
		\draw[-stealth, shorten >=1pt,auto] (nx11) to [] node []{} (b2);
		\draw[-stealth, shorten >=1pt,auto] (b2) to [] node []{} (b2i1);
		\draw[-stealth, shorten >=1pt,auto] (b2) to [] node []{} (b2i2);
		\draw[-stealth, shorten >=1pt,auto] (b2i1) to [] node []{} (bm);
		\draw[-stealth, shorten >=1pt,auto] (b2i2) to [] node []{} (bm);
		\draw[-stealth, shorten >=1pt,auto] (bm) to [] node []{} (x1m);
		\draw[-stealth, shorten >=1pt,auto] (bm) to [] node []{} (nx1m);		
		\draw[-stealth,shorten >=1pt,auto,in=-15,out=15,looseness=6] (x1m) edge [] node {} (x1m);
		\draw[-stealth,shorten >=1pt,auto,in=-15,out=15,looseness=5.2] (nx1m) edge [] node {} (nx1m);

		\draw[loosely dashed, rounded corners] (-5.625, -2.125) rectangle (1, 0.875) {};
		\node[] at (-5.625 + 0.4, 0.875 - 0.3) {$G_1$};
		
		\node[draw, rectangle, minimum size=0.8cm] (a1) at (-5,-0.625){$v_1$};
		\node[draw, rectangle, minimum size=0.8cm] (x21) at (-4,0.375){$x_1$};
		\node[draw, rectangle, minimum size=0.8cm] (nx21) at (-4,-1.625){$\neg x_1$};
	    \node[draw, rectangle, minimum size=0.8cm] (a2) at (-3,-0.625){};
		\node[minimum size=0.8cm] (a2i1) at (-2,0.375){$\dots$};
		\node[minimum size=0.8cm] (a2i2) at (-2,-1.625){$\dots$};
		\node[draw, rectangle, minimum size=0.8cm] (am) at (-1,-0.625){};
		\node[draw, rectangle, minimum size=0.8cm] (x2m) at (0,0.375){$x_m$};
		\node[draw, rectangle, minimum size=0.8cm] (nx2m) at (0,-1.625){$\neg x_m$};
		
		\draw[-stealth, shorten >=1pt,auto] (o) to [] node []{} (a1);
		\draw[-stealth, shorten >=1pt,auto] (a1) to [] node []{} (x21);
		\draw[-stealth, shorten >=1pt,auto] (a1) to [] node []{} (nx21);
		\draw[-stealth, shorten >=1pt,auto] (x21) to [] node []{} (a2);
		\draw[-stealth, shorten >=1pt,auto] (nx21) to [] node []{} (a2);
		\draw[-stealth, shorten >=1pt,auto] (a2) to [] node []{} (a2i1);
		\draw[-stealth, shorten >=1pt,auto] (a2) to [] node []{} (a2i2);
		\draw[-stealth, shorten >=1pt,auto] (a2i1) to [] node []{} (am);
		\draw[-stealth, shorten >=1pt,auto] (a2i2) to [] node []{} (am);
		\draw[-stealth, shorten >=1pt,auto] (am) to [] node []{} (x2m);
		\draw[-stealth, shorten >=1pt,auto] (am) to [] node []{} (nx2m);		
		\draw[-stealth,shorten >=1pt,auto,in=-15,out=15,looseness=6] (x2m) edge [] node {} (x2m);
		\draw[-stealth,shorten >=1pt,auto,in=-15,out=15,looseness=5.2] (nx2m) edge [] node {} (nx2m);
		
		
		\draw[loosely dashed, rounded corners] (-4.25,-5.5) rectangle (10, -2.5) {};
		\node[] at (-4.25 + 0.4, -2.5 - 0.3) {$G_3$};

		\draw[loosely dashed, rounded corners] (-4,-4.625) rectangle (-2.75, -3.375) {};
		\node[rectangle, minimum size=0.8cm] (j1) at (-3.375,-4){$Q_\problemParam$};

		\draw[-stealth, shorten >=1pt,auto] (o) to [] ((-4,-4);
		
		
		\node[draw, circle, minimum size=0.8cm] (c1) at (-2,-4){$v_3$};
		\node[draw, rectangle, minimum size=0.8cm] (y1) at (-1,-5){$\neg y_1$};
		\node[draw, rectangle, minimum size=0.8cm] (ny1) at (-1,-3){$y_1$};
		\node[draw, circle, minimum size=0.8cm] (c2) at (0,-4){};
		\node[minimum size=0.8cm] (ci1) at (1,-3){$\dots$};
		\node[minimum size=0.8cm] (ci2) at (1,-5){$\dots$};
		\node[draw, circle, minimum size=0.8cm] (cm) at (2,-4){};
		\node[draw, rectangle, minimum size=0.8cm] (ym) at (3,-5){$\neg y_n$};
		\node[draw, rectangle, minimum size=0.8cm] (nym) at (3,-3){$y_n$};
		
		\draw[-stealth, shorten >=1pt,auto] ((-2.75,-4) to [] node []{} (c1);

		\draw[-stealth, shorten >=1pt,auto] (c1) to [] node []{} (y1);
		\draw[-stealth, shorten >=1pt,auto] (c1) to [] node []{} (ny1);
		\draw[-stealth, shorten >=1pt,auto] (y1) to [] node []{} (c2);
		\draw[-stealth, shorten >=1pt,auto] (ny1) to [] node []{} (c2);
		\draw[-stealth, shorten >=1pt,auto] (c2) to [] node []{} (ci1);
		\draw[-stealth, shorten >=1pt,auto] (c2) to [] node []{} (ci2);
		\draw[-stealth, shorten >=1pt,auto] (ci1) to [] node []{} (cm);
		\draw[-stealth, shorten >=1pt,auto] (ci2) to [] node []{} (cm);
		\draw[-stealth, shorten >=1pt,auto] (cm) to [] node []{} (ym);
		\draw[-stealth, shorten >=1pt,auto] (cm) to [] node []{} (nym);			
		
		
		\node[draw, rectangle, minimum size=0.8cm] (d1) at (4,-4){};
		\node[draw, rectangle, minimum size=0.8cm] (x31) at (5,-5){$\neg x_1$};
		\node[draw, rectangle, minimum size=0.8cm] (nx31) at (5,-3){$x_1$};
		\node[draw, rectangle, minimum size=0.8cm] (d2) at (6,-4){};
		\node[minimum size=0.8cm] (d2i1) at (7,-5){$\dots$};
		\node[minimum size=0.8cm] (d2i2) at (7,-3){$\dots$};
		\node[draw, rectangle, minimum size=0.8cm] (dm) at (8,-4){};
		\node[draw, rectangle, minimum size=0.8cm] (x3m) at (9,-5){$\neg x_m$};
		\node[draw, rectangle, minimum size=0.8cm] (nx3m) at (9,-3){$x_m$};
		
		\draw[-stealth, shorten >=1pt,auto] (ym) to [] node []{} (d1);		
		\draw[-stealth, shorten >=1pt,auto] (nym) to [] node []{} (d1);	
		\draw[-stealth, shorten >=1pt,auto] (d1) to [] node []{} (x31);
		\draw[-stealth, shorten >=1pt,auto] (d1) to [] node []{} (nx31);
		\draw[-stealth, shorten >=1pt,auto] (x31) to [] node []{} (d2);
		\draw[-stealth, shorten >=1pt,auto] (nx31) to [] node []{} (d2);
		\draw[-stealth, shorten >=1pt,auto] (d2) to [] node []{} (d2i1);
		\draw[-stealth, shorten >=1pt,auto] (d2) to [] node []{} (d2i2);
		\draw[-stealth, shorten >=1pt,auto] (d2i1) to [] node []{} (dm);
		\draw[-stealth, shorten >=1pt,auto] (d2i2) to [] node []{} (dm);
		\draw[-stealth, shorten >=1pt,auto] (dm) to [] node []{} (x3m);
		\draw[-stealth, shorten >=1pt,auto] (dm) to [] node []{} (nx3m);		
		\draw[-stealth,shorten >=1pt,auto,in=-15,out=15,looseness=5.2] (x3m) edge [] node {} (x3m);
		\draw[-stealth,shorten >=1pt,auto,in=-15,out=15,looseness=6] (nx3m) edge [] node {} (nx3m);

		\end{tikzpicture}
		}%
	
	\caption{The arena $G$ used in the reduction from the \succinctSetCoverAb{}.}
	\label{sscp_game_long}
	\Description{Figure 6. Fully described in the text.}
\end{figure}

\paragraph{Objectives.}
The game is played between Player~$0$ with reachability objective $\ObjPlayer{0}$ and Player~$1$ with $\nbrObjectives = 1 + 2 \cdot m + p + q$ reachability objectives. The payoff of a play therefore consists in a single Boolean for objective $\ObjPlayer{1}$, a vector of $2 \cdot m$ Booleans for objectives $\ObjPlayer{x_1}, \ObjPlayer{\neg x_1}, \dots,  \ObjPlayer{x_m}, \ObjPlayer{\neg x_m}$, a vector of $p$ Booleans for objectives $ \ObjPlayer{C_1}, \dots, \ObjPlayer{C_p}$ and a vector of $q$ Booleans for objectives $\ObjPlayer{D_1}, \dots, \ObjPlayer{D_q}$. The objectives are defined as follows.
\begin{itemize}
    \item The target set for objective $\ObjPlayer{0}$ of Player~$0$ and objective $\ObjPlayer{1}$ of Player~$1$ is $\{v_2, v_3\}$.
    \item The target set for objective $\ObjPlayer{x_i}$ (resp.\ $\ObjPlayer{\neg x_i}$) with $i \in \{1, \dots, m\}$ is the set of vertices labeled $x_i$ (resp.\ $\neg x_i$) in $G_1$, $G_2$ and $G_3$. 
    \item The target set for objective $\ObjPlayer{C_i}$ with $i \in \{1, \dots, p\}$ is the set of vertices in $G_1$ and $G_3$ corresponding to the literals of $X$ which make up the clause $C_i$ in $\phi$. In addition, vertex $a_j$ in $G_2$ belongs to the target set of objective $\ObjPlayer{C_\ell}$ for all $\ell \in \{1, \dots, p\}$ such that $\ell \neq j$. 
    \item The target set of objective $\ObjPlayer{D_i}$ with $i \in \{1, \dots, q\}$ is the set of vertices in $G_3$ corresponding to the literals of $X$ and $Y$ which make up the clause $D_i$ in $\psi$. In addition, vertices $v_1$ and $v_2$ satisfy every objective $\ObjPlayer{D_i}$ with $i \in \{1, \dots, q\}$. 
\end{itemize}

\paragraph{Sub-Arenas $G_1$ and $G_2$.}
In each sub-arena $G_1$ and $G_2$, for each variable $x_i \in X$, there is one choice vertex controlled by Player~$1$ which leads to $x_i$ and $\neg x_i$. These vertices have the next choice vertex as their successor, except for vertices $x_m$ and $\neg x_m$ which have a self loop. In $G_2$, there is also a vertex $v_2$ controlled by Player~$1$ with $p$ successors $a_1, \dots, a_p$, each leading to the first choice vertex for the variables in $X$. Sub-arenas $G_1$ and $G_2$ are completely controlled by Player~$1$. Plays entering these sub-arenas are therefore consistent with any strategy of Player~$0$.

\paragraph{Payoff of Plays in $G_1$.}
Plays in $G_1$ do not satisfy objective $\Omega_0$ of Player $0$ nor objective $\ObjPlayer{1}$ of Player~$1$. A play in $G_1$ is of the form $v_0 \: v_1 \: z_1 \boxempty \dots \boxempty (z_m)^\omega$ where $z_i$ is either $x_i$ or $\neg x_i$. It follows that a play satisfies the objective $\ObjPlayer{x_i}$ or $\ObjPlayer{\neg x_i}$ for each $x_i \in X$. The vector of Booleans for these objectives corresponds to a valuation of the variables in $X$, expressed as a vector of $2 \cdot m$ Booleans. In addition, due to the way the objectives are defined, objective $\ObjPlayer{C_i}$ is satisfied in a play if and only if clause $C_i$ of $\phi$ is satisfied by the valuation this play corresponds to. The objective $\ObjPlayer{D_i}$ for $i \in \{1, \dots, q\}$ is satisfied in every play in $G_1$.

\begin{lemma} \label{lem:G1}
    Plays in $G_1$ are consistent with any strategy of Player~$0$. Their payoff are of the form $(0, val, sat(\phi, val), 1, \dots, 1)$ where $val$ is a valuation of the variables in $X$ expressed as a vector of $2 \cdot m$ Booleans for objectives $\ObjPlayer{x_1}$ to $\ObjPlayer{\neg x_m}$ and $sat(\phi, val)$ is the vector of $p$ Booleans for objectives $\ObjPlayer{C_1}$ to $\ObjPlayer{C_p}$ corresponding to that valuation. All plays in $G_1$ are lost by Player~$0$.
\end{lemma}

\paragraph{Payoff of Plays in $G_2$.}
Plays in $G_2$ satisfy the objectives $\ObjPlayer{0}$ of Player~$0$ and $\ObjPlayer{1}$ of Player~$1$. A play in $G_2$ is of the form $v_0 \: v_2 \: a_j \boxempty z_1 \boxempty \dots \boxempty (z_m)^\omega$ where $z_\ell$ is either $x_\ell$ or $\neg x_\ell$. It follows that a play satisfies either the objective $\ObjPlayer{x}$ or $\ObjPlayer{\neg x}$ for each $x \in X$ which again corresponds to a valuation of the variables in $X$. The objective $\ObjPlayer{D_i}$ for $i \in \{1, \dots, q\}$ is satisfied in every play in $G_2$. Compared to the plays in $G_1$, the difference lies in the objectives corresponding to clauses of $\phi$ which are satisfied. In any play in $G_2$, a vertex $a_{j}$ with $j \in \{1, \dots, p\}$ is first visited, satisfying all the objectives $\ObjPlayer{C_\ell}$ with $\ell \in \{1, \dots, p\}$ and $\ell \neq j$. All but one objective corresponding to the clauses of $\phi$ are therefore satisfied. 

\begin{lemma} \label{lem:G2}
    Plays in $G_2$ are consistent with any strategy of Player~$0$. Their payoff are of the form $(1, val, vec, 1, \dots, 1)$ where $val$ is a valuation of the variables in $X$ expressed as a vector of $2 \cdot m$ Booleans for objectives $\ObjPlayer{x_1}$ to $\ObjPlayer{\neg x_m}$ and $vec$ is a vector of $p$ Booleans for objectives $\ObjPlayer{C_1}$ to $\ObjPlayer{C_p}$ in which all of them except one are satisfied. All plays in $G_2$ are won by Player~$0$.
\end{lemma}

From the two previous lemmas, we can state the following lemma when considering the payoffs of plays in $G_1$ and $G_2$.

\begin{lemma} \label{lem:valsatis}
\label{larger_payoff}
    For every play in $G_1$ which corresponds to a valuation of the variables in $X$ that does not satisfy $\phi$, there is a play in $G_2$ with a strictly larger payoff.
\end{lemma}

\begin{proof}
Let $\rho$ be a play in $G_1$ which corresponds to a valuation of the variables in $X$ that does not satisfy $\phi$. It follows that at least one objective, say $\ObjPlayer{C_\ell}$, is not satisfied in $\rho$ as at least one clause of $\phi$ (clause $C_\ell$) is not satisfied by that valuation. Let us consider the play $\rho'$ in $G_2$ which visits vertex $a_\ell$ and after visits the vertices corresponding to the same valuation of the variables in $X$ as $\rho$. By Lemmas~\ref{lem:G1} and~\ref{lem:G2}, it follows that the payoff of $\rho'$ is strictly larger than that of $\rho$ (as we have $(0, val, sat(\phi, val), 1, \dots, 1) < (1, val, vec, 1, \dots, 1)$ with $sat(\phi, val) \leq vec$). 
\end{proof}

The following lemma is a consequence of Lemma~\ref{lem:valsatis}.

\begin{lemma}
\label{lem:g1_g2_payoffs}
    Let $\sigma_0$ be a strategy of Player~$0$. The set of payoffs of plays in $G_1$ that are \paretoOptimal{} when considering $G_1 \cup G_2$ is equal to the set of payoffs of plays in $G_1$ whose valuation of $X$ satisfy $\phi$.
\end{lemma}

\begin{proof} Recall that all plays in $G_1 \cup G_2$ are consistent with any strategy of Player~$0$. The property of Lemma~\ref{lem:g1_g2_payoffs} stems from the following observations. First, any play in $G_1$ which satisfies every objective $\ObjPlayer{C_i}$ with $i \in \{1, \dots, p\}$, and therefore corresponds to a valuation of $X$ which satisfies $\phi$, has a payoff that is incomparable to every possible payoff of plays in $G_2$. This is because such a play satisfies more objectives in $\ObjPlayer{C_1}, \dots, \ObjPlayer{C_p}$ than the plays in $G_2$ but does not satisfy objective $\ObjPlayer{1}$ while the plays in $G_2$ do. Second, every other play in $G_1$ has a strictly smaller payoff than at least one play in $G_2$ due to Lemma \ref{larger_payoff} and its payoff is therefore not \paretoOptimal{}.
\end{proof}

\paragraph{Problematic Payoffs in $G_1$.}
The plays of $G_1$ described in the previous lemma correspond exactly to the valuations of $X$ which satisfy $\phi$ and therefore to the elements we aim to cover in the \succinctSetCoverAb{}. They are \paretoOptimal{} when considering $G_1 \cup G_2$ and are lost by Player~$0$. All other \paretoOptimal{} payoffs in $G_1 \cup G_2$ are only realized by plays in $G_2$ which are all won by Player~$0$. It follows that in order for Player~$0$ to find a strategy $\sigma_0$ from $v_0$ that is solution to the \problemAb{}, it must hold that those payoffs are not \paretoOptimal{} when considering $G_1 \cup G_2 \cup G_3$. Otherwise, a play of $G_1$ consistent with $\sigma_0$ with a \paretoOptimal{} payoff is lost by Player~$0$. We therefore call the payoffs of plays in $G_1$ mentioned in Lemma~\ref{lem:g1_g2_payoffs} \emph{problematic payoffs}.

In order for Player~$0$ to find a strategy $\sigma_0$ which is a solution to the \problemAb{}, this strategy must be such that for each problematic payoff in $G_1$, there is a play in $G_3$ consistent with $\sigma_0$ and with a strictly larger payoff. Since the plays in $G_3$ are all won by Player~$0$, this would ensure that the strategy $\sigma_0$ is a solution to the problem. This corresponds in the \succinctSetCoverAb{} to selecting a series of sets in order to cover the valuations of $X$ which satisfy $\phi$.

\paragraph{Sub-Arena $G_3$.}
Sub-arena $G_3$ starts with gadget $Q_\problemParam$ whose vertices are controlled by Player~$1$. Then, for each variable $y_i \in Y$, there is one choice vertex controlled by Player~$0$ which leads to $y_i$ and $\neg y_i$. These vertices have the next choice vertex as their successor, except for $y_n$ and $\neg y_n$ which lead to the first choice vertex for the variables in $X$. 

\paragraph{Payoff of Plays in $G_3$.}
Plays in $G_3$ satisfy the objectives $\ObjPlayer{0}$ of Player~$0$ and $\ObjPlayer{1}$ of Player~$1$. A play in $G_3$ consistent with a strategy $\sigma_0$ is of the form $v_0 \boxempty \dots \boxempty v_3 \, t_1 \raisebox{0.2ex}{$\scriptstyle\varbigcirc$} \cdots \raisebox{0.2ex}{$\scriptstyle\varbigcirc$} \, t_n \boxempty z_1 \boxempty \dots \boxempty (z_m)^\omega$ where $t_i$ is either $y_i$ or $\neg y_i$ and $z_i$ is either $x_i$ or $\neg x_i$. Since only the vertices leading to $y$ or $\neg y$ for $y \in Y$ belong to Player~$0$, it holds that $v_3 \, t_1 \raisebox{0.2ex}{$\scriptstyle\varbigcirc$} \cdots \raisebox{0.2ex}{$\scriptstyle\varbigcirc$} \, t_n$ is the only part of any play in $G_3$ which is directly influenced by $\sigma_0$. That part of a play comes after a history from $v_0$ to $v_3$ of which there are $\problemParam$, provided by gadget $Q_k$. By definition of a strategy, this can be interpreted as Player~$0$ making a choice of valuation of the variables in $Y$ after each of those $\problemParam$ histories. After this, the play satisfies either the objective $\ObjPlayer{x}$ or $\ObjPlayer{\neg x}$ for each $x \in X$ which corresponds to a valuation of $X$. Due to the way the objectives are defined, the objective $\ObjPlayer{C_i}$ (resp.\ $\ObjPlayer{D_i}$) is satisfied if and only if clause $C_i$ of $\phi$ (resp.\ $D_i$ of $\psi$) is satisfied by the valuation of the variables in $X$ (resp. $X$ and $Y$) the play corresponds to. 

\paragraph{Creating Strictly Larger Payoffs in $G_3$.}
In order to create a play with a payoff $r'$ that is strictly larger than a problematic payoff $r$, $\sigma_0$ must choose a valuation of $Y$ such that there exists a valuation of the remaining variables $X$ which together with this valuation of $Y$ satisfies $\psi$ and $\phi$ (since in $r$ every objective $\ObjPlayer{C_i}$ for $i \in \{1, \dots, p\}$ and $\ObjPlayer{D_i}$ for $i \in \{1, \dots, q\}$ is satisfied). Since the plays in $G_3$ also satisfy the objective $\ObjPlayer{1}$ and plays in $G_1$ do not, this ensures that $r < r'$. 

We can finally establish that our reduction is correct.
\begin{proposition}\label{prop:nexptime-hard-correct-reach}
    There is a solution to an instance of the \succinctSetCoverAb{} if and only if Player~$0$ has a strategy $\sigma_0$ that is a solution to the \problemAb{} in the corresponding reachability \gameAb{} played on $G$.
\end{proposition}

\begin{proof}
Let us assume that that $\sigma_0$ is a solution to the \problemAb{} in $G$ and show that there is a solution to the \succinctSetCoverAb{}. Let $val_X$ be a valuation of the variables in $X$ which satisfies $\phi$. This valuation corresponds to a play in $G_1$ with a problematic payoff $r$. Since the objective of Player~$0$ is not satisfied in that play and since $\sigma_0$ is a solution to the \problemAb{}, it holds that $r$ is not \paretoOptimal{}. It follows that there exists a play in $G_3$ that is consistent with $\sigma_0$ and whose payoff is strictly larger than $r$. As described above, such a play corresponds to a valuation $val_Y$ of the variables in $Y$ such that $val_X \in \llbracket \psi[val_Y] \rrbracket$. Since this can be done for each $val_X \in \llbracket \phi \rrbracket$ and since there is a set $K$ of $\problemParam$ possible valuations $val_Y$ in $G_3$, it holds that $\llbracket \phi \rrbracket \subseteq \bigcup\limits_{val_Y \in K} \llbracket \psi[val_Y] \rrbracket$. 

Let us now assume that there is a solution to the \succinctSetCoverAb{} and show that we can construct a strategy $\sigma_0$ that is solution to the \problemAb{}. Let $K$ be the set of $\problemParam$ valuations $val_Y$ of the variables in $Y$ which is a solution to the \succinctSetCoverAb{}. Since there are $\problemParam$ possible histories from $v_0$ to $v_3$ in $G_3$ provided by the gadget $Q_\problemParam$ described previously, we define $\sigma_0$ such that the $n$ vertices $y_i$ or $\neg y_i$ for $i \in \{1, \dots , n\}$ visited after each history correspond to a valuation in $K$. We can now show that this strategy is a solution to the \problemAb{}. We do this by showing that each play $\rho$ with problematic payoff $r$ in $G_1$ has a strictly smaller payoff than that of some play $\rho'$ with payoff $r'$ in $G_3$. Such a payoff $r$ corresponds to a valuation $val_X \in \llbracket \phi \rrbracket$. Since $K$ is a solution to the \succinctSetCoverAb{}, it holds that there exists some valuation $val_Y \in K$ such that $val_X \in \llbracket \psi[val_Y] \rrbracket$. It follows, given the definition of $\sigma_0$, that there exists a play $\rho'$ in $G_3$ corresponding to that valuation $val_Y$ and which visits the vertices $x$ or $\neg x$ for each $x \in X$ such that it corresponds to the valuation $val_X$. Given the properties mentioned before, the payoff $r'$ of this play is such that $r < r'$.
\end{proof}

The previous proof yields the \nexptime{}-hardness of the \problemAb{} stated in Theorem~\ref{thm:nexptimehard} for the case of reachability \gamesAb{}.

\subsection{\textsf{NEXPTIME}-Hardness for Safety SP Games}

In this subsection, we prove Theorem~\ref{thm:nexptimehard} for safety \gamesAb{}. The arguments used here are similar to the ones used in the previous subsection and we therefore only highlight the main differences.

Given an instance of the \succinctSetCoverAb{}, we construct a safety \gameAb{} with an arena $G$ of polynomial size and with a polynomial number of safety objectives. The arena $G$, provided in Figure \ref{fig:sscp_game_long_safety}, is composed of the same sub-arenas $G_2$ and $G_3$ described in the previous subsection and a modified sub-arena $G'_1$. 

\begin{figure}
	\centering
	\resizebox{\textwidth}{!}{%
		
		\begin{tikzpicture}
		
		\node[draw, rectangle, minimum size=0.8cm] (o) at (-5.375,-0.5){$v_0$};
		
		
		\node[draw, rectangle, minimum size=0.8cm] (i0) at (1.5,-0.5){$v_2$};

		\draw[loosely dashed, rounded corners] (0.5,-2) rectangle (9.5, 1)  {};
        \node[] at (0.5 + 0.4, -2 + 0.3) {$G_2$};
		        
		\node[draw, rectangle, minimum size=0.8cm] (i1) at (2.5,0.5){$a_1$};
		\node[] (idot) at (2.5,-0.4){$\vdots$};
		\node[draw, rectangle, minimum size=0.8cm] (in) at (2.5,-1.5){$a_p$};
		
		\node[draw, rectangle, minimum size=0.8cm] (b1) at (3.5,-0.5){};
		\node[draw, rectangle, minimum size=0.8cm] (x11) at (4.5,0.5){$x_1$};
		\node[draw, rectangle, minimum size=0.8cm] (nx11) at (4.5,-1.5){$\neg x_1$};
		\node[draw, rectangle, minimum size=0.8cm] (b2) at (5.5,-0.5){};
		\node[minimum size=0.8cm] (b2i1) at (6.5,0.5){$\dots$};
		\node[minimum size=0.8cm] (b2i2) at (6.5,-1.5){$\dots$};
		\node[draw, rectangle, minimum size=0.8cm] (bm) at (7.5,-0.5){};
		\node[draw, rectangle, minimum size=0.8cm] (x1m) at (8.5,0.5){$x_m$};
		\node[draw, rectangle, minimum size=0.8cm] (nx1m) at (8.5,-1.5){$\neg x_m$};
		
		\draw[-stealth, shorten >=1pt,auto] (o) to [] node []{} (i0);
		\draw[-stealth, shorten >=1pt,auto] (i0) to [] node []{} (i1);
		\draw[-stealth, shorten >=1pt,auto] (i0) to [] node []{} (2.35,-0.5);
		\draw[-stealth, shorten >=1pt,auto] (i0) to [] node []{} (in);
		\draw[-stealth, shorten >=1pt,auto] (i1) to [] node []{} (b1);
		\draw[-stealth, shorten >=1pt,auto] (2.65,-0.5) to [] node []{} (b1);
		\draw[-stealth, shorten >=1pt,auto] (in) to [] node []{} (b1);
		\draw[-stealth, shorten >=1pt,auto] (b1) to [] node []{} (x11);
		\draw[-stealth, shorten >=1pt,auto] (b1) to [] node []{} (nx11);
		\draw[-stealth, shorten >=1pt,auto] (x11) to [] node []{} (b2);
		\draw[-stealth, shorten >=1pt,auto] (nx11) to [] node []{} (b2);
		\draw[-stealth, shorten >=1pt,auto] (b2) to [] node []{} (b2i1);
		\draw[-stealth, shorten >=1pt,auto] (b2) to [] node []{} (b2i2);
		\draw[-stealth, shorten >=1pt,auto] (b2i1) to [] node []{} (bm);
		\draw[-stealth, shorten >=1pt,auto] (b2i2) to [] node []{} (bm);
		\draw[-stealth, shorten >=1pt,auto] (bm) to [] node []{} (x1m);
		\draw[-stealth, shorten >=1pt,auto] (bm) to [] node []{} (nx1m);		
		\draw[-stealth,shorten >=1pt,auto,in=-15,out=15,looseness=6] (x1m) edge [] node {} (x1m);
		\draw[-stealth,shorten >=1pt,auto,in=-15,out=15,looseness=5.2] (nx1m) edge [] node {} (nx1m);
		
		\draw[loosely dashed, rounded corners] (-3.5, 1.25) rectangle (9.5, 4.25) {};
		\node[] at (-3.5 + 0.4, 1.25 + 0.3) {$G'_1$};
		
		\node[draw, rectangle, minimum size=0.8cm] (a1) at (-2.5,2.75){$v_1$};
		\node[draw, rectangle, minimum size=0.8cm] (x21) at (-1.5,3.75){$x_1$};
		\node[draw, rectangle, minimum size=0.8cm] (nx21) at (-1.5,1.75){$\neg x_1$};
		\node[draw, rectangle, minimum size=0.8cm] (a2) at (-0.5,2.75){};
		\node[minimum size=0.8cm] (a2i1) at (0.5,3.75){$\dots$};
		\node[minimum size=0.8cm] (a2i2) at (0.5,1.75){$\dots$};
		\node[draw, rectangle, minimum size=0.8cm] (am) at (1.5,2.75){};
		\node[draw, rectangle, minimum size=0.8cm] (x2m) at (2.5,3.75){$x_m$};
		\node[draw, rectangle, minimum size=0.8cm] (nx2m) at (2.5,1.75){$\neg x_m$};
		\node[draw, circle, minimum size=0.8cm] (d1) at (3.5,2.75){$D_1$};
		\node[draw, rectangle, minimum size=0.8cm] (d1l3) at (4.5,1.75){$l^{D_1}_{n_{D_1}}$};
		\node[rectangle, minimum size=0.8cm] (d1l2) at (4.5,2.85){$\vdots$};
		\node[draw, rectangle, minimum size=0.8cm] (d1l1) at (4.5,3.75){$l^{D_1}_1$};
		\node[draw, circle, minimum size=0.8cm] (d2) at (5.5,2.75){$D_2$};
		\node[rectangle, minimum size=0.8cm] (d2l3) at (6.5,1.75){$\dots$};
		\node[rectangle, minimum size=0.8cm] (d2l2) at (6.5,2.85){$\vdots$};
		\node[rectangle, minimum size=0.8cm] (d2l1) at (6.5,3.75){$\dots$};
		\node[draw, circle, minimum size=0.8cm] (dq) at (7.5,2.75){$D_q$};
		\node[draw,rectangle, minimum size=0.8cm] (dql3) at (8.5,1.75){$l^{D_q}_{n_{D_q}}$};
		\node[rectangle, minimum size=0.8cm] (dql2) at (8.5,2.85){$\vdots$};
		\node[draw,rectangle, minimum size=0.8cm] (dql1) at (8.5,3.75){$l^{D_q}_1$};
		
		\draw[] (o) to [] (-5.375, 2.75);
		\draw[-stealth] (-5.375, 2.75) to [] (a1);
		\draw[-stealth, shorten >=1pt,auto] (a1) to [] node []{} (x21);
		\draw[-stealth, shorten >=1pt,auto] (a1) to [] node []{} (nx21);
		\draw[-stealth, shorten >=1pt,auto] (x21) to [] node []{} (a2);
		\draw[-stealth, shorten >=1pt,auto] (nx21) to [] node []{} (a2);
		\draw[-stealth, shorten >=1pt,auto] (a2) to [] node []{} (a2i1);
		\draw[-stealth, shorten >=1pt,auto] (a2) to [] node []{} (a2i2);
		\draw[-stealth, shorten >=1pt,auto] (a2i1) to [] node []{} (am);
		\draw[-stealth, shorten >=1pt,auto] (a2i2) to [] node []{} (am);
		\draw[-stealth, shorten >=1pt,auto] (am) to [] node []{} (x2m);
		\draw[-stealth, shorten >=1pt,auto] (am) to [] node []{} (nx2m);		
		\draw[-stealth, shorten >=1pt,auto] (x2m) to [] node []{} (d1);
		\draw[-stealth, shorten >=1pt,auto] (nx2m) to [] node []{} (d1);
		\draw[-stealth, shorten >=1pt,auto] (d1) to [] node []{} (d1l1);
		\draw[-stealth, shorten >=1pt,auto] (d1) to [] node []{} (d1l3);
		\draw[-stealth, shorten >=1pt,auto] (d1l1) to [] node []{} (d2);
		\draw[-stealth, shorten >=1pt,auto] (d1l3) to [] node []{} (d2);
		\draw[-stealth, shorten >=1pt,auto] (d2) to [] node []{} (d2l1);
		\draw[-stealth, shorten >=1pt,auto] (d2) to [] node []{} (d2l3);
		\draw[-stealth, shorten >=1pt,auto] (d2l1) to [] node []{} (dq);
		\draw[-stealth, shorten >=1pt,auto] (d2l3) to [] node []{} (dq);		
		\draw[-stealth, shorten >=1pt,auto] (dq) to [] node []{} (dql1);
		\draw[-stealth, shorten >=1pt,auto] (dq) to [] node []{} (dql3);
		\draw[-stealth,shorten >=1pt,auto,in=-15,out=15,looseness=6] (dql1) edge [] node {} (dql1);
		\draw[-stealth,shorten >=1pt,auto,in=-15,out=15,looseness=5.2] (dql3) edge [] node {} (dql3);

		
		\draw[loosely dashed, rounded corners] (-5.75,-5.25) rectangle (9.5, -2.25) {};
		\node[] at (-5.75 + 0.4, -5.25 + 0.3) {$G_3$};

		\draw[loosely dashed, rounded corners] (-5,-4.5) rectangle (-3.5, -3) {};
		\node[rectangle, minimum size=0.8cm] (j1) at (-4.25,-3.75){$Q_\problemParam$};

		\draw[] (o) to [] ((-5.375,-3.75);
		\draw[-stealth, shorten >=1pt,auto] (-5.375,-3.75) to [] ((-5,-3.75);

		
		\node[draw, circle, minimum size=0.8cm] (c1) at (-2.5,-3.75){$v_3$};
		\node[draw, rectangle, minimum size=0.8cm] (y1) at (-1.5,-4.75){$\neg y_1$};
		\node[draw, rectangle, minimum size=0.8cm] (ny1) at (-1.5,-2.75){$y_1$};
		\node[draw, circle, minimum size=0.8cm] (c2) at (-0.5,-3.75){};
		\node[minimum size=0.8cm] (ci1) at (0.5,-4.75){$\dots$};
		\node[minimum size=0.8cm] (ci2) at (0.5,-2.75){$\dots$};
		\node[draw, circle, minimum size=0.8cm] (cm) at (1.5,-3.75){};
		\node[draw, rectangle, minimum size=0.8cm] (ym) at (2.5,-4.75){$\neg y_n$};
		\node[draw, rectangle, minimum size=0.8cm] (nym) at (2.5,-2.75){$y_n$};
		
		\draw[-stealth, shorten >=1pt,auto] ((-3.5,-3.75) to [] node []{} (c1);

		\draw[-stealth, shorten >=1pt,auto] (c1) to [] node []{} (y1);
		\draw[-stealth, shorten >=1pt,auto] (c1) to [] node []{} (ny1);
		\draw[-stealth, shorten >=1pt,auto] (y1) to [] node []{} (c2);
		\draw[-stealth, shorten >=1pt,auto] (ny1) to [] node []{} (c2);
		\draw[-stealth, shorten >=1pt,auto] (c2) to [] node []{} (ci1);
		\draw[-stealth, shorten >=1pt,auto] (c2) to [] node []{} (ci2);
		\draw[-stealth, shorten >=1pt,auto] (ci1) to [] node []{} (cm);
		\draw[-stealth, shorten >=1pt,auto] (ci2) to [] node []{} (cm);
		\draw[-stealth, shorten >=1pt,auto] (cm) to [] node []{} (ym);
		\draw[-stealth, shorten >=1pt,auto] (cm) to [] node []{} (nym);			
		
		
		\node[draw, rectangle, minimum size=0.8cm] (d1) at (3.5,-3.75){};
		\node[draw, rectangle, minimum size=0.8cm] (x31) at (4.5,-4.75){$\neg x_1$};
		\node[draw, rectangle, minimum size=0.8cm] (nx31) at (4.5,-2.75){$x_1$};
		\node[draw, rectangle, minimum size=0.8cm] (d2) at (5.5,-3.75){};
		\node[minimum size=0.8cm] (d2i1) at (6.5,-4.75){$\dots$};
		\node[minimum size=0.8cm] (d2i2) at (6.5,-2.75){$\dots$};
		\node[draw, rectangle, minimum size=0.8cm] (dm) at (7.5,-3.75){};
		\node[draw, rectangle, minimum size=0.8cm] (x3m) at (8.5,-4.75){$\neg x_m$};
		\node[draw, rectangle, minimum size=0.8cm] (nx3m) at (8.5,-2.75){$x_m$};

		\draw[-stealth, shorten >=1pt,auto] (ym) to [] node []{} (d1);		
		\draw[-stealth, shorten >=1pt,auto] (nym) to [] node []{} (d1);	
		\draw[-stealth, shorten >=1pt,auto] (d1) to [] node []{} (x31);
		\draw[-stealth, shorten >=1pt,auto] (d1) to [] node []{} (nx31);
		\draw[-stealth, shorten >=1pt,auto] (x31) to [] node []{} (d2);
		\draw[-stealth, shorten >=1pt,auto] (nx31) to [] node []{} (d2);
		\draw[-stealth, shorten >=1pt,auto] (d2) to [] node []{} (d2i1);
		\draw[-stealth, shorten >=1pt,auto] (d2) to [] node []{} (d2i2);
		\draw[-stealth, shorten >=1pt,auto] (d2i1) to [] node []{} (dm);
		\draw[-stealth, shorten >=1pt,auto] (d2i2) to [] node []{} (dm);
		\draw[-stealth, shorten >=1pt,auto] (dm) to [] node []{} (x3m);
		\draw[-stealth, shorten >=1pt,auto] (dm) to [] node []{} (nx3m);
		\draw[-stealth,shorten >=1pt,auto,in=-15,out=15,looseness=5.2] (x3m) edge [] node {} (x3m);
		\draw[-stealth,shorten >=1pt,auto,in=-15,out=15,looseness=6] (nx3m) edge [] node {} (nx3m);
		
		\end{tikzpicture}
		}%
	\caption{The arena $G$ used in the reduction from the \succinctSetCoverAb{} for safety \gamesAb{}.}
	\label{fig:sscp_game_long_safety}
	\Description{Figure 7. Fully described in the text.}
\end{figure}
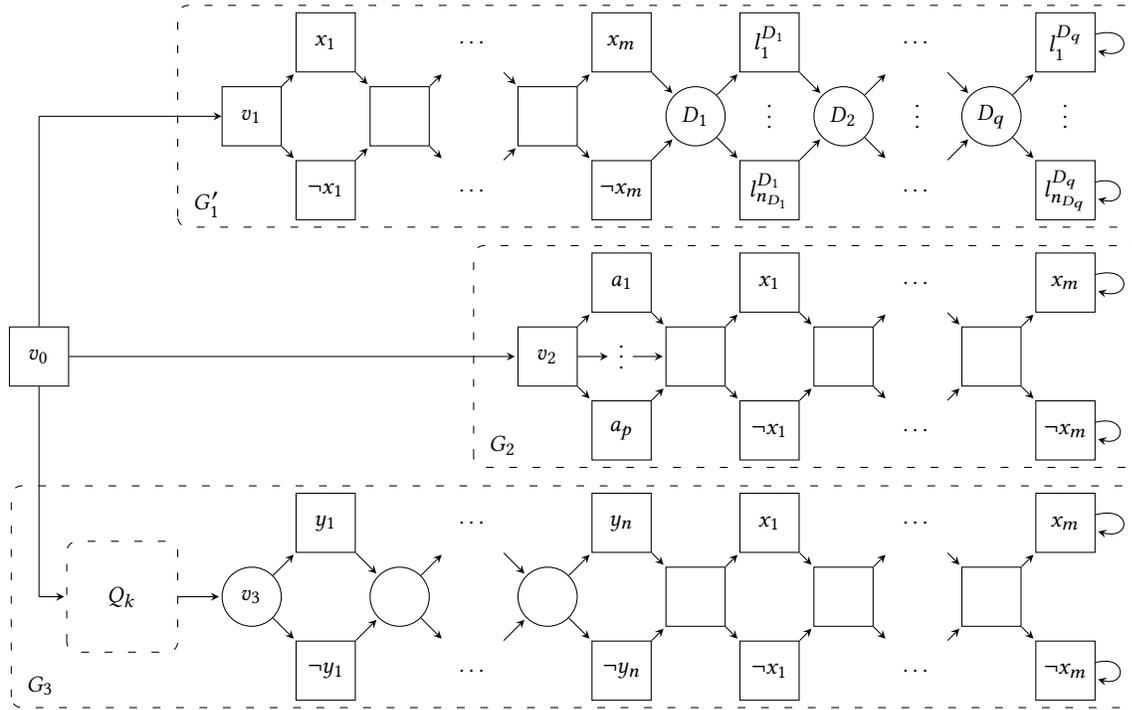

\paragraph{Modified Sub-Arena $G'_1$.}
We assume that each clause $C_j$ of $\phi$, with $j \in \{1, \dots, p\}$, is of the form $C_j = l_1^{C_j} \vee \ldots \vee l_{n_{C_j}}^{C_j}$ and each clause $D_j$ of $\psi$, with $j \in \{1, \dots, q\}$, is of the form $D_j = l_1^{D_j} \vee \ldots \vee l_{n_{D_j}}^{D_j}$.
In the sub-arena $G'_1$, compared to Figure~\ref{sscp_game_long}, we see a new part where Player~$0$ can select for each clause $D_j$ of $\psi$ some literal $l_i^{D_j}$ among the literals composing $D_j$.

\paragraph{Objectives.}
Objectives $\Omega_0$, $\Omega_1$, and $\ObjPlayer{x_i}, \ObjPlayer{\neg x_i}$ with $i \in \{1, \dots, m\}$ retain the same function as in the reduction for reachability objectives, and we therefore define those objectives as the equivalent safety objective to their reachability counterpart in the previous subsection.

We reason differently about the objectives used for the clauses for the following reason. While it was possible to use a single objective per clause in the case of reachability objectives, this reasoning is not valid for safety objectives. In the former case, it sufficed to check that a vertex corresponding to a literal of a clause was reached to ascertain that clause was satisfied. In the latter case, safety objectives only allow us to reason about the literals which are not visited. This does not allow us to define a single objective to ensure that some literal of the clause is visited. We instead associate one safety objective with \emph{each literal} $l$ of each clause $C_j$ and $D_j$. Therefore, there are $n_{C_j}$ (resp. $n_{D_j}$) safety objectives for clause $C_j$ (resp. $D_j$).


The game is played between Player~$0$ with safety objective $\ObjPlayer{0}$ and Player~$1$ with $\nbrObjectives = 1 + 2 \cdot m + \sum_{j=1}^{p} n_{C_j}	 + \sum_{j=1}^{q} n_{D_j}$ safety objectives. The objectives are defined as follows where $V$ denotes the set of vertices of $G$.

\begin{itemize}
    \item The safe set for objective $\ObjPlayer{0}$ of Player~$0$ and objective $\ObjPlayer{1}$ of Player~$1$ is $V \setminus \{v_1\}$ (which is the safety equivalent of $\reach{\{v_2,v_3\}}$ used for reachability \gamesAb{}).
   
    \item The safe set for objective $\ObjPlayer{x_i}$ with $i \in \{1, \dots, m\}$ is equal to $V$ from which we remove all vertices in $G$ labeled by $\neg x_i$ (which is the safety equivalent of reaching some vertex labeled $x_i$). Similarly, the safe set for objective $\ObjPlayer{\neg x_i}$ consists in $V$ from which we remove all vertices labeled by $x_i$.
    
    \item Consider the objective $\ObjPlayer{l^{C_j}_i}$ with $j \in \{1, \dots, p\}$ and $i \in \{1, \dots, n_{C_j}\}$, corresponding to literal $l^{C_j}_i$ of clause $C_j$. Its safe set is equal to $V$ from which we remove: 
    \begin{itemize}
        \item the vertices in $G'_1 \cup G_3$ labeled by $\neg x_u$ (resp. $x_u$) if $l^{C_j}_i = x_u$ (resp. $l^{C_j}_i = \neg x_u$) for some $u \in \{1, \dots, m\}$ (which is the safety equivalent of reaching vertex $l^{C_j}_i$, therefore satisfying the literal $l^{C_j}_i$),
        \item and vertex $a_j$ in $G_2$.
    \end{itemize}

    \item Consider the objective $\ObjPlayer{l^{D_j}_i}$ with $j \in \{1, \dots, q\}$ and $i \in \{1, \dots, n_{D_j}\}$, corresponding to literal $l^{D_j}_i$ of clause $D_j$. Its safe set if equal to $V$ from which we remove: 
    \begin{itemize}
        \item all vertices $l^{D_j}_\ell$ of $G'_1$ with $\ell \neq i$ (which is the safety equivalent of reaching vertex $l^{D_j}_i$, meaning the objective is satisfied in plays in $G'_1$ only if Player~$0$ chooses to visit $l^{D_j}_i$ among the literals composing $D_j$),
        \item the vertex $\neg x_u$ in $G_3$ (resp. $x_u$, $\neg y_u$, $y_u$) if $l^{D_j}_i = x_u$ (resp. $\neg x_u$, $y_u$, $\neg y_u$) for some $u$ (as we did above for objective $\ObjPlayer{l^{C_j}_i}$, but in this case with variables in $X \cup Y$).
    \end{itemize}
    Notice that the safe set for objective $\ObjPlayer{l^{D_j}_i}$ contains every vertex in $G_2$, meaning that the objective $\ObjPlayer{l^{D_j}_i}$ is always satisfied in $G_2$.
\end{itemize}

\paragraph{Correctness of the Reduction.}
We ascertain the correctness of this reduction using the following properties that are the counterpart of the lemmas proved in the previous subsection for reachability \gamesAb{}, leading to the proof of Theorem~\ref{thm:nexptimehard} for safety \gamesAb{}.

\begin{itemize}
    \item Since we now have an objective per literal per clause of $\psi$, Player~$0$ must now make sure that in his solution to the \problemAb{}, plays in $G_3$ satisfy at least one literal per clause of $\phi$ (in the reduction for reachability objectives, there was one objective per clause and we required all those objectives to be satisfied). To ensure this, for each valuation of the variables in $X$, the rightmost part of $G'_1$ forces Player~$0$ to select one literal per clause of $\psi$ which he announces will be satisfied in his play in $G_3$ covering that valuation.

    \item The plays in $G'_1$ have a payoff of the form $(0, val, sat(\phi, val), w_1, \dots, w_q)$ where $val$ is a valuation of the variables in $X$ expressed as a vector of Booleans for objectives $\ObjPlayer{x_1}$ to $\ObjPlayer{\neg x_m}$, $sat(\phi, val)$ is the vector of Booleans for objectives $\ObjPlayer{l^{C_1}_1}, \dots, \ObjPlayer{l^{C_p}_{n_{C_p}}}$ corresponding to the literals of the clauses of $\phi$ satisfied by that valuation and each $w_j$ is a vector of $n_{D_j}$ Booleans in which only a single Boolean is 1. All plays in $G'_1$ are lost by Player~$0$.

    \item The plays in $G_2$ have a payoff of the form $(1, val, vec, 1, \dots, 1)$ where $val$ is a valuation of the variables in $X$ expressed as a vector of Booleans for objectives $\ObjPlayer{x_1}$ to $\ObjPlayer{\neg x_m}$ and $vec$ is a vector of Booleans for objectives $\ObjPlayer{l^{C_1}_1}, \dots, \ObjPlayer{l^{C_p}_{n_{C_p}}}$ in which all of them except the ones for some clause $C_j$ are satisfied. All plays in $G_2$ are won by Player~$0$.

    \item Let $\sigma_0$ be a strategy of Player~$0$. The set of payoffs of plays in $G'_1$ that are \paretoOptimal{} when considering $G'_1 \cup G_2$ is equal to the set of payoffs of plays in $G'_1$ whose valuation of $X$ satisfy $\phi$. Those payoffs are the problematic payoffs for which Player~$0$ has to create strictly larger payoffs in $G_3$ for strategy $\sigma_0$ to be a solution to the problem.
    
    \item In order to create a play in $G_3$ with a payoff $r'$ that is strictly larger than a problematic payoff $r = (0, val, sat(\phi, val), w_1, \dots, w_q)$ in $G'_1$, $\sigma_0$ must be able to choose a valuation of $Y$ which together with the valuation $val$ of $X$ satisfies $\psi$ in the same way as announced in $r$. That is, it satisfies at least the announced literal $l_i^{D_j}$ in $w_j$ for each $j \in  \{1, \dots, q\}$. If this is the case, the payoff of the corresponding play is therefore $r' = (1, val, sat(\phi, val), w'_1, \dots, w'_q)$ with $w_i' \geq w_i$ for $i\in \{1, \dots, q\}$ and it follows that $r < r'$.
\end{itemize}


\subsection{\textsf{NEXPTIME}-Hardness for SP Games with Prefix-Independent Objectives}

We finally provide the proof of Theorem~\ref{thm:nexptimehard} for \gamesAb{} with prefix-independent objectives (except for B\"uchi and co-B\"uchi objectives for which we have previously shown the \npHard ness). We first consider the case of parity objectives, the proof of which follows arguments related to the ones used in the reduction for reachability and safety \gamesAb{}. At the end of the subsection, we then explain how the \nexptime-hardness result for parity \gamesAb{} easily leads to the hardness result for \gamesAb{} with Boolean B\"uchi, Muller, Streett, and Rabin objectives, by using results from Section~\ref{sec:link}.

The proof of the \nexptime-hardness of the \problemAb{} for parity \gamesAb{} again uses a reduction from the \succinctSetCoverAb{} in which we construct a specific arena $G''$. However, we here use a different approach to define problematic payoffs which are to be covered in the reduction. While letting Player~$0$ select $\problemParam$ valuations of the variables in $Y$ (corresponding to $\problemParam$ sets) in $G_3$ was straightforward in the previous reductions, doing so for parity objectives is more difficult. This is because, in order to make the reduction work with parity objectives, plays in $G''$ need to have a repeating pattern. If such a pattern appeared in $G_3$ (for example with edges from $x_m$ and $\neg x_m$ to $v_3$), it would be possible for Player~$0$ to select more than $\problemParam$ valuations of $Y$. Indeed, the number of plays in $G_3$ consistent with any of his strategies would be infinite due to the presence of vertices of Player~$1$\footnote{The proof of the \nexptime-hardness of the \problemAb{} for parity objectives provided in \cite{DBLP:conf/concur/BruyereRT21} is erroneous due to this argument.}. The arena $G''$ used for the reduction is depicted in Figure~\ref{fig:succinctparity} and contains four sub-arenas $G''_1$, $G''_2$, $G''_3$ and $G''_4$.

\paragraph{Intuition of the Reduction.} The intuition behind the reduction is the following. In sub-arena $G''_1$, Player~$1$ selects a valuation of the variables in $X$ and Player~$0$ then selects a valuation of the variables in $Y$, with this process being repeated to form a play in $G''_1$. The payoff of that play corresponds to a valuation of the variables in $X$ and $Y$, together with the clauses of $\phi$ and $\psi$ satisfied by those valuations. In addition, the objective of Player~$0$ is not satisfied in those plays. Sub-arenas $G''_2$ and $G''_3$ are devised such that the payoff of plays entering them is strictly larger than that of plays in $G''_1$ which do not satisfy $\phi$ or which are not proper valuations of $X$. As the objective of Player~$0$ is satisfied in those plays, and in order for Player~$0$ to have a solution to the \problemAb{}, it remains to make sure that plays in $G''_1$ whose valuation of $X$ is proper and satisfies $\phi$ are not \paretoOptimal{}. This is only the case when the instance of the \succinctSetCoverAb{} is positive. In that case, given any valuation of $X$ which satisfies $\phi$ selected by Player~$1$ in $G''_1$, Player~$0$ can select a valuation of $Y$ which together with this valuation of $X$ satisfies $\psi$. It holds that he can do so with $k$ different valuations of $Y$. Then, in $G''_4$, Player~$0$ is allowed exactly $k$ different plays in which he selects those valuations of $Y$, with their payoff being strictly larger than that of plays in $G''_1$ which satisfy both $\phi$ and $\psi$. 

\paragraph{Structure of a Payoff.} 
We now detail the objectives used in the reduction and the corresponding structure of a payoff in $G''$. The game is played between Player~$0$ with parity objective $\ObjPlayer{0}$ and Player~$1$ with $\nbrObjectives = 1 + 2 \cdot m + 2 \cdot n + p + q$ parity objectives (recall that there are $n$ variables in $X$, $m$ variables in $Y$, $p$ clauses in $\phi$ and $q$ clauses in $\psi$). The payoff of a play in $G''$ therefore consists in a vector of $\nbrObjectives$ 
Booleans for the following objectives:
$$(\Omega_{1}, \Omega_{x_1}, \Omega_{\neg x_1}, \dots, \Omega_{x_m}, \Omega_{\neg x_m}, \Omega_{y_1}, \Omega_{\neg y_1}, \dots, \Omega_{y_n}, \Omega_{\neg y_n}, \Omega_{C_1}, \dots, \Omega_{C_p},\Omega_{D_1}, \dots,\Omega_{D_q}).$$
Notice that, compared to the reduction for reachability objectives, we now also include objectives for the variables in $Y$. The priority function $c$ of objective $\Omega_0 = \Omega_1$ is such that $c(v_1) = 1$ and that every other vertex has priority $2$. Since plays in $G''_1$ visits $v_1$ infinitely often, they are the only ones not to satisfy $\Omega_0$ nor $\Omega_1$. We detail the priority function of the other objectives for the vertices in each sub-arena later.

\begin{figure}
	\centering
	\resizebox{\textwidth}{!}{%
		
		\begin{tikzpicture}
		
		\node[draw, rectangle, minimum size=0.8cm] (o) at (-4.25,6.25){$v_0$};
		
		
		\draw[loosely dashed, rounded corners] (-5.1,1) rectangle (-2.5, 4.25) {};
        \node[] at (-5.1 + 0.4, 1 + 0.3) {$G''_2$};

		\node[draw, rectangle, minimum size=0.8cm] (b1) at (-3.5,3.75){$a_1$};
		\node[draw, rectangle, minimum size=0.8cm] (x11) at (-4.5,2.75){$v_2$};
		\node[draw, rectangle, minimum size=0.8cm] (b2) at (-3.5,1.75){$a_m$};
		\node[rectangle, minimum size=0.5cm] (b22) at (-3.5,2.75){$\vdots$};

		\draw[-stealth, shorten >=1pt,auto] (x11) to [] node []{} (b2);
		\draw[-stealth, shorten >=1pt,auto] (x11) to [] node []{} (b22);
		\draw[-stealth, shorten >=1pt,auto] (x11) to [] node []{} (b1);
		
		\draw[-stealth, shorten >=1pt,auto] (o.240) to [] node []{} (x11);
		
		\draw[-stealth,shorten >=1pt,auto,in=-15,out=15,looseness=6] (b1) edge [ ] node {} (b1);
		\draw[-stealth,shorten >=1pt,auto,in=-15,out=15,looseness=6] (b2) edge [ ] node {} (b2);
		
		
        \draw[loosely dashed, rounded corners] (-2.1,1) rectangle (0.5, 4.25)  {};
        \node[] at (-2.1 + 0.4, 1 + 0.3) {$G''_3$};

		\node[draw, rectangle, minimum size=0.8cm] (bm) at (-0.5,3.75){$b_1$};
		\node[draw, rectangle, minimum size=0.8cm] (x1m) at (-1.5,2.75){$v_3$};
		\node[draw, rectangle, minimum size=0.8cm] (end1) at (-0.5,1.75){$b_p$};
		\node[rectangle, minimum size=0.5cm] (l1) at (-0.5,2.75){$\vdots$};

		\draw[-stealth, shorten >=1pt,auto] (x1m) to [] node []{} (bm);
		\draw[-stealth, shorten >=1pt,auto] (x1m) to [] node []{} (l1);
		\draw[-stealth, shorten >=1pt,auto] (x1m) to [] node []{} (end1);
		
		\draw[-stealth,shorten >=1pt,auto,in=-15,out=15,looseness=6] (bm) edge [ ] node {} (bm);
		\draw[-stealth,shorten >=1pt,auto,in=-15,out=15,looseness=6] (end1) edge [ ] node {} (end1);

		\draw[] (o) to [] (-4.25, 4.4);
		\draw[] (-4.25, 4.4) to [] (-1.5, 4.4);
		\draw[-stealth, shorten >=1pt,auto] (-1.5, 4.4) to [] (x1m);

		\draw[loosely dashed, rounded corners] (-3.1, 4.75) rectangle (10.1, 8) {};
		\node[] at (-3.1 + 0.4, 4.75 + 0.3) {$G''_1$};
		
		\node[draw, rectangle, minimum size=0.8cm] (a1) at (-2.5,6.25){$v_1$};
		\node[draw, rectangle, minimum size=0.8cm] (x21) at (-1.5,7.25){$x_1$};
		\node[draw, rectangle, minimum size=0.8cm] (nx21) at (-1.5,5.25){$\neg x_1$};
		\node[draw, rectangle, minimum size=0.8cm] (a2) at (-0.5,6.25){};
		\node[minimum size=0.8cm] (a2i1) at (0.5,7.25){$\dots$};
		\node[minimum size=0.8cm] (a2i2) at (0.5,5.25){$\dots$};
		\node[draw, rectangle, minimum size=0.8cm] (am) at (1.5,6.25){};
		\node[draw, rectangle, minimum size=0.8cm] (x2m) at (2.5,7.25){$x_m$};
		\node[draw, rectangle, minimum size=0.8cm] (nx2m) at (2.5,5.25){$\neg x_m$};
		\node[draw, circle, minimum size=0.8cm] (d1) at (3.5,6.25){$v'_1$};
		\node[draw, circle, minimum size=0.8cm] (d1l3) at (4.5,5.25){$\neg y_1$};
		\node[rectangle, minimum size=0.8cm] (d1l2) at (4.5,6.35){};
		\node[draw, circle, minimum size=0.8cm] (d1l1) at (4.5,7.25){$y_1$};
		\node[draw, circle, minimum size=0.8cm] (d2) at (5.5,6.25){};
		\node[rectangle, minimum size=0.8cm] (d2l3) at (6.5,5.25){$\dots$};
		\node[rectangle, minimum size=0.8cm] (d2l2) at (6.5,6.35){};
		\node[rectangle, minimum size=0.8cm] (d2l1) at (6.5,7.25){$\dots$};
		\node[draw, circle, minimum size=0.8cm] (dq) at (7.5,6.25){};
		\node[draw,circle, minimum size=0.8cm] (dql3) at (8.5,5.25){$\neg y_n$};
		\node[rectangle, minimum size=0.8cm] (dql2) at (8.5,6.35){};
		\node[draw,circle, minimum size=0.8cm] (dql1) at (8.5,7.25){$y_n$};
		\node[draw, circle, minimum size=0.8cm] (endg1) at (9.5,6.25){$v''_1$};

		\draw[-stealth] (o)  to [] (a1);
		\draw[-stealth, shorten >=1pt,auto] (a1) to [] node []{} (x21);
		\draw[-stealth, shorten >=1pt,auto] (a1) to [] node []{} (nx21);
		\draw[-stealth, shorten >=1pt,auto] (x21) to [] node []{} (a2);
		\draw[-stealth, shorten >=1pt,auto] (nx21) to [] node []{} (a2);
		\draw[-stealth, shorten >=1pt,auto] (a2) to [] node []{} (a2i1);
		\draw[-stealth, shorten >=1pt,auto] (a2) to [] node []{} (a2i2);
		\draw[-stealth, shorten >=1pt,auto] (a2i1) to [] node []{} (am);
		\draw[-stealth, shorten >=1pt,auto] (a2i2) to [] node []{} (am);
		\draw[-stealth, shorten >=1pt,auto] (am) to [] node []{} (x2m);
		\draw[-stealth, shorten >=1pt,auto] (am) to [] node []{} (nx2m);		
		\draw[-stealth, shorten >=1pt,auto] (x2m) to [] node []{} (d1);
		\draw[-stealth, shorten >=1pt,auto] (nx2m) to [] node []{} (d1);
		\draw[-stealth, shorten >=1pt,auto] (d1) to [] node []{} (d1l1);
		\draw[-stealth, shorten >=1pt,auto] (d1) to [] node []{} (d1l3);
		\draw[-stealth, shorten >=1pt,auto] (d1l1) to [] node []{} (d2);
		\draw[-stealth, shorten >=1pt,auto] (d1l3) to [] node []{} (d2);
		\draw[-stealth, shorten >=1pt,auto] (d2) to [] node []{} (d2l1);
		\draw[-stealth, shorten >=1pt,auto] (d2) to [] node []{} (d2l3);
		\draw[-stealth, shorten >=1pt,auto] (d2l1) to [] node []{} (dq);
		\draw[-stealth, shorten >=1pt,auto] (d2l3) to [] node []{} (dq);		
		\draw[-stealth, shorten >=1pt,auto] (dq) to [] node []{} (dql1);
		\draw[-stealth, shorten >=1pt,auto] (dq) to [] node []{} (dql3);
		\draw[-stealth, shorten >=1pt,auto] (dql1) to [] node []{} (endg1);
		\draw[-stealth, shorten >=1pt,auto] (dql3) to [] node []{} (endg1);
		
		\draw[] (endg1) to [] (9.5,7.8125);
		\draw[] (9.5,7.8125) to [] (-2.5,7.8125);
		\draw[-stealth, shorten >=1pt,auto] (-2.5,7.8125) to [] (a1);

		
		\draw[loosely dashed, rounded corners] (0.9,1) rectangle (10.1, 4.25) {};
		\node[] at (0.9 + 0.4,1 + 0.3) {$G''_4$};

		\draw[loosely dashed, rounded corners] (1.25,2) rectangle (2.75, 3.5) {};
		\node[rectangle, minimum size=0.8cm] (j1) at (2,2.75){$Q_\problemParam$};

		\draw[] (o.300) to [] (-4, 4.6);
		\draw[] (-4, 4.6) to [] (2, 4.6);
		\draw[-stealth, shorten >=1pt,auto] (2, 4.6) to [] (2, 3.5);

		\node[draw,circle, minimum size=0.8cm] (d1) at (3.5,2.75){$v_4$};
		\node[draw, circle, minimum size=0.8cm] (x31) at (4.5,1.75){$\neg y_1$};
		\node[draw, circle, minimum size=0.8cm] (nx31) at (4.5,3.75){$y_1$};
		\node[draw, circle, minimum size=0.8cm] (d2) at (5.5,2.75){};
		\node[minimum size=0.8cm] (d2i1) at (6.5,1.75){$\dots$};
		\node[minimum size=0.8cm] (d2i2) at (6.5,3.75){$\dots$};
		\node[draw, circle, minimum size=0.8cm] (dm) at (7.5,2.75){};
		\node[draw, circle, minimum size=0.8cm] (x3m) at (8.5,1.75){$\neg y_n$};
		\node[draw, circle, minimum size=0.8cm] (nx3m) at (8.5,3.75){$y_n$};
		\node[draw, circle, minimum size=0.8cm] (eng3) at (9.5,2.75){};
				
		\draw[-stealth, shorten >=1pt,auto] (d1) to [] node []{} (x31);
		\draw[-stealth, shorten >=1pt,auto] (d1) to [] node []{} (nx31);
		\draw[-stealth, shorten >=1pt,auto] (x31) to [] node []{} (d2);
		\draw[-stealth, shorten >=1pt,auto] (nx31) to [] node []{} (d2);
		\draw[-stealth, shorten >=1pt,auto] (d2) to [] node []{} (d2i1);
		\draw[-stealth, shorten >=1pt,auto] (d2) to [] node []{} (d2i2);
		\draw[-stealth, shorten >=1pt,auto] (d2i1) to [] node []{} (dm);
		\draw[-stealth, shorten >=1pt,auto] (d2i2) to [] node []{} (dm);
		\draw[-stealth, shorten >=1pt,auto] (dm) to [] node []{} (x3m);
		\draw[-stealth, shorten >=1pt,auto] (dm) to [] node []{} (nx3m);
		\draw[-stealth, shorten >=1pt,auto] (x3m) to [] node []{} (eng3);
		\draw[-stealth, shorten >=1pt,auto] (nx3m) to [] node []{} (eng3);

		\draw[] (eng3) to [] (9.5,1.1874);
		\draw[] (9.5,1.1874) to [] (3.5,1.1874);
		\draw[-stealth,shorten >=1pt,auto] (3.5,1.1874) to [] (d1);
		
		\draw[-stealth, shorten >=1pt,auto] (2.75,2.75) to [] node []{} (d1);
		
		\end{tikzpicture}
		}%
	\caption{The arena $G'$ used in the reduction from the \succinctSetCoverAb{} for parity \gamesAb{}.}
	\label{fig:succinctparity}
	\Description{Figure 9. Fully described in the text.}
\end{figure}
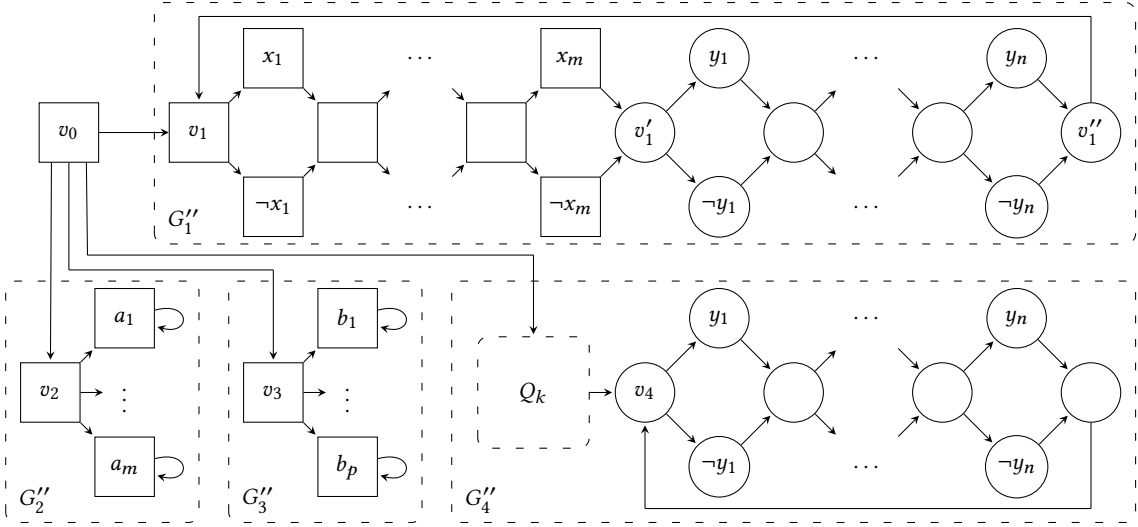

\paragraph{Sub-Arena $G''_1$.}
The intuition behind sub-arena $G''_1$ is as follows (we formally define the objectives in that sub-arena later on). Consider a play $\rho$ in $G''_1$ where Player~$0$ and Player~$1$ always make the same choice of literals for the variables in $X$ and $Y$. The payoff of $\rho$ corresponds to a valuation of these variables expressed using the objectives for their literals. In addition, the objective for clause $C$ of $\phi$ is satisfied in the payoff of $\rho$ if and only if the valuation of $X$ satisfies $C$. The objective for clause $D$ of $\psi$ is satisfied if and only if together, valuations $X$ and $Y$ \emph{falsify} $D$. Since plays in $G''_1$ are lost by Player~$0$, the choices he makes in $G''_1$ aim at making sure that the payoff of these plays is not \paretoOptimal{} when considering the whole arena $G''$. We will see that plays in $G''_1$ whose valuations satisfy $\phi$ and falsify $\psi$ do not have a payoff that is strictly smaller than some other in $G''$. This forces Player~$0$ in $G''_1$ to select valuations of $Y$ which falsifies the objectives for the clauses of $\psi$ (and therefore satisfies $\psi$) whenever $\phi$ is also satisfied.

\paragraph{Sub-Arena $G''_2$.}
All vertices in $G''_2$ belong to Player~$1$. There are $m$ possible plays in $G''_2$, one for each variable $x$ of $X$. Given $i \in \{1, \dots, m\}$, we define the priority function $c$ of objectives $\Omega_{x_i}$ (resp. $c'$ of $\Omega_{\neg x_i}$) for vertices in $G''_2$ such that $c(a_i) = 1$ (resp. $c'(a_i) = 1$) and such that every other vertex in $G''_2$ has priority 2 according to $c$ (resp. $c'$). The priority function $c$ of every other objective is such that every vertex of $G''_2$ has priority 2. Let us consider the play $\rho = v_0 v_2 a_i^\omega$ in $G''_2$ for variable $x_i$, it holds that $\rho$ satisfies \emph{(i)} the objective of Player~$0$ and objective $\Omega_1$ of Player~$1$ and \emph{(ii)} every other objective except $\Omega_{x_i}$ and $\Omega_{\neg x_i}$.
\begin{lemma}
    \label{lem:eliminate_unstable}
	Plays in $G''_2$ are consistent with any strategy of Player~$0$ and their payoff is of the form $(1, s_1, \dots, s_{m}, 1, \dots, 1)$ such that $s_i = (0,0)$ for some $i \in \{1, \dots, m\}$ and $s_j = (1, 1)$ for $j \neq i$.
\end{lemma}

\paragraph{Sub-Arena $G''_3$.}
All vertices in $G''_3$ belong to Player~$1$. There are $p$ possible plays in $G''_3$, one for each clause $C$ of $\phi$. Given $i \in \{1, \dots, p\}$, we define the priority function $c$ of objectives $\Omega_{C_i}$ for vertices in $G''_3$ such that $c(b_i) = 1$ and such that every other vertex in $G''_3$ has priority 2 according to $c$. The priority function $c$ of every other objective is such that every vertex of $G''_3$ has priority 2. Let us consider the play $\rho = v_0 v_3 b_i^\omega$ in $G''_3$ for clause $C_i$, it holds that $\rho$ satisfies \emph{(i)} the objective of Player~$0$ and objective $\Omega_1$ of Player~$1$ and \emph{(ii)} every other objective except $\Omega_{C_i}$.
\begin{lemma}
    \label{lem:eliminate_not_phi}
	Plays in $G''_3$ are consistent with any strategy of Player~$0$ and their payoff is of the form $(1, 1, \dots, 1, 1, \dots, 1, r_1, \dots, r_p, 1, \dots, 1)$ such that $r_i = 0$ for some $i \in \{1, \dots, p\}$ and $r_j = 1$ for $j \neq i$.
\end{lemma}

\paragraph{$X$-Stability of Plays in Sub-Arena $G''_1$.}
We define the priority function $c$ of objective $\Omega_{x}$ (resp. $c'$ of $\Omega_{\neg x}$) for the vertices in $G''_1$ such that $c(x) = 2$, $c(\neg x) = 1$ (resp. $c'(\neg x) = 2$ and $c'(x) = 1$) for the vertices labelled $x$ and $\neg x$ in $G''_1$ and such that every other vertex in $G''_1$ has priority 2 according to $c$ (resp. $c'$).

Notice that in $G''_1$, Player~$1$ first decides to visit one literal $x$ or $\neg x$ for each variable $x \in X$. Following this, Player~$0$ decides to visit one literal $y$ or $\neg y$ for each variable $y \in Y$. This procedure is repeated infinitely often to form a play in $G''_1$. Between two visits of $v_1$ in a play in $G''_1$, the choices made by either player can be different. We call \emph{$X$-unstable} those plays which visit both $x$ and $\neg x$ infinitely often for some $x \in X$ and \emph{$X$-stable} those which visit infinitely often $x$ and finitely often $\neg x$ or infinitely often $\neg x$ and finitely often $x$ for each $x \in X$. Given an $X$-unstable play $\rho$, we write $u(\rho) \in \{1, \dots, m\}$ the smallest index $i$ such that both $x_i$ and $\neg x_i$ are visited infinitely often. It is direct to see that the payoff of an $X$-stable play for objectives $\Omega_{x_1}, \dots, \Omega_{\neg x_m}$ can be interpreted as a proper valuation of the variables in $X$ expressed as a vector of $2 \cdot m$ Booleans (as either the objective for literal $x_i$ or $\neg x_i$ is satisfied for each $i \in \{1, \dots, m\}$). We introduce the following lemma on the $X$-instability of plays in $G''_1$.
\begin{lemma}
    \label{lem:par_stability}
	Let $\sigma_0$ be a strategy for Player~$0$. Let $\rho$ be a play consistent with this strategy in $G''_1$. If $\rho$ is $X$-unstable, it does not have a \paretoOptimal{} payoff.
\end{lemma}
\begin{proof}
Let $\sigma_0$ be a strategy for Player~$0$. Let us consider the play $\rho = v_0 v_1 \dots$ consistent with $\sigma_0$ in $G''_1$ which is $X$-unstable. It holds that $x_{u(\rho)} \in X$ and that both $x_{u(\rho)}$ and $\neg x_{u(\rho)}$ are visited infinitely often in $\rho$. Let us consider the play $\rho' = v_0 v_2 a_{u(\rho)}^\omega$ which is consistent with $\sigma_0$ in $G''_2$. It is direct to see, given Lemma \ref{lem:eliminate_unstable}, that the payoff of $\rho$ is strictly smaller than that of $\rho'$. First, notice that $\Omega_1$ is satisfied in $\rho'$ and not in $\rho$. Second, $\Omega_{x_{u(\rho)}}$ and $\Omega_{\neg x_{u(\rho)}}$ are not satisfied in $\rho$ since both $x_{u(\rho)}$ and $\neg x_{u(\rho)}$ are visited infinitely often. They are not satisfied in $\rho'$ either by construction. Finally, since every other objective is satisfied in $\rho'$, it follows that $\payoff{\rho} < \payoff{\rho'}$.
\end{proof}

\paragraph{Satisfying $\phi$ in $G''_1$.}
We define the priority function $c$ of objective $\Omega_{C}$ for the vertices in $G''_1$ such that $c(l) = 2$ and $c(\neg l) = 3$ for the vertices labelled with any literal $l$ of the disjunction making up clause $C$ of $\phi$. Every other vertex in $G''_1$ has priority 3 according to $c$. It follows that the objective $\Omega_{C}$ corresponding to clause $C$ is satisfied if and only if \emph{some} literal of that clause is visited infinitely often. When considering an $X$-stable play $\rho$, this objective is satisfied if and only if the valuation of $X$ corresponding to $\rho$ satisfies $C$. We state the following property on $X$-stable plays in $G''_1$ whose valuation of $X$ does not satisfy $\phi$.
\begin{lemma}
    \label{lem:par_not_phi}
	Let $\sigma_0$ be a strategy for Player~$0$. Let $\rho$ be a play consistent with this strategy in $G''_1$. If $\rho$ is $X$-stable and such that its corresponding valuation of $X$ does not satisfy $\phi$, then it does not have a \paretoOptimal{} payoff.
\end{lemma}
\begin{proof}
Let $\sigma_0$ be a strategy for Player~$0$. Let us consider an $X$-stable play $\rho = v_0 v_1 \dots$ consistent with $\sigma_0$ in $G''_1$ such that its valuation of $X$ does not satisfy clause $C_i$ of $\phi$. It therefore holds that objective $\Omega_{C_i}$ is not satisfied in $\rho$. Let us consider the play $\rho' = v_0 v_3 b_{i}^\omega$ which is consistent with $\sigma_0$ in $G''_3$. It is direct to see, given Lemma \ref{lem:eliminate_not_phi}, that the payoff of $\rho$ is strictly smaller than that of $\rho'$. First, notice that $\Omega_1$ is satisfied in $\rho'$ and not in $\rho$. Second, $\Omega_{C_i}$ is not satisfied in $\rho$ given its corresponding valuation of $X$ nor in $\rho'$ by construction. Finally, since every other objective is satisfied in $\rho'$, it follows that $\payoff{\rho} < \payoff{\rho'}$.
\end{proof}

\paragraph{Problematic Payoffs.} Let $\sigma_0$ be a strategy for Player~$0$ in $G''$. We now consider the set of \paretoOptimal{} payoffs in $G''_1 \cup G''_2 \cup G''_3$ and state the following lemma.
\begin{lemma}
    \label{lem:parity_po}
	Let $\sigma_0$ be a strategy for Player~$0$. The \paretoOptimal{} payoffs in $G''_1 \cup G''_2 \cup G''_3$ which do not satisfy the objective of Player~$0$ are those of plays $\rho$ in $G''_1$ which are $X$-stable and such that their corresponding valuation of $X$ satisfies $\phi$.
\end{lemma}
\begin{proof}
First, notice that all plays in $G''_2$ and $G''_3$ satisfy the objective of Player~$0$ while none of the plays in $G''_1$ do. Second, by Lemma \ref{lem:par_stability}, $X$-unstable plays don't have a \paretoOptimal{} payoff. Third, by Lemma \ref{lem:par_not_phi}, $X$-stable plays in $G''_1$ whose valuation of $X$ does not satisfy $\phi$ don't have a \paretoOptimal{} payoff.
\end{proof}
The plays mentioned in Lemma \ref{lem:parity_po} are \paretoOptimal{} when considering $G''_1 \cup G''_2 \cup G''_3$ and do not satisfy the objective of Player~$0$. They must therefore not be \paretoOptimal{} when considering the entire arena $G''$ for $\sigma_0$ to be a solution to the problem.

\paragraph{Sub-Arena $G''_4$.} 
Sub-arena $G''_4$ starts with the gadget described in Subsection \ref{subsec:nexptime_reach}, which creates exactly $k$ different paths to $v_4$. It follows that given any strategy of Player~$0$, there are exactly $k$ different plays consistent with this strategy in $G''_4$. The priority function in $G''_4$ of each objective is defined as follows. The priority function $c$ of objective $\Omega_{y}$ (resp. $c'$ of objective $\Omega_{\neg y}$) for the vertices in $G''_4$ is such that $c(y) = 2$, $c(\neg y) = 1$ (resp. $c'(\neg y) = 2$, $c'(y) = 1$) for the vertices labelled $y$ and $\neg y$ in $G''_4$ and such that every other vertex in $G''_4$ has priority 2 according to $c$ (resp. $c'$). For all $i \in \{1, \dots, q\}$, the priority function $c$ of objective $\Omega_{D_i}$ is such that every vertex in $G''_4$ has priority $1$. The priority function $c$ of every other objective is such that every vertex of $G''_4$ has priority 2. It follows that a play in $G''_4$ satisfies \emph{(i)} objective $\Omega_0$ of Player~$0$ and $\Omega_1$ of Player~$1$, \emph{(ii)} either $\Omega_y$ or $\Omega_{\neg y}$ or neither of these objectives for each $y \in Y$, \emph{(iii)} none of the objectives for the clauses of $\psi$, and \emph{(iv)} every other objective. We state the following lemma on the payoff of plays in $G''_4$.
\begin{lemma}
    \label{lem:par_g_prime_four}
	Let $\sigma_0$ be a strategy for Player~$0$. There are exactly $k$ plays in $G''_4$ consistent with $\sigma_0$ and their payoff is of the form $(1, 1, \dots, 1, z_1, \dots, z_n, 1, \dots, 1, w_1, \dots, w_q)$ such that $z_i \in \{(0,0), (1,0), (0,1)\}$ for each $i \in \{1, \dots, n\}$ and $w_i = 0$ for each $i \in \{1, \dots, q\}$.
\end{lemma}

\paragraph{$Y$-Stability of Plays in $G''_1$.}
We define the priority function $c$ of objective $\Omega_{y}$ (resp. $c'$ of $\Omega_{\neg y}$) for the vertices in $G''_1$ such that $c(y) = 2$, $c(\neg y) = 3$, (resp. $c'(\neg y) = 2$ and $c'(y) = 3$) for the vertices labelled $y$ and $\neg y$ in $G''_1$ and such that every other vertex in $G''_1$ has priority 2 according to $c$ (resp. $c'$). We call \emph{$Y$-unstable} those plays which visit both $y$ and $\neg y$ infinitely often for some $y \in Y$ and \emph{$Y$-stable} those which visit infinitely often either $y$ or $\neg y$ for each $y \in Y$. Compared to sub-arena $G''_4$, we set $c(\neg y) = 3$ and $c'(y) = 3$ (instead of $c(\neg y) = 1$ and $c'(y) = 1$).

\begin{lemma}
    \label{lem:par_y_stability}
	Let $\sigma_0$ be a strategy for Player~$0$. Let $\rho$ be a play consistent with this strategy in $G''_1$ which is $X$-stable and such that its corresponding valuation of $X$ satisfies $\phi$. If $\rho$ is $Y$-unstable, then strategy $\sigma_0$ is not a solution to the \problemAb{}. 
\end{lemma}
\begin{proof}
Let $\sigma_0$ be a strategy for Player~$0$. Let $\rho$ be a play consistent with this strategy in $G''_1$ which is $X$-stable such that its corresponding valuation of $X$ satisfies $\phi$ and $Y$-unstable for some $y \in Y$. Both objectives $\Omega_y$ and $\Omega_{\neg y}$ are satisfied in the payoff of $\rho$, by definition of the objectives in $G''_1$. The payoff of this play, which is lost by Player~$0$, is therefore incomparable to that of every other payoff of plays in $G''$. This holds in particular for plays in $G''_4$ which satisfy either or none of objectives $\Omega_y$ and $\Omega_{\neg y}$ as described in Lemma~\ref{lem:par_g_prime_four}. Strategy $\sigma_0$ is therefore not a solution to the problem as $\rho$ has a \paretoOptimal{} payoff and is lost by Player~$0$.
\end{proof}

\paragraph{Satisfying $\psi$ in $G''_1$.}
We define the priority function $c$ of objective $\Omega_{D}$ for the vertices in $G''_1$ such that $c(l) = 1$ and $c(\neg l) = 2$ for the vertices labelled $l$ if $l$ is a literal of $D$ and such that every other vertex in $G''_1$ has priority 3 according to $c$. Notice that the way we reason about the clauses of $\psi$ using objectives is different to how we handle clauses $C$ of $\phi$. Objective $\Omega_{D}$ is satisfied in a play $\rho$ in $G''_1$ if and only if none of the literals of that clause are visited infinitely often in $\rho$. When considering an $X$-stable and $Y$-stable play $\rho$, this objective is satisfied if and only if the valuations of $X$ and $Y$ corresponding to $\rho$ does not satisfy the clause. Formula $\psi$ is satisfied by those valuations if none of the objectives $\Omega_{D_1}, \dots, \Omega_{D_q}$ are satisfied. We state the following lemma on the payoff of plays in $G''_1$.
\begin{lemma}
    \label{lem:par_payoff_g1}
	Let $\sigma_0$ be a strategy for Player~$0$ and let $\rho$ be a play in $G''_1$ consistent with $\sigma_0$, $X$-stable such that its valuation of $X$ satisfies $\phi$ and $Y$-stable such that its valuation of $X$ and $Y$ satisfies $\psi$. The payoff of $\rho$ is of the form $(0, s_1, \dots, s_m, z_1, \dots, z_n, 1, \dots, 1, w_1, \dots, w_q)$ such that $s_i \in \{(1,0), (0,1)\}$ for each $i \in \{1, \dots, m\}$, $z_i \in \{(1,0), (0,1)\}$ for each $i \in \{1, \dots, n\}$ and $w_i = 0$ for each $i \in \{1, \dots, q\}$.
\end{lemma}

Using all the arguments that we have established and summarized in lemmas above, we can finally demonstrate that our reduction is correct.
\begin{proposition}
    \label{prop:nexptime-hard-correct-par}
    An instance of the \succinctSetCoverAb{} is positive if and only if Player~$0$ has a strategy $\sigma_0$ that is a solution to the \problemAb{} in the corresponding parity \gameAb{} played on $G''$.
\end{proposition} 
\begin{proof}
Let us assume that the instance of the \succinctSetCoverAb{} is positive and show that we can create a strategy $\sigma_0$ which is a solution to the \problemAb{} in $G''$. Let us consider an history $h$ in $G''_1$ and let $val_X$ be the valuation of the variables in $X$ corresponding to the most recently visited vertices from $v_1$ to $v'_1$ in $h$. We devise $\sigma_0$ such that if $val_X$ satisfies $\phi$, Player~$0$ selects vertices from $v'_1$ to $v''_1$ corresponding to the valuation $val_Y$ of $Y$ used in the solution to the \succinctSetCoverAb{} such that $val_X \in \llbracket \psi_{val_Y} \rrbracket$. Recall that since the instance of the \succinctSetCoverAb{} is positive, $\sigma_0$ only requires $k$ different valuations of $Y$. If $val_X$ does not satisfy $\phi$, $\sigma_0$ selects one of the $k$ valuations of $Y$ arbitrarily. We define $\sigma_0$ in $G''_4$ such that after the $i$th of the $k$ possible histories from $v_0$ to $v_4$, $\sigma_0$ always selects the $i$th valuation of $Y$ used in the solution to the \succinctSetCoverAb{}. The resulting play is therefore $Y$-stable. Let us show that this strategy is a solution to the \problemAb{}. Given Lemma \ref{lem:parity_po}, only $X$-stable plays in $G''_1$ whose valuation of $X$ satisfies $\phi$ have a problematic payoff that cannot be \paretoOptimal{} when considering $G''$ in order for $\sigma_0$ to be a solution to the problem. Let us consider such a play $\rho$ consistent with $\sigma_0$. It holds that $\rho$ is $Y$-stable given the definition of $\sigma_0$. We easily show that its payoff is strictly smaller than that of the play $\rho'$ in $G''_4$ corresponding to the same valuation of $Y$ as the one associated to $\rho$ ($\rho'$ exists by definition of $\sigma_0$). Indeed, given Lemmas \ref{lem:par_g_prime_four} and \ref{lem:par_payoff_g1}, it holds that \emph{(i)} for each $i \in \{1, \dots, m\}$, both $\Omega_{x_i}$ and $\Omega_{\neg x_i}$ are satisfied in $\rho'$ and only either of these objectives is satisfied in $\rho$, \emph{(ii)} for each $i \in \{1, \dots, n\}$ the same objective $\Omega_{y_i}$ or $\Omega_{\neg y_i}$ is satisfied in $\rho$ and $\rho'$ as their valuations of $Y$ are identical \emph{(iii)} all objectives corresponding to the clauses of $\phi$ are satisfied in $\rho$ and $\rho'$ \emph{(iv)} none of the objectives corresponding to the clauses of $\psi$ are satisfied in $\rho$ nor $\rho'$. Therefore it holds that for each problematic play in $G''_1$ there is a play in $G''_4$ with a strictly larger payoff and that is winning for Player~$0$.

Given a strategy $\sigma_0$ which is solution to the \problemAb{} in $G''$ we derive a solution to the \succinctSetCoverAb{} as follows. As $\sigma_0$ is a solution, for every play $\rho$ in $G''_1$ which is $X$-stable and whose valuation of $X$ satisfies $\phi$, it is also $Y$-stable (as $\sigma_0$ is a solution to the problem and by Lemma \ref{lem:par_y_stability}). Together those valuations satisfy $\psi$, as if this were not the case some clause $D$ of $\psi$ would be unsatisfied, the corresponding objective $\Omega_D$ would be satisfied and plays with any such objective satisfied cannot be covered by plays in $G''_4$ given Lemma \ref{lem:par_g_prime_four}. It follows that for each valuation of $X$ which satisfies $\phi$, there exists a valuation of $Y$ which together with this valuation of $X$ satisfies $\psi$. In addition, those problematic plays can only be covered by plays in $G''_4$ by using $k$ different valuations of $Y$. It follows that these $k$ valuations of $Y$ are a solution to the \succinctSetCoverAb{}.
\end{proof}

\paragraph{Other Prefix-Independent Objectives}
The \nexptime{}-hardness of the \problemAb{} for parity objectives allows us to show the same result by reduction for Boolean B\"uchi objectives (by Proposition~\ref{prop_bb_encoding}), Streett and Rabin objectives (by Proposition~\ref{prop:parity_into_streett_rabin}), and Muller objectives (by Proposition~\ref{prop:parity_to_muller}). This completes the proof of Theorem~\ref{thm:nexptimehard}.

\paragraph{Co-B\"uchi Objectives}
Notice that the proof for parity \gamesAb{} requires to use three priorities overall. It is not clear how to reduce this number to two priorities, that is, to adapt the proof for co-B\"uchi \gamesAb{} (as we have shown the problem to be \npComplete{} for B\"uchi \gamesAb{}). Indeed notice that in the previous proof, we sometimes require to work with priorities $2$ and $3$ (resp. $1$ and $2$) corresponding to a B\"uchi (resp. co-B\"uchi) objective.

\section{Conclusion}
\label{sec:conclusion}
We have introduced in this paper the class of two-player \gamesAb{} with $\omega$-regular objectives and the \problemAb{} in those games. We have considered reachability and safety objectives as well as several classical prefix-independent $\omega$-regular objectives (B\"uchi, co-B\"uchi, Boolean B\"uchi, parity, Muller, Streett, and Rabin). We provided a reduction from \gamesAb{} to a two-player zero-sum game called the \challengerProverAb{} game in order to provide \FPT{} results on solving this problem. We then showed how the arena and the generic objective of this \challengerProverAb{} game can be adapted to specifically handle the objectives studied in this paper. This allowed us to prove that \gamesAb{} are in \FPT{} for the parameters described in Table~\ref{table:fpt-summary}. The techniques used for these \FPT{} results in the case of prefix-independent objectives allowed us to provide an improved \FPT{} algorithm for Boolean B\"uchi zero-sum games. We then turned to the complexity class of the \problemAb{} and provided a proof of its \nexptime-membership for all the objectives we study, which relied on showing that any solution to the \problemAb{} can be transformed into a solution with an exponential memory. We then provided a better upper bound for B\"uchi \gamesAb{} by introducing an \np{} algorithm. We provided a proof of the \np-completeness of the problem in the simple setting of reachability \gamesAb{} played on tree arenas. We then came back to regular game arenas and provided the proof of the \nexptime-hardness of the \problemAb{} for all the objectives except for B\"uchi and co-B\"uchi objectives for which we showed the \npHard ness. The \nexptime-hardness proof relied on a reduction from the \succinctSetCoverAb{} which we proved to be \nexptimeComplete{}, a result of potential independent interest. These complexity results are summarized in Table~\ref{table:comp_summary}.

In future work, we want to find the exact complexity class of the \problemAb{} for co-B\"uchi objectives. We also want to study quantitative objectives such as mean-payoff in the framework of \gamesAb{} and the \problemAb{}. It would also be interesting to study whether other works, such as rational synthesis, could benefit from the approaches used in this paper.

\begin{acks}
This work is partially supported by the PDR project Subgame perfection in graph games (F.R.S.-FNRS), the ARC project Non-Zero Sum Game Graphs: Applications to Reactive Synthesis and Beyond (Fédération Wallonie-Bruxelles), the EOS project Verifying Learning Artificial Intelligence Systems (F.R.S.-FNRS and FWO), and the COST Action 16228 GAMENET (European Cooperation in Science and Technology).
\end{acks}

\bibliographystyle{ACM-Reference-Format}
\bibliography{main}
\appendix

\section{Useful Result on SP Games}
\label{app:usefull_notions}

\begin{proposition}
\label{prop:transform_two_suc}
Every \gameAb{} $\mathcal{G}$ with arena $G$ containing $n$ vertices can be transformed into an \gameAb{} $\bar{\mathcal{G}}$ with the same objectives and with arena $\bar{G}$ containing at most $n^2$ vertices such that any vertex in $\bar{G}$ has at most $2$ successors and Player~$0$ has a strategy $\sigma_0$ that is solution to the \problemAb{} in $G$ if and only if Player~$0$ has a strategy $\bar{\sigma}_0$ that is solution to the problem in $\bar{G}$.
\end{proposition}

\begin{proof}
Let $\mathcal{G}$ be an \gameAb{} with arena $G$. Let us first describe the arena $\bar{G}$ of $\bar{\mathcal{G}}$. Let $v \in V$ be a vertex of $G$, then $v$ is also a vertex of $\bar{G}$ such that it belongs to the same player and is the root of a complete binary tree with $\ell = |\{v' \mid (v, v') \in E\}|$ leaves if $(v, v) \not \in E$. Otherwise, $v$ has a self loop and its other successor is the root of such a tree with $\ell-1$ leaves. The internal vertices of the tree (that is vertices which are not $v$, nor the leaves) belong to the same player as $v$. Each leaf vertex $v'$ of this tree is such that $(v, v') \in E$, belongs to the same player as in $G$ and is again the root of its own tree. The initial vertex $v_0$ of $G$ remains unchanged in $\bar{G}$. Since every vertex in $\bar{G}$ is part of a binary tree or has a self loop and a single successor, it holds that it has at most two successors. Since $G$ is a game arena, this transformation is such that each vertex in $\bar{G}$ has at least one successor. It follows that $\bar{G}$ is a game arena containing $n$ vertices $v \in V$ and at most $n - 1$ internal vertices per tree in the case where $v \in V$ has $n$ successors in $G$. It follows that the number of vertices in $\bar{G}$ is at most $n + n \cdot (n - 1) = n^2$. 

We now define the objectives in $\bar{\mathcal{G}}$. In the case of reachability, B\"uchi, co-B\"uchi, Boolean B\"uchi, Streett, and Rabin, the objective $\bar{\Omega}$ in $\bar{\mathcal{G}}$ corresponding to objective $\Omega$ in $\mathcal{G}$ is defined using the same sets of vertices (recall that the vertices of $G$ appear in $\bar{G}$). For parity and Muller \gamesAb{}, the priority function $c'$ of objective $\bar{\Omega}$ remains unchanged for vertices $v \in V$ and we define $c'(v') = c(v)$ for $v' \in \bar{V}\setminus V$ such that $v'$ is an internal vertex of a tree whose root is $v$. In the case of safety objectives, the safe set $\bar{S}$ of objective $\bar{\Omega}$ corresponds to that of $\Omega$, augmented with every newly added vertex of $\bar{G}$, that is $\bar{S} = S \cup (\bar{V} \setminus V)$.

Finally, let us show that there is a solution to the \problemAb{} in $\mathcal{G}$ if and only if there is a solution in $\bar{\mathcal{G}}$. From each root $v$ of a tree in $\bar{G}$ (corresponding to a vertex $v$ of Player~$i$ in $G$) there is a set of $\ell = |\{v' \mid (v, v') \in E\}|$ different paths controlled by  Player~$i$, each leading to a vertex $v'$. It follows that there exists a play $\rho = v_0 v_1 v_2 \ldots \in \Plays_G$ if and only if there exists a play $\rho' = v_0 a_0 \dots a_{n_1} v_1 b_0 \dots b_{n_2} v_2 \ldots \in \Plays_{\bar{G}}$ such that every vertex $a_i$ (resp.\ $b_i$) belongs to the same player as $v_0$ (resp.\ $v_1$) and so on. Given the way the objectives are defined, it holds that $\payoff{\rho} = \payoff{\rho'}$ and $\won{\rho} = \won{\rho'}$. Therefore, a strategy $\sigma_0$ that is solution to the \problemAb{} in $G$ can be transformed into a strategy $\bar{\sigma}_0$ which is a solution in $\bar{G}$ and vice-versa.
\end{proof}

\end{document}